\documentclass[12pt]{article}
\usepackage{amsmath,amsfonts,amssymb, amsthm}
\usepackage{hyperref} % \url, input url
\hypersetup{colorlinks=true, citecolor=blue}
\usepackage{multirow} % various type of table
\usepackage{caption} % figure cation
\usepackage{kotex}
\usepackage{xcolor} % to use alert
\usepackage[sort]{natbib}
\usepackage{fullpage} % use full page
\usepackage{float} % to use "H" in table
\usepackage{chngpage} % to use "adjustwidth"
\usepackage{lscape}
\usepackage{mathtools}
\usepackage[blocks]{authblk}
\usepackage{booktabs} % to use \toprule in table
\usepackage{pbox} % to use line break in table
\usepackage{algorithmic, algorithm}
\usepackage{tikz}
\usepackage{setspace} % double line space
\usepackage{xr} % cross reference

\usetikzlibrary{shapes,decorations,arrows,calc,arrows.meta,fit,positioning}
\tikzset{
	-Latex,auto,node distance =1 cm and 1 cm,semithick,
	state/.style ={ellipse, draw, minimum width = 0.7 cm},
	point/.style = {circle, draw, inner sep=0.04cm,fill,node contents={}},
	bidirected/.style={Latex-Latex,dashed},
	el/.style = {inner sep=2pt, align=left, sloped}
}

\newcommand{\pkg}[1]{\texttt{#1}} % typesetting for R packages

\bibpunct{(}{)}{;}{a}{}{,} %http://merkel.zoneo.net/Latex/natbib.php

\newtheorem{theorem}{Theorem}
\newtheorem*{theorem*}{Theorem}
\newtheorem*{remark*}{Remark}
\newtheorem{assumption}{Assumption}
\newtheorem{corollary}{Corollary}

\newtheorem{lemma}{Lemma}
\newtheorem*{lemma*}{Lemma}

\newtheorem*{fact*}{Fact}

\newcommand{\bA}{\boldsymbol{A}}
\newcommand{\bB}{\boldsymbol{B}}

\newcommand{\bI}{\boldsymbol{\rm I}}

\newcommand{\bQ}{\boldsymbol{Q}}
\newcommand{\bS}{\boldsymbol{S}}
\newcommand{\bR}{\boldsymbol{R}}
\newcommand{\bV}{\boldsymbol{V}}
\newcommand{\bW}{\boldsymbol{W}}

\newcommand{\bzero}{\boldsymbol{0}}

\newcommand{\bGamma}{\boldsymbol{\Gamma}}
\newcommand{\bSigma}{\boldsymbol{\Sigma}}

\newcommand{\bOmega}{\boldsymbol{\Omega}}

\newcommand{\diag} {{{\rm {diag}}}}
\newcommand{\tr} {{{\rm {tr}}}}

\title{Estimating High-dimensional Covariance and Precision Matrices under General Missing Dependence}
\begin{document}

% first authors
\author[1]{Seongoh Park}

% co-authors by alphabetical order
\author[2]{Xinlei Wang}

% corresponding authors
\author[3]{Johan Lim\footnote{To whom all correspondence should be addressed. Email: \texttt{johanlim@snu.ac.kr}}}

% affiliation
\affil[1]{Department of Statistics, 
Sungshin Women's University, Seoul, Korea}
\affil[2]{Department of Statistical Science, Southern Methodist University, Dallas, TX, USA}
\affil[3]{Department of Statistics, Seoul National University, Seoul, Korea} 

\date{} 
\maketitle

%\tableofcontents
\begin{abstract} 
	\noindent 
		A sample covariance matrix $\bS$ of completely observed data is the key statistic in a large variety of multivariate statistical procedures, such as structured covariance/precision matrix estimation, principal component analysis, and testing of equality of mean vectors. However, when the data are partially observed, the sample covariance matrix from the available data is biased and does not provide valid multivariate procedures. To correct the bias, a simple adjustment method called inverse probability weighting (IPW) has been used in previous research, yielding the IPW estimator. The estimator can play the role of $\bS$ in the missing data context, thus replacing $\bS$ in off-the-shelf multivariate procedures such as the graphical lasso algorithm. However, theoretical properties (e.g. concentration) of the IPW estimator have been only established in earlier work under very simple missing structures; every variable of each sample is independently subject to missingness with equal probability. 
	We investigate the deviation of the IPW estimator when observations are partially observed under general missing dependency. We prove the optimal convergence rate $O_p(\sqrt{\log p / n})$ of the IPW estimator based on the element-wise maximum norm, even when two unrealistic assumptions (known mean and/or missing probability) frequently assumed to be known in the past work are relaxed. 	
	The optimal rate is especially crucial in estimating a precision matrix, because of the ``meta-theorem'' \cite{Liu:2012} that claims the rate of the IPW estimator governs that of the resulting precision matrix estimator. In the simulation study, we discuss one of practically important issues, non-positive semi-definiteness of the IPW estimator, and compare the estimator with imputation methods.
%	\textcolor{blue}{(Should we briefly mention the data application?)}
\end{abstract} 
\baselineskip 18pt

\doublespacing
% Main text entry area
\section{Introduction}\label{sec:intro}

%Covariance matrix estimation under missing values is an interest of some fields (fields examples...)
One of the overarching themes in statistical and machine learning societies is to discover complex relationships among high-dimensional variables. Out of many approaches to understand dependency among variables, the covariance matrix and its inverse matrix (i.e., the precision matrix) are arguably important statistical tools in this line of research. Hence, methodological and theoretical analyses of these statistics, such as scalability, consistency, and convergence rate, have been established by many researchers (see the section Introduction from \cite{Fan:2016} for a comprehensive literature review, and references therein), because of their utility in a broad range of disciplines such as biology, geophysics, economics, public health, and social sciences. 
Despite much advance over decades in the estimation of a covariance/precision matrix under the high-dimensional setting, most approaches to date have been oblivious to handling missing observations. However, widespread applications have emerged in modern sciences where the primary interest is placed on estimating the correlation structure from observations subject to missingness. To name a few, climate data (\cite{Schneider:2001}), genomic studies (\cite{Cui:2017,Liang:2018}), and remote sensing data (\cite{Glanz:2018}). Even so, there has been relatively less development in both methodology and theory that deal with the (inverse) covariance estimation problem in the presence of missing data.

\subsection{Existing work on (inverse) covariance matrix estimation with missing values}
Previous research in the field of estimation of an (inverse) covariance matrix with incomplete data, though not many to our best knowledge, can be classified into two branches; the likelihood-based approach and the plug-in approach.

The first line of the works is the likelihood-based inference, mostly achieving the maximum likelihood estimator by an expectation-maximization (EM) algorithm (or its variants) (\cite{Huang:2007,Stadler:2012,Allen:2010,Thai:2014, Liang:2018}). In spite of individual successes in covariance/precision matrix estimation when missing observations are present, the major drawback of this approach is separate development of estimating algorithms and supporting theories. That is, one considering a new proposal under this framework should put huge efforts on implementing the new method for practical purposes and prove theoretical properties (e.g. consistency). Furthermore, the Gaussian assumption on observations commonly used in the likelihood inference could be restrictive in the high-dimensional setting.

The other scheme of research studied rather in recent years is based on the fact that many procedures for estimating a covariance/precision matrix solely rely on the sample covariance matrix $\bS$, not the data itself.
However, if missing observations exist, the sample covariance matrix $\bS_Y$ (see the definition in (\ref{eq:samplecov_Y})) using partial observations is no longer a proper estimator. Thus, preceding work (\cite{Lounici:2014,Kolar:2012,Cai:2016:JMA,Wang:2014,Rao:2017,Pavez:2019}) have considered adjusting the missing proportion, or a bias that appears in $\bS_Y$. The modified estimator is often referred to as an inverse probability weighting (IPW) estimator and is plugged in  procedures of multivariate data analyses instead of $\bS_Y$.
For example, \cite{Kolar:2012} put the IPW estimator into the graphical lasso algorithm (\cite{Friedman:2008}) to estimate a sparse precision matrix, while \cite{Cai:2016:JMA} plug it in banding, tapering, or thresholding operators to recover a structured covariance matrix in the missing data context. \cite{Wang:2014} apply the CLIME method (\cite{Cai:2011:CLIME}) to the bias-corrected rank-based correlation matrix to estimate a sparse precision matrix of a non-paranormal distribution.
In the low-rank approximation problem, the IPW estimator is plugged into the matrix lasso (\cite{Rohde:2011}) by \cite{Lounici:2014}, which is extended by \cite{Rao:2017} to vector autoregressive processes. 
All of these works are based on one common assumption about missingness; for each sample, each variable is independently subject to missingness with equal (uniform) probability. Their theoretical analyses, though recovering the aimed rate $\sqrt{\log p /n}$ ($n$: the sample size, $p$: dimension), are established based on such a restrictive independence assumption.
In contrast, dependent (and non-uniform) missing structure has not been paid attention to until very recent year when \cite{Park:2019} made an initial attempt and \cite{Pavez:2019} made a further investigation. 
While those papers are based on the spectral norm using the effective rank of a matrix (see Table \ref{tab:miscov_review}), this paper derive the optimal convergence rate of the IPW estimator in terms of the element-wise maximum norm under general missing dependency.

\subsection{Our contributions}
Our main contributions are outlined below.

\noindent\textit{Derivation of the optimal convergence rate under dependent missing structure.}\hspace{0.5cm}
We develop a non-asymptotic deviation inequality of the IPW estimator in the element-wise maximum norm by extending missing independence to missing dependency (Theorem \ref{thm:dev_IPW}). The theoretical results maintain the conventional convergence rate $\sqrt{\log p / n}$ achieved by the earlier works (\cite{Bickel:2008a}) and the references in Table \ref{tab:miscov_review}). Theorem \ref{thm:dev_IPW} can be further incorporated with existing theories in estimation of structured precision matrix, because assumptions made in this paper do not conflict with those made in the existing theories.

\noindent\textit{Relaxation of implicit assumptions to derive the rates.}\hspace{0.5cm} In analyzing the concentration of the IPW estimator, estimation of the population mean and missing probability has been largely unexplored (\cite{Lounici:2014}, \cite{Wang:2014}, \cite{Park:2019}, \cite{Pavez:2019}), which is not desirable in practice. Filling the gaps, this paper establishes the concentration inequalities for the IPW estimator under unknown mean (Theorem \ref{thm:dev_IPW_unknown_mean}) and missing probability (Theorem \ref{thm:dev_IPW_unknown_misprob}).

%\noindent\textit{Non-positive semi-definiteness of the IPW estimator.}\hspace{0.5cm}

\subsection{Outline}
The remainder of this paper is organized as follows. At the beginning of Section \ref{sec:IPW_rate}, we formally state the problem setup and introduce the IPW estimator under general missing dependency. 
Under the setting where the missing probability and population means are assumed to be known, we present our theoretical results related to the estimator in Section  \ref{sec:IPW_rate} and apply them to estimation of a sparse precision matrix. 
In Section \ref{sec:relax}, the two unrealistic assumptions are relaxed, and we show similar results to the previous section. 
Section \ref{sec:PSDness} deals with non-positive semi-definiteness of the IPW estimator and its potential remedies.
In Section \ref{sec:simulation} and \ref{sec:realdata}, we show our numerical studies on simulated and real data, respectively. We conclude this paper with a brief discussion and summary in Section \ref{sec:discussion}.

\section{The IPW estimator under general missing dependency and its rate}\label{sec:IPW_rate}
%hinge upon....
Let $X = (X_{1}, \ldots, X_{p})^{\rm T}$ be a $p$-dimensional vector of random variables with mean zero and covariance matrix $\bSigma = \mathbb{E}(X X^{\rm T})$. We denote missing observations by $0$, which has a simple mathematical representation using a missing indicator\footnote{Technically, this is a ``response'' indicator as termed in \cite{Kim:2013}, since the value $1$ indicates an observed (responded) variable, but we insist on using ``missing'' to emphasize the context of missing data.} $\delta_j$ that takes its value either $0$ (missing) or $1$ (observed);
$$
Y=(Y_{1}, \ldots, Y_{p})^{\rm T}, \quad Y_{j} = \delta_j X_j, \quad j=1,\ldots, p.
$$
The multivariate binary vector $\delta = (\delta_1, \ldots, \delta_p)^{\rm T}$ is assumed to follow some distribution where a marginal distribution of $\delta_j$ is the Bernoulli distribution with success probability $0 \le \pi_j \le 1$. This formulation is an extension of independent missing structure used in previous works (\cite{Lounici:2014,Wang:2014,Kolar:2012,Cai:2016:EJS}), which assume $\delta_k$ is independent of $\delta_\ell$ ($k \neq \ell$). Contrary to it, this paper assumes the $p$ random variables $\{\delta_j, j=1, \ldots, p\}$ are allowed to be dependent and not identically distributed. The probability of observing at multiple positions is henceforth denoted by 
$$
{\rm P}(\delta_i=\delta_j=\delta_k=\ldots=1) = \pi_{ijk\ldots}.
$$
Dependent missing structure naturally occurs through a longitudinal clinical study since a patient absent at visit(=variable) $k$ would have more possibility of not showing up at forthcoming visits $\ell(>k)$. There exists more general and plausible scenarios where extrinsic covariates are involved in occurrence of missingness.

% sample version
Let us consider $n$ samples from the population above where the covariance matrix $\bSigma = (\sigma_{k\ell}, 1 \le k,\ell \le p)$ is to be estimated. Denote the $i$-th sample version of $X, Y, \delta_j$ by $X_i, Y_i, \delta_{ij}$, respectively. Then, the sample covariance matrix from partially observed data is obtained by 
\begin{equation}\label{eq:samplecov_Y}
\bS_Y = \dfrac{1}{n} \sum\limits_{i=1}^n Y_i Y_i^{\rm T} = \Big(\dfrac{1}{n} \sum\limits_{i=1}^n \delta_{ij}\delta_{ik}X_{ij}X_{ik}, 1\le j,k \le p\Big).
\end{equation}
It can be easily checked that $\bS_Y$ is biased for $\bSigma$, since its expectation is 
$\bSigma^\pi = \Big(\pi_{jk}\sigma_{jk}, 1\le j,k \le p\Big)$ by assuming independence between $\{X_i\}_{i=1}^n$ and $\{\delta_{ij}\}_{i,j}$. 
This motivates one to adjust each component of $\bS_Y$ by a weight and define the IPW estimator $\widehat{\bSigma}^{IPW} = \Big((\widehat{\bSigma}^{IPW})_{jk}, 1\le j,k \le p\Big)$ by
\begin{equation}\label{eq:IPWest}
(\widehat{\bSigma}^{IPW})_{jk} = 
\left\{ 
\begin{array}{lc}
\dfrac{1}{n} \sum\limits_{i=1}^n\dfrac{\delta_{ij}}{\pi_j}X_{ij}^2 & j=k,\\
\dfrac{1}{n} \sum\limits_{i=1}^n\dfrac{\delta_{ij}\delta_{ik}}{\pi_{jk}}X_{ij}X_{ik} & j\neq k,
\end{array}
\right.
\end{equation}
provided that $\pi_{jk}>0, \forall j,k$. Then, $\widehat{\bSigma}^{IPW}$ is unbiased for $\bSigma$ under the missing completely at random (MCAR) mechanism (\cite{Little:1986}), that is, $\{\delta_{ij}\}_{j=1}^p$ is independent of $\{X_{ij}\}_{j=1}^p$ for $i=1,\ldots, n$. For example, when data acquisition is carried out through censors (e.g. remote sensing data), loss of data arises due to faults in censors and thus is independent of values to be measured.

We note this adjustment technique is frequently used in a general context of missing data and also known as the propensity score method. The underlying idea of it is to construct an unbiased estimating equation by reweighting the contribution of each sample on the equation. The corresponding equation for the covariance estimation problem under the Gaussian setting without missingness is a score function given by
\begin{equation}\label{eq:est_eq}
\dfrac{1}{n}\sum_{i=1}^n Q(X_i; \bSigma) = 0,
\end{equation}
where $Q(X_i; \bSigma)=\bSigma^{-1} X_i X_i^{\rm T} \bSigma^{-1} - \bSigma^{-1}$. Since (\ref{eq:est_eq}) is equivalent to solving $n^{-1} \sum_{i=1}^n ( X_i X_i^{\rm T} - \bSigma)=0$, the reweighted version of the equation above would be
\begin{equation}\label{eq:est_eq_reweight}
\dfrac{1}{n}\sum_{i=1}^n \bR_i * (X_i X_i^{\rm T} - \bSigma) = 0,
\end{equation}
where $\bR_i = \big(\delta_{ij} \delta_{ik} / \pi_{jk}, 1 \le j, k \le p \big)$ and $*$ is an element-wise product. Solving the equation above with respect to $\bSigma$ yields an empirical version of the IPW estimator that replaces $\pi_{jk}$ in (\ref{eq:IPWest}) with $n^{-1} \sum_{i=1}^n \delta_{ij} \delta_{ik}$. This estimator has been used and analyzed before in \cite{Kolar:2012} and \cite{Cai:2016:JMA}, which will be studied in Section \ref{sec:relax_misprob} of this paper under general missing dependency. 
Remark that the inverse probability $\pi_{jk}$ in $\bR_i$ is ignorable and does not play any role in defining the empirical estimator. However, when the probability is dependent on sample-specific variables ($X_i$ or extrinsic covariates $W_i$), we should give weights in the form of the conditional probability defined by  $\text{P}(\delta_{ij}=\delta_{ik} =1|X_i, W_i)$, which adjusts the selection bias from partial observations $\{i:\delta_{ij}=\delta_{ik}=1\}$. For the sake of simplicity, analyses in this paper only concern the identical setting on missing indicators, that is, $\pi_{jk\ell\ldots}\overset{\forall i}{=}\text{P}(\delta_{ij}=\delta_{ik}=\delta_{i\ell}=\ldots =1)$.

\subsection{Notation}
Throughout this paper, we use the following matrix norms; for a matrix $\bA$, the element-wise maximum norm is $||\bA||_{max} = \max_{i,j} |\bA_{ij}|$, 
the operator 1-norm is $||\bA||_1=\max_j \sum_i |\bA_{ij}|$, 
%the operator $\infty$-norm is $||\bA||_\infty = \max_i \sum_j |\bA_{ij}|$, 
the operator 2-norm $||\bA||_2$ is the largest singular value (or eigenvalue if $\bA$ is symmetric),
the element-wise 1-norm is $|\bA|_1=\sum_{i,j} |\bA_{ij}|$, and
the Frobenius norm is $||\bA||_F=\sqrt{\sum_{i,j} \bA_{ij}^2}$.
$\text{diag}(\bA)$ is a diagonal matrix whose diagonal entries are inherited from $\bA$.
For a vector $v$, we define $||v||_1 = \sum_j |v_j|$.
Also, we define $R(\theta) = \exp(1/(4e\theta^2)) - 1/2 - 1/(4e\theta^2)$ for $\theta >0$, which is monotonically decreasing and satisfies $R(\theta)> 1/2$.

\subsection{Main results}
We state our assumptions used in the following theoretical analyses; (i) sub-Gaussianity for each component of $X_i$, (ii) a general dependency structure for $\delta_i$, and (iii) MCAR for missing mechanism.
We begin with one of the equivalent definitions of the sub-Gaussian variable (\cite{Vershynin:2018}): the uniformly bounded moments.
\begin{assumption}[Sub-Gaussianity]\label{assum:subG_moment}
	$X$ is a sub-Gaussian random variable in $\mathbb{R}$ satisfying
	\begin{equation}
	\mathbb{E}X=0, \quad \mathbb{E}X^2=1, \quad  \text{ and  } \sup_{r\ge 1} \dfrac{\big\{ \mathbb{E}|X|^r \big\}^{1/r}}{\sqrt{r}} \le K
	\end{equation}
	for some $K>0$.
\end{assumption}
\noindent
We note that the Gaussian random variable $X \sim N(0,\sigma^2)$ satisfies 
$$
\sup_{r\ge 1} \big\{ \mathbb{E}|X|^r \big\}^{1/r} / \sqrt{r} \le \sigma K
$$
for some numeric constant $K>0$.
Missing is assumed to occur with general dependency in sense of the following;
\begin{assumption}[General missing dependency]\label{assum:general_missing}
	A missing indicator vector $\delta = (\delta_1, \ldots, \delta_p)^{\rm T}\in \{0,1\}^p$ follows some multivariate distribution where each marginal distribution is a Bernoulli distribution with a missing probability\footnote{Following the previous footnote, this is called a ``missing'' probability.} $\pi_j \in (0,1]$, i.e., $\delta_j \sim \text{Ber}(\pi_j)$.
	Further, assume that $\pi_{jk}\neq 0$ for all $1 \le j,k \le p$.
\end{assumption}
\noindent
This distribution is examined by \cite{Dai:2013} where they call it the multivariate Bernoulli distribution. If interaction terms are considered up to the second-order, this multivariate model coincides with a well-known Ising model.
\color{black}
The non-degenerate condition for the missing probabilities (i.e., $\pi_j>0, \pi_{jk}>0$) is required since, for example, $\pi_{jk}=0$ implies no data could be observed for estimating the second moment $\sigma_{jk}$, which is unrealistic for our discussion. Next, we formally state our missing mechanism again;
\begin{assumption}[Missing completely at random]\label{assum:mcar}
	An event that an observation is missing is independent of both observed and unobserved random variables.
\end{assumption}
\noindent
Under the data structure in this paper, the above mechanism essentially says that two random vectors, $\delta_i$ and $X_i$, are independent. 
We note that Assumptions \ref{assum:subG_moment} and \ref{assum:mcar} are commonly used in the context of covariance estimation with incomplete data, while Assumption \ref{assum:general_missing} is more general than the independent structure that previous research depends on.
Based on these assumptions, Lemma \ref{lem:element-wise_ineq} describes the element-wise deviation of the IPW estimator from a true covariance matrix. 
\begin{lemma}\label{lem:element-wise_ineq}
	Let $\{X_{i}\}_{i=1}^n$ be i.i.d. random vectors in $\mathbb{R}^p$ with mean $0$ and covariance $\bSigma$. Suppose the scaled random variable $X_{ik}/\sqrt{\sigma_{kk}}$ satisfies Assumption \ref{assum:subG_moment} with a constant $K>0$ for all $k$. Also, let $\{\delta_i\}_{i=1}^n$ be i.i.d. binary random vectors satisfying Assumption \ref{assum:general_missing}. By observing samples $\{Y_i\}_{i=1}^n$ under Assumption \ref{assum:mcar}, we have
	\begin{equation}\label{eq:element-wise_ineq}
	{\rm P}
	\bigg[
	n^{-1} \Big|
	\sum\limits_{i=1}^n \Big( \dfrac{Y_{ik} Y_{i\ell}}{\pi_{k\ell}} - \sigma_{k\ell} \Big)
	\Big|
	\ge \dfrac{C(\sigma_{kk} \sigma_{\ell\ell})^{1/2} K^2 R(K)^{1/2}}{\pi_{k \ell}} t
	\bigg] \le 4\exp (
	- nt^2),
	\end{equation}
	if $t \ge 0$ satisfies 
	$$
	\begin{cases}
	t^2 \le c R\bigg(\dfrac{2K}{ \sqrt{\pi_k + \pi_\ell -2\pi_{k \ell} |\rho_{k\ell}|}}\bigg), & \text{if } k \neq \ell,\\
	t^2 \le cR\big(K / \sqrt{\pi_k}\big), & \text{if } k = \ell,
	\end{cases}
	$$
	where $c,C>0$ are scalar constants and $\rho_{k\ell}=\sigma_{k\ell}/\sqrt{\sigma_{kk}\sigma_{\ell\ell}}$.
	
\end{lemma}
\noindent
A proof of Lemma \ref{lem:element-wise_ineq} can be found in Section \ref{app:pf_lem} of Appendix. We provide some remarks as regards this lemma.
This concentration inequality covers the existing results as special cases. First, if data is assumed to be fully observed (i.e., $\pi_{k \ell}=1,\forall k, \ell$), then (\ref{eq:element-wise_ineq}) is reduced to 
$$
{\rm P}
\bigg[
n^{-1} \Big|
\sum\limits_{i=1}^n \big(X_{ik} X_{i\ell} - \sigma_{k\ell} \big)
\Big|
\ge C_1 \sqrt{\sigma_{kk}\sigma_{\ell\ell}} ~ t 
\bigg] 
\le 
C_2 \exp(-nt^2), \quad 0 \le t \le C_3,
$$
where $C_1, C_2, C_3$ are scalar constants. It can be seen that this form is equivalent to Lemma A.3. in \cite{Bickel:2008b} (Gaussian) or Lemma 1 in \cite{Ravikumar:2011} (sub-Gaussian), up to multiple constant difference. When an independent and identical structure of missing indicators is assumed (i.e., $\delta_k \overset{\forall k}{\sim} \text{Ber}(\pi)$) in Lemma \ref{lem:element-wise_ineq}, the reduced probabilistic bound is similar to that from \cite{Kolar:2012} (plugging in $t \leftarrow \sqrt{\log(4/\delta)/n}$ in (\ref{eq:element-wise_ineq}))
$$
{\rm P}
\bigg[
n^{-1} \Big|
\sum\limits_{i=1}^n \Big( \dfrac{Y_{ik} Y_{i\ell}}{\pi^2} - \sigma_{k\ell} \Big)
\Big|
\ge \dfrac{CK^2}{\pi^2} \sqrt{\dfrac{R(K)\sigma_{kk}\sigma_{\ell\ell}\log(4/\delta)} {n}}
\bigg] 
\le 
\delta
$$
for the sample size $n$ chosen according to Lemma \ref{lem:element-wise_ineq}.
Rigorously speaking, the proposed IPW estimator in Lemma \ref{lem:element-wise_ineq} and that of \cite{Kolar:2012} (see (\ref{eq:IPW_emp})) are different by the inverse weighting factor when correcting missing observations. However, replacing missing probabilities with unbiased empirical estimates will not cause a considerable change in our result (see Section \ref{sec:relax_misprob}).
%The concentration bound of \cite{Lounici:2014} is presented in the spectral norm (See Table \ref{tab:miscov_review}), and its form is similar to ours.

Using the lemma above, the rate of convergence of the IPW estimator can be derived in terms of the element-wise maximum norm. Let us define the maximum and minimum value of parameters that appear in Lemma \ref{lem:element-wise_ineq} as follows;
$$
%\begin{array}{l}
\sigma_{max} = \max\limits_k \sigma_{kk}, \quad
\pi_{min} = \min\limits_{k, \ell} \pi_{k\ell}, \quad v_{min} = \min\limits_{k \neq \ell}(\pi_k + \pi_\ell -2\pi_{k \ell} |\rho_{k\ell}|).
%\end{array}
$$
%In the third term, three stars is used to denote the order in $\pi_k$; observe that $\min_{k, \ell}\pi_{k \ell}\sqrt{\pi_k \pi_\ell}= \pi^3$ if $\pi_k=\pi$ for all $k$. 

\begin{theorem}\label{thm:dev_IPW}
	Assume the conditions of Lemma \ref{lem:element-wise_ineq} hold, and further assume the sample size and dimension satisfy
	\begin{equation}\label{eq:thm_sample_size}
	n / \log p
	> 
	c \bigg\{
	\exp\Big( \dfrac{v_{min}}{16eK^2} 
	\Big) - \dfrac{1}{2} - \dfrac{v_{min}}{16eK^2}
	\bigg\}^{-1},
	\end{equation}
	then it holds that
	$$
	{\rm P}
	\bigg[
	\big|\big|\widehat{\bSigma}^{IPW} - \bSigma\big|\big|_{max}
	\ge 
	\dfrac{C \sigma_{max}K^2 }{\pi_{min}} \sqrt{\dfrac{R(K)\log p}{n}}
	\bigg]
	\le
	4p^{-1},
	$$
	where $c, C>0$ are scalar constants.
\end{theorem}
\noindent
A proof of the theorem can be found in Section \ref{app:pf_dev_IPW} of Appendix. 
The above result provides a few intuitions. First of all, the convergence rate $\sqrt{\log p  / n}$ is satisfied with the IPW estimator when missing data are present. 
Also, small portion of missingness in data agrees with a faster convergence rate since $\pi_{min} \approx 1$. Furthermore, if we reparametrize the missing probabilities by $\{p_{k\ell}^{(a,b)}: a,b=0,1\}$ where $p_{k\ell}^{(a,b)}=\text{P}(\delta_{ik}=a, \delta_{i\ell}=b)$, then we see that the entries in $v_{min}$ are rewritten by
$$
\pi_k + \pi_\ell -2\pi_{k\ell} |\rho_{k\ell}| = p_{k\ell}^{(1,0)} + p_{k\ell}^{(0,1)} + 2p_{k\ell}^{(1,1)} (1 - |\rho_{12}|), \quad k,\ell=1,\ldots, p.
$$
Thus, if less observations are missing (i.e. larger values of $p_{k\ell}^{(1,0)}, p_{k\ell}^{(0,1)}, p_{k\ell}^{(1, 1)}$), less samples are needed to achieve the same convergence rate.

If we assume an independent structure on missing indicators, we get the following result, which is comparable to those from \cite{Kolar:2012}, \cite{Lounici:2014}, and \cite{Pavez:2018}. Let $\widehat{\bSigma}^{IPW}_{ind}$ be the IPW estimator (\ref{eq:IPWest}) with $\pi_{jk}=\pi^2, j\neq k$ and $\pi_{jj}=\pi$ for all $j,k$.

\begin{corollary}[Identical and independent missing structure]\label{cor:dev_IPW_iid}
	Under the conditions of Lemma \ref{lem:element-wise_ineq}, we further assume $\delta_{ik} \sim \text{Ber}(\pi)$, independently, $k=1,\ldots, p$. Then, when the sample size and dimension satisfy
	$$
	n / \log p
	> c \Big\{
	\exp\Big(\dfrac{\pi(1 - \pi \rho_{max})}{8eK^2}\Big) - \dfrac{1}{2} - \dfrac{\pi(1 - \pi \rho_{max})}{8eK^2}
	\Big\}^{-1},
	$$
	then it holds that
	$$
	{\rm P}
	\bigg[
	\big|\big|\widehat{\bSigma}^{IPW}_{ind} - \bSigma\big|\big|_{max}
	\ge 
	\dfrac{C \sigma_{max}K^2 }{\pi^2} \sqrt{\dfrac{R(K)\log p}{n}}
	\bigg]
	\le
	4p^{-1},
	$$
	where $c, C>0$ are scalar constants and $\rho_{max}=\max\limits_{k\neq \ell}|\rho_{k\ell}|$.
\end{corollary}
\noindent
%\color{blue}
First, we note \cite{Pavez:2018} prove a bound with general element-wise
$q$-norm under non-uniform missing probabilities (i.e.,
$\pi_{jk}=\pi_j \pi_k$, $j\neq k$). By sending $q\to \infty$ in
Theorem 1 of \cite{Pavez:2018}, we could obtain a deviation inequality
in the element-wise maximum norm, which is given in Table
\ref{tab:miscov_review} and comparable to Corollary
\ref{cor:dev_IPW_iid}.

\color{black}
Second, for comparison with other previous research (\cite{Kolar:2012,Lounici:2014,Wang:2014}), we observe that the Taylor expansion of an exponential function yields
$$
\Big\{
\exp\Big(\dfrac{\pi(1 - \pi \rho_{max})}{8eK^2}\Big) - \dfrac{1}{2} - \dfrac{\pi(1 - \pi \rho_{max})}{8eK^2}
\Big\}^{-1} = c_1 / \big(1 + c_2 \pi^2 + o(\pi^2) \big), \quad \text{as } \pi \to 0,
$$
for some $c_1, c_2>0$.
Therefore, the sample size (relative to the dimension) required for accurate estimation is less sensitive to the missing probability $\pi$ compared to the previous work whose magnitude is in order of $1/\pi^2$ (see Table \ref{tab:miscov_review}). However, the bound of the IPW estimator in the element-wise maximum norm increases in the order of magnitude $1/\pi^2$, which is larger than the rate $1/\pi$ claimed in other literature (see Table \ref{tab:miscov_review}).

Finally, Table \ref{tab:miscov_review} summarizes the rate and sample size of the IPW estimator from the related works. \cite{Cai:2016:JMA} have considered the minimax optimality (with a structured covariance matrix), which is, however, not comparable to what is given in Table \ref{tab:miscov_review}. Hence, their work is not included here.
\begin{table}[H]
	\footnotesize
	\begin{adjustwidth}{-2.0cm}{}
		\begin{tabular}{cccccc}
			\toprule
			Article & Est. & Assump. & Norm & Rate & Size\\ 
			\toprule
			%	Covariance Matrix  & & &  & & \\
			%\cmidrule(l){1-1}
			K2012 & $\widehat{\bSigma}^{emp}$ (\ref{eq:IPW_emp}) & Indep & $||\cdot||_{max}$ & $ \sigma_{max}\sqrt{\dfrac{\log (8p)}{\pi^2 n - \sqrt{2\pi^2 n\log (2p)}}}$ & $p=O(\exp(n\pi^2))$ \\
			L2014 & $\widehat{\bSigma}^{IPW}_{ind}$ & Indep & $||\cdot||_2$ & $ \sqrt{\dfrac{\tr(\bSigma)||\bSigma||_2\log p}{\pi^2 n}}$ & $p=O\left(\exp\left(\dfrac{n\pi^2 ||\bSigma||_2}{\tr(\bSigma)}\right)\right)$\\
			PO2018 & $\widehat{\bSigma}^{IPW}_{ind}$ & Indep & $||\cdot||_{max}$ & $ \sigma_{max}\sqrt{\dfrac{\log p}{\min_k \pi_k^4 \ n}}$ & $p=O(\exp(n))$\\	%\hline
			W2014 & Spearman's $\rho$ &Indep &  $||\cdot||_{max}$& $\sqrt{\dfrac{\log p}{\pi^2 n}}$ & $p=O(\exp(n\pi^2))$  \\
			W2014 & Kendall's $\tau$ & Indep &$||\cdot||_{max}$&  $\sqrt{\dfrac{\log p}{\pi^2 n}}$ & $p=O(\exp(n\pi^2))$ \\
			\hline
			PL2019 & $\widehat{\bSigma}^{IPW}$ (\ref{eq:IPWest})& Depen &  $||\cdot||_2$ & $|| M||_2\sqrt{\dfrac{\tr(\bSigma)||\bSigma||_2\log p}{n}}$ & $p=O\left(\exp\left(\Big\{\dfrac{n ||\bSigma||_2}{\tr(\bSigma)}\Big\}^{1/3}\right)\right)$ \\
			PO2019 & $\widehat{\bSigma}^{IPW}$ (\ref{eq:IPWest})& Depen & $\sqrt{\mathbb{E}||\cdot||_2^2}$ & 
			$\sqrt{\dfrac{\tr(\bSigma)||\bSigma||_2\log p}{\pi_{min} n}}$ & $p=O\left(\exp\left(\dfrac{n\pi_{min} ||\bSigma||_2}{(\log n)^2\tr(\bSigma)}\right)\right)$\\
			Theorem \ref{thm:dev_IPW}, \ref{thm:dev_IPW_unknown_misprob} & $\widehat{\bSigma}^{IPW}$ (\ref{eq:IPWest}), $\widehat{\bSigma}^{emp}$ (\ref{eq:IPW_emp})& Depen &  $||\cdot||_{max}$ & $ \sigma_{max}\sqrt{\dfrac{\log p}{ \pi_{min}^2 n}}$ & See (\ref{eq:thm_sample_size}), (\ref{eq:thm_emp_sample_size}) \\
			\toprule
		\end{tabular}
	\end{adjustwidth}
	\caption{Summary of literature using the idea of the IPW estimator. %		$n_{min}^*=\min\limits_{k, \ell}\sum\limits_{i}\delta_{ik}\delta_{i\ell}$,
		$M = (1/\pi_{k \ell}, 1\le k ,\ell \le p)$. $\sigma_{max}=\max\limits_{k}\sigma_{kk}$, 
		$\pi_{min}=\min\limits_{k, \ell} \pi_{k \ell}$.
		``Rate'' is the convergence rate (up to a constant factor depending only on distributional parameters) of an estimator (``Est.'') measured by a matrix norm (``Norm''). 
		``Assump.'' indicates which structure of missing dependency is imposed to derive the rate.
		``Size'' is a condition for $n$ and $p$ to guarantee the rate holds with probability at least $1/p$. 
		In the first column, we use the following labels: L2014=\cite{Lounici:2014},
		%		CZ2016=\cite{Cai:2016:JMA},
		KX2012=\cite{Kolar:2012},
		W2014=\cite{Wang:2014},
		PL2019=\cite{Park:2019},
		PO2018=\cite{Pavez:2018}, and 
		PO2019=\cite{Pavez:2019}.
	}
	\label{tab:miscov_review}
\end{table}

\noindent
Table \ref{tab:miscov_review} shows the rate of convergence $\sqrt{\log p / n}$ has appeared in the previous literature. When dependency for missing indicators is allowed, the achieved rate in \cite{Park:2019} under the spectral norm is not optimal, though they have first tackled it. Very recently, \cite{Pavez:2019} show an improved rate for expectation of an estimation error based on the spectral norm. In terms of the element-wise maximum norm, to our best knowledge, this paper is among the first to obtain the optimal rate.

\subsection{The meta-theorem in estimation of a precision matrix}
The derived concentration inequality is crucial because of its application to precision matrix estimation. The related theory, known as the meta-theorem that has first appeared in \cite{Liu:2012}, implies that the rate of the precision matrix estimator $\widehat{\bOmega}$ is determined by the rate $||\cdot||_{max}$ of an input matrix (e.g. the IPW estimator) used to estimate $\widehat{\bOmega}$. Therefore, when there is no missing observation, the success of the graphical lasso (\cite{Ravikumar:2011}), the CLIME (\cite{Cai:2011:CLIME}), and the graphical Dantzig selector (\cite{Yuan:2010}) in accurate estimation and graph recovery depends on the fact that the sample covariance matrix $\bS$ satisfies 
\begin{equation}\label{eq:rate_sample_cov}
{\rm P}\Big(
\big|\big| \bS - \bSigma\big|\big|_{max} \ge C \sqrt{\log p / n}
\Big) \le d/ p,
\end{equation}
for some $C,d>0$. To grasp the underlying mechanism of the meta-theorem, we refer readers to the proof of Corollary \ref{cor:meta}.
Since the claimed rate of convergence in Theorem \ref{thm:dev_IPW} is the same as that of $\bS$ in (\ref{eq:rate_sample_cov}), the meta-theorem also guarantees the same optimal rates of the precision matrix estimators with missing observations.

It should be remarked that the rate in Theorem \ref{thm:dev_IPW} is not driven for a certain class of covariance/precision matrices (e.g. sparse or low-rank) or with a specific restriction on $n$ and $p$ such as an asymptotic ratio between them, i.e., $p/n \to \alpha \in [0, \infty)$. Such flexibility makes it possible to adopt different conditions (on $\bSigma$, $\bOmega$, $n$, or $p$) required from different precision matrix estimation methods (e.g. the graphical lasso). We describe the meta-theorem under the dependent missing structure below, which is an extension of Theorem 4.3 in \cite{Liu:2012}. A proof of the corollary can be found in Section \ref{app:pf_cor} of Appendix.
%Appendix \ref{app:pf_cor}. 
\begin{corollary}\label{cor:meta}
	Let the true covariance matrix $\bSigma$ satisfy the same assumptions that a precision matrix estimation procedure such as the graphical lasso, the graphical Dantzig selector, and the CLIME requires to guarantee the consistency and the support recovery of a graph.
	
	If we plug the IPW estimator $\widehat{\bSigma}^{IPW}$ into one of the aforementioned methods, the end product retrieves the optimal rate of convergence, and thus has consistency and support recovery properties\footnote{The support recovery is not guaranteed with the graphical Dantzig selector, since its rate is achieved in the matrix $\ell_1$-norm, not $||\cdot||_{max}$.} even under general missing dependency.
\end{corollary}
%\noindent

\section{Relaxation of implicit assumptions}\label{sec:relax}
Estimation using the IPW estimator with missing data depends on two implicit assumptions other than Assumptions \ref{assum:subG_moment}, \ref{assum:general_missing}, and \ref{assum:mcar}; known mean (or equivalently zero mean) and missing probabilities. In this section, we will relax such conditions and show corresponding concentration results.

\subsection{The case of unknown mean}\label{sec:relax_mean}
%\color{blue}
When the first moment of random variables is unknown, an estimator should be modified to be unbiased. It requires a small trick since we do not directly estimate the mean parameter $\mu_k$, but $\mu_k\mu_\ell$ because of the dependent missing structure. The resulting estimator (\ref{eq:IPWest_unknown_mean}) holds the same rate $O_p(\sqrt{\log p/n})$ (Theorem \ref{thm:dev_IPW_unknown_mean}). It should be pointed out that the extension to the unknown mean case is not so obvious because a quadratic form of sub-Gaussian variables has newly appeared (the second part of (\ref{eq:IPWest_unknown_mean})). We deal with it by proving a new version of Hanson-Wright inequality described in Lemma \ref{lem:new_HW_ineq}.
\color{black}

Assume that we observe $\tilde{Y}_{ik} = \delta_{ik} \tilde{X}_{ik}$ where $\tilde{X}_{ik}$ has an unknown mean $\mu_k$. Adopting previous notations, we define $X_{ik}$ to satisfy $\tilde{X}_{ik}=X_{ik} + \mu_k$.
Then, it is easy to show that 
$$
\mathbb{E}\big[\sum\limits_{i=1}^n \tilde{Y}_{ik} \tilde{Y}_{i\ell} \big] = n \pi_{k \ell} (\sigma_{k\ell} + \mu_k \mu_\ell), \quad \mathbb{E}\big[\sum\limits_{i\neq j}^n \tilde{Y}_{ik} \tilde{Y}_{j\ell} \big] = n(n-1) \pi_k\pi_{\ell} \mu_k \mu_\ell.
$$
With a simple calculation, we can define the unbiased covariance matrix estimator by $\widehat{\bSigma}^{IPW\mu} = 
\Big((\widehat{\bSigma}^{IPW\mu})_{k\ell}, 1\le k, \ell \le p\Big)$ with 
\begin{equation}\label{eq:IPWest_unknown_mean}
(\widehat{\bSigma}^{IPW\mu})_{k\ell} = \dfrac{\sum_{i=1}^n \tilde{Y}_{ik} \tilde{Y}_{i\ell}}{n \pi_{k \ell}} - \dfrac{\sum_{i\neq j}^n \tilde{Y}_{ik} \tilde{Y}_{j\ell}}{n(n-1) \pi_k \pi_{\ell}}.
\end{equation}
It is not difficult to find resemblance of (\ref{eq:IPWest_unknown_mean}) with the sample covariance matrix $\bS$ when data is completely observed. The $(k,\ell)$-th component of $\bS$ is defined by
$$
\bS_{k\ell} = \dfrac{1}{n-1}\sum\limits_{i=1}^n (\tilde{X}_{ik} - \hat{\mu}_k)(\tilde{X}_{i\ell} - \hat{\mu}_\ell),
$$
where $\hat{\mu}_k= n^{-1}\sum_{i=1}^n \tilde{X}_{ik}$, and it can be rearranged by
$$
\bS_{k\ell} = \dfrac{\sum_{i=1}^n \tilde{X}_{ik} \tilde{X}_{i\ell}}{n} - \dfrac{\sum_{i\neq j}^n \tilde{X}_{ik} \tilde{X}_{j\ell}}{n(n-1)},
$$
which is equal to (\ref{eq:IPWest_unknown_mean}) when $\pi_{k \ell}=\pi_{k}=1$ for all $k,\ell$.
The following theorem shows the concentration of $\widehat{\bSigma}^{IPW\mu}$.
\color{black}
\begin{theorem}\label{thm:dev_IPW_unknown_mean}
	Assume the conditions of Lemma \ref{lem:element-wise_ineq} hold except a mean zero condition, and further assume the sample size and dimension satisfy
	$$
	n / \log p
	> c\max \left\{
	\dfrac{1}{R\big(2K/\sqrt{v_{min}}\big)}, 
	~ \dfrac{K^2}{\pi_{min,d} + 2e^2K^2}
	\right\},
	$$
	then it holds that
	$$
	{\rm P}
	\bigg[
	\big|\big|\widehat{\bSigma}^{IPW\mu} - \bSigma\big|\big|_{max}
	\ge 
	C
	\sqrt{\dfrac{\log p }{n}}
	\bigg]
	\le
	dp^{-1},
	$$
	where $c>0, d>0$ are scalar constants and $C>0$ is a scalar constant depending only on $K$, $\sigma_{max}$, $\max_k |\mu_k|$, $\pi_{min}$, and $\min\limits_k \pi_k$.
\end{theorem}
\noindent
A proof of Theorem \ref{thm:dev_IPW_unknown_mean} can be found in Section \ref{app:pf_dev_IPW_unknown_mean} of Appendix, where we introduce a new version of Hanson-Wright inequality by extending the result of \cite{Rudelson:2013}. While the previous research deals with a quadratic form $X^{\rm T} A X$ of a sub-Gaussian vector $X$, we handle $X_1^{\rm T} A X_2$ of sub-Gaussian vectors $X_1$ and $X_2$ whose pair $(X_{1i}, X_{2i})$ may be dependent with each other.
%\color{blue}
We later notice that \cite{Zhou:2019} considers missing indicators in the quadratic form $(X * \delta)^{\rm T} A (X * \delta)$. However, the result from \cite{Zhou:2019} cannot be used to derive the rate in covariance estimation. It only enables one to calculate the convergence rate for diagonal components (i.e., marginal variances), since off-diagonal components are of form $(X_1 * \delta_1)^{\rm T} A (X_2 * \delta_2)$ for distinct $X_1 \neq X_2$ and $\delta_1 \neq \delta_2$ where $*$ is the element-wise product of two vectors.

%which can be also extended in the same manner $(X_1 * \delta_1)^{\rm T} A (X_2 * \delta_2)$, but using her result does not bring meaningful difference in the obtained rates.
\color{black}

%\color{blue}
\begin{remark*}
	In the theorem above, dependency of the constant $C$ on the parameters can be specified by, (up to a constant factor)
	$$
	C = \dfrac{\max\{\sigma_{max}, \mu_{max}, \mu_{max}^2 \} \max \Big\{K^2\sqrt{R(K)}, \sqrt{1+2e^2K^2} \Big\} } {\min\{\pi_{min}^{3/2}, \pi_{min,d}^{2}\}}
	$$
	where $\mu_{max}=\max_k |\mu_k|$ and $\pi_{min,d}=\min\limits_k \pi_k$. Supposedly, dependency on the mean parameter $\mu_{max}$ can be taken away in $C$ if a missing value is filled by the empirical mean of available data. However, we leave this as future work.
\end{remark*}

\subsection{The case of unknown missing probability }\label{sec:relax_misprob}
In real applications, the missing probability $\pi_{jk}$ is rarely known and needs to be estimated. Let $\hat{\pi}_{jk}$ be any estimate satisfying $\hat{\pi}_{jk}>0, \forall j,k$, with high probability. Then, the resulting IPW estimator is presented by
\begin{equation}\label{eq:IPWest_unknown_misprob}
\widehat{\bSigma}^{IPW\pi} = \Big((\widehat{\bSigma}^{IPW})_{jk} \dfrac{\pi_{jk}}{\hat{\pi}_{jk}}, ~ 1\le j, k \le p \Big),
\end{equation}
provided that the population mean is known for the sake of simplicity. When additional information on missingness is not available for estimating $\pi_{jk}$, it is natural to use the empirical proportions
$\hat{\pi}_{jk}^{emp}= n^{-1}\sum_{i=1}^n \delta_{ij} \delta_{ik}$ of observed samples since it is asymptotically unbiased for $\pi_{jk}$ (by the law of large numbers).
We denote the empirical version $\widehat{\bSigma}^{emp}$ of the IPW estimator by 
\begin{equation}\label{eq:IPW_emp}
(\widehat{\bSigma}^{emp})_{jk} = \dfrac{\sum\limits_{i=1}^n\delta_{ij}\delta_{ik} X_{ij}X_{ik}}{\sum\limits_{i=1}^n \delta_{ij}\delta_{ik}}, \quad 1 \le j, k \le p,
\end{equation}
which corresponds to (\ref{eq:IPWest_unknown_misprob}) with $\hat{\pi}_{jk}^{emp}$ in place of $\hat{\pi}_{jk}$. One may realize the equivalence of the empirical estimate (\ref{eq:IPW_emp}) to a pairwise complete analysis. 
Based on Lemma \ref{lem:dev_invprop} that describes the concentration for the inverse probability of $\hat{\pi}_{jk}^{emp}$, we can derive the concentration inequality of $\widehat{\bSigma}^{emp}$.
\begin{theorem}\label{thm:dev_IPW_unknown_misprob}
	Assume the conditions of Lemma \ref{lem:element-wise_ineq} without knowing missing probabilities hold, and further assume the sample size and dimension satisfy
	\begin{equation}\label{eq:thm_emp_sample_size}
	n / \log p
	> c\max \left\{
	\dfrac{1}{R\big(2K/\sqrt{v_{min}}\big)}, 
	~ \dfrac{1}{\pi_{min}}
	\right\},
	\end{equation}
	then it holds that
	$$
	{\rm P}
	\bigg[
	\big|\big|\widehat{\bSigma}^{emp} - \bSigma\big|\big|_{max}
	\ge 
	\dfrac{C \sigma_{max} \max \big\{K^2\sqrt{R(K)}, 1 \big\} }{\pi_{min}}
	\sqrt{\dfrac{\log p}{n}}
	\bigg]
	\le
	dp^{-1},
	$$
	where $c>0, d>0$ are scalar constants and $C>0$ is a scalar constant depending only on $K$, $\sigma_{max}$, and $\pi_{min}$.
\end{theorem}
\noindent
A proof of the theorem can be found in Section \ref{app:pf_dev_IPW_unknown_misprob} of Appendix.
This result has an implication that the convergence rate $\sqrt{\log p / n}$ in Theorem \ref{thm:dev_IPW} is preserved, and thus the same statements in Theorem \ref{cor:meta} hold true with $\widehat{\bSigma}^{emp}$. It should be pointed out that \cite{Kolar:2012} use the estimator $\widehat{\bSigma}^{emp}$, while their theory is limited to the independent missing structure. Thus, Theorem \ref{thm:dev_IPW_unknown_misprob} generalizes the theory for the empirical IPW estimator to the dependent structure.
%\begin{remark*}
%	In the theorem above, dependency of the constant $C$ on the parameters can be specified by, (up to a constant factor)
%	$$
%	C = \dfrac{ \sigma_{max} \max \big\{K^2\sqrt{R(K)}, 1 \big\} }{\pi_{min}}.
%	$$	
%\end{remark*}

%\color{blue}
\section{Non-positive semi-definiteness of the plug-in estimator}\label{sec:PSDness}
Despite its straightforward derivation and applicability to multivariate procedures in the presence of missing data, the IPW estimator has one critical issue from a practical point of view; non-positive semi-definiteness (non-PSDness). Note that this does not cause problems in the convergence rate, since the norm is element-wisely defined. 
It is well known that the element-wise product of two matrices may not preserve a nice property of the matrices. As addressed in high-dimensional covariance estimation (thresholding, banding, and tapering) (\cite{Bickel:2008a,Rothman:2009}), the positive semi-definiteness is one of the typical examples to be broken down by the Hadamard product of a positive semi-definite (PSD) matrix and a general matrix. This is also the case for the IPW estimator, which makes it practically difficult to use the IPW estimator when using existing algorithms for estimating a precision matrix. 
For instance, we can plug the IPW estimator into the graphical lasso or the CLIME to estimate a sparse precision matrix $\bOmega = (\omega_{k\ell}, 1 \le k, \ell \le p)$, when missing data occur. However, the popularly used algorithms (\pkg{glasso} package or \pkg{clime} package in R) require the plugged-in estimator to be positive semi-definite. 
In this section, we examine the graphical lasso algorithm from this point of view and also suggest possible solutions. A similar discussion about the CLIME can be found in Section \ref{app:nonPSD_clime} of Appendix.
%In this section, we examine their algorithms from this point of view and also suggest possible solutions.

\subsubsection*{Graphical lasso}
In what follows, we distinguish between a plug-in matrix (estimator) $\widehat{\bSigma}^{plug}$ and an initial matrix (estimator) $\bSigma^{(0)}$ (or $\bOmega^{(0)}$) that is used to initialize iterative steps.

The graphical lasso proposed by \cite{Friedman:2008} aims to maximize the penalized likelihood function
\begin{equation}\label{eq:glasso}
\max\limits_{\bOmega \succeq 0} \Big\{\log |\bOmega| - \tr(\bOmega \widehat{\bSigma}^{plug}) - \lambda \sum\limits_{k, \ell} |\omega_{k\ell}| \Big\},
\end{equation}	
for a penalty parameter $\lambda > 0$. To solve (\ref{eq:glasso}), a coordinate descent algorithm described in Algorithm \ref{alg:glasso} is proposed by \cite{Friedman:2008} and implemented in R package \pkg{glasso}.
\begin{algorithm}[H]
	\caption{The coordinate descent algorithm for the graphical lasso}
	\label{alg:glasso}
	\begin{algorithmic}[1]
		\REQUIRE An initial matrix $\bSigma^{(0)}$ of $\bSigma$, the plug-in matrix $\widehat{\bSigma}^{plug}$
		\FOR{$i = 1, 2, \ldots,$}
		\FOR{$j = 1, \ldots, p,$}
		\STATE Solve the least squared regression with the $\ell_1$-penalty
		\begin{equation}\label{eq:lasso}
		\hat{\beta}_j = \arg\min_{\beta \in \mathbb{R}^{p-1}} \dfrac{1}{2} \beta^{\rm T} \bSigma^{(i-1)}_{\backslash j \backslash j} \beta - \beta^{\rm T} \widehat{\bSigma}_j^{plug}  + \lambda ||\beta||_1,
		\end{equation}
		where $\bSigma^{(i-1)}_{\backslash j \backslash j}$ is obtained by removing the $j$-th row and column in $\bSigma^{(i-1)}$ and $\widehat{\bSigma}^{plug}_j$ is the $j$-th column of $\widehat{\bSigma}^{plug}$ without the $j$-th entry.
		\STATE Replace the $j$-th column and row of off-diagonal entries in $\bSigma^{(i-1)}$ with $\bSigma^{(i-1)}_{\backslash j\backslash j}\hat{\beta}_j$.
		\ENDFOR % end for j
		\STATE Let $\bSigma^{(i)} \leftarrow \bSigma^{(i-1)}$.
		\ENDFOR % end for i
		\STATE Let $\bSigma^{(\infty)}$ and $\{\hat{\beta}_1,\ldots, \hat{\beta}_p\}$ be the final outputs from lines 1-7.
		\FOR{$j = 1, \ldots, p$}
		\STATE $\widehat{\bOmega}_{jj}= \big(\bSigma^{(\infty)}_{jj} -\hat{\beta}_j^{\rm T} \bSigma^{(\infty)}_j \big)^{-1}$ and $\widehat{\bOmega}_j=- \widehat{\bOmega}_{jj}\hat{\beta}_j$.
		\ENDFOR % end for j
		\ENSURE $\widehat{\bOmega}$: the final estimate.
	\end{algorithmic}
\end{algorithm}
\noindent
One can easily see that the optimization problem (\ref{eq:glasso}) is convex regardless of  $\widehat{\bSigma}^{plug}$ ($\because$ the trace term is a linear function in $\bOmega$), but PSDness of $\widehat{\bSigma}^{plug}$ is needed when the algorithm is initialized.

First, PDness of ${\bSigma}^{(i-1)}$ is required in (\ref{eq:lasso}) to find a well-defined solution of the lasso problem. Since PD $\bSigma^{(i-1)}$ guarantees the updated matrix $\bSigma^{(i)}$ to be PD (\cite{Banerjee:2008}), the PD initial ${\bSigma}^{(0)}$ is necessary to make sure every step runs successfully. 
However, currently available R packages (e.g. \pkg{glasso} version 1.10 from \cite{Friedman:2008} or \pkg{huge} version 1.3.2 from \cite{Zhao:2012}) set $\bSigma^{(0)} \leftarrow \widehat{\bSigma}^{plug} +\lambda \bI$ where $\lambda$ is the same parameter used in (\ref{eq:glasso}). 
%Moreover, the diagonal entries are intact through iterations in Algorithm \ref{alg:glasso}.
As a consequence, unless $\lambda$ is bigger than the absolute value of the smallest (possibly negative) eigenvalue of $\widehat{\bSigma}^{IPW}$, the coordinate descent algorithm would fail to converge. 
% is not suitable for $\widehat{\bSigma}^{plug}$ in such algorithms 
For this reason, we propose to use the following inputs 
\begin{equation}\label{eq:glasso_init}
\widehat{\bSigma}^{plug} \leftarrow \widehat{\bSigma}^{IPW}, \quad \bSigma^{(0)} \leftarrow \text{diag}\big(\widehat{\bSigma}^{IPW} + \lambda \bI\big).
\end{equation}
The above proposal for the initial matrix is made because diagonals of the solution $\bSigma^{(\infty)}$ should satisfy 
$
\bSigma^{(\infty)}_{ii} = \widehat{\bSigma}^{plug}_{ii} + \lambda, \forall i,
$
by the subgradient condition of (\ref{eq:glasso}), as noted in \cite{Friedman:2008}, and because diagonals of $\bSigma^{(i)}$ do not change as iterations proceed.
To use these proposed inputs, one should modify the off-the-shelf code (e.g. \pkg{glasso} function in \pkg{glasso} package) since it does not currently allow users to control $\bSigma^{(0)}$ and $\widehat{\bSigma}^{plug}$ individually.

Last but not least, it should be remarked that there is an algorithm developed to solve (\ref{eq:glasso}) by approximating the Hessian function (R package \pkg{QUIC} from \cite{Hsieh:2014}). This method does not suffer from the PSDness issue discussed here, which is verified through a numerical experiment given in Section \ref{sec:simul_fail} of Appendix.
However, solving the similar issue in the other multivariate procedures remains open.

\subsubsection*{More general solution: matrix approximation}
Previously, we present the solutions that are specific to the precision matrix estimation problem, but we can circumvent the non-PSD issue for general statistical procedures. The idea is to approximate $\widehat{\bSigma}^{plug}$ by the nearest PSD matrix, which can be achieved by
\begin{equation}\label{eq:PSDification}
\widehat{\bSigma}^{psd}= \arg\min\limits_{\bSigma \succeq 0} d(\bSigma, \widehat{\bSigma}^{plug})
\end{equation}
where $d$ measures the distance between two matrices. For instance, the Frobenius norm (\cite{Wang:2014,Katayama:2018}) and the element-wise maximum norm (\cite{Loh:2018}) are used previously. Then, the nearest matrix $\widehat{\bSigma}^{psd}$ would be put into the subsequent multivariate analyses (e.g. the graphical lasso) without modification in the current implementations. However, solving the problem (\ref{eq:PSDification}) comes at the price of such convenience.

When the Frobenius norm is used,  (\ref{eq:PSDification}) amounts to a well-known projection onto the convex cone of PSD matrices. The solution denoted by $\widehat{\bSigma}^{psd}_F$ can be explicitly expressed by
$$
\widehat{\bSigma}^{psd}_F = \bV \bW_+ \bV^{\rm T}, \quad \bW_+ = \max(\bW, \bzero)
$$
where $\widehat{\bSigma}^{plug}$ has the spectral decomposition $\bV \bW \bV^{\rm T}$ and the maximum between two matrices operates element-wisely. The computational cost for this case is mostly from the eigenvalue decomposition, but the convergence rates derived for the IPW estimator in terms of the element-wise maximum norm (e.g. Theorem \ref{thm:dev_IPW}) are not guaranteed for $\widehat{\bSigma}^{psd}_F$.

In contrast, when $d$ is the element-wise maximum norm (\cite{Loh:2018}), the convergence rate is preserved for the solution $\widehat{\bSigma}^{psd}_M$ since
$$
||\widehat{\bSigma}^{psd}_M - \bSigma||_{max} \le ||\widehat{\bSigma}^{psd}_M -  \widehat{\bSigma}^{plug}||_{max} + ||\widehat{\bSigma}^{plug}- \bSigma||_{max} \le 2 ||\widehat{\bSigma}^{plug}-  \bSigma||_{max}
$$
where the first inequality uses the triangular inequality and the second is from the definition of $\widehat{\bSigma}^{psd}_M$. 
The algorithm to solve (\ref{eq:PSDification}) with the element-wise maximum norm is first proposed by \cite{Xu:2012} and used in the robust covariance estimation context (\cite{Loh:2018,Han:2014}).
We note, however, by experience that the approximation based on $||\cdot||_{max}$ is computationally heavy so that it often dominates the computation time of multivariate procedures (e.g. the graphical lasso and the CLIME).
%\color{blue}
On the other hands, \cite{Greenewald:2017, Datta:2017} use the alternating direction method of multipliers (ADMM) to solve (\ref{eq:PSDification}) with the element-wise maximum norm.

\color{black}
\section{Numerical study}\label{sec:simulation}
%\color{red}
%As noted in \cite{Fan:2019}, we can compare the likelihood-based method with nodewise regression methods. 
%Loh and Wainwright (2012)
%and Rudelson and Zhou (2017) both study the statistical
%and computational convergence properties of using
%errors-in-variables regression to handle indefinite
%input matrices in high-dimensional settings.
%\color{black}

In this section, we perform a number of simulations for estimating a covariance/precision matrix with partially observed data. First, in Section \ref{sec:simul_asymp}, we experimentally check the convergence rate of the IPW estimator given in our theorems. In Section \ref{sec:simul_impute}, we conduct a comparison study between several imputation methods and the IPW method. Performance of the estimates is also measured and compared according to simulation parameters, and the related results can be found in Section \ref{app:simul_res_extra} of Appendix.

\subsection{Setting}

\subsubsection*{Data generation}
We generate Gaussian random vectors $X_i$, $i=1,\ldots, n$, in $\mathbb{R}^p$ with mean vector $0$ and precision matrix $\bOmega = (\omega_{ij}, 1 \le i, j \le p)$ under different pairs of $n=50, 100, 200$ and $p$ satisfying $r(=p/n)=0.2, 1, 2$. 
We consider three types of precision matrix as follows, which have been used in the previous literature (\cite{Cai:2011:CLIME,Loh:2012}): chain, star, random graphs.
Two structures (independent, dependent) are under consideration to impose missingness on data where the missing proportion is set to $0\%, 15\%, 30\%$. More precise definitions of true precision matrices and missing structures are given in Section \ref{app:simul_setting} of Appendix.

\subsubsection*{Estimators}
We compare two types of plug-in estimator: $\widehat{\bSigma}^{IPW}$, an oracle type estimator labeled by ``orc'' and $\widehat{\bSigma}^{emp}$, an empirical type estimator labeled by ``emp''.
A closed form of the weight $\pi_{k \ell}$ is accessible according to each missing structure, so the oracle IPW estimator is explicitly computable. It is noteworthy that the estimator $\widehat{\bSigma}^{emp}$ is used in \cite{Kolar:2012}, but their theoretical analysis is limited to the independent missing structure.

We exploit \pkg{QUIC} algorithm proposed by \cite{Hsieh:2014} to solve the graphical lasso (\ref{eq:glasso}). The grid of a tuning parameter $\lambda\in \Lambda$ is defined adaptively to the plug-in matrix $\widehat{\bSigma}^{plug}$. 

\subsection{The rate of convergence}\label{sec:simul_asymp}
We verify our theoretical results (Theorem \ref{thm:dev_IPW} and \ref{thm:dev_IPW_unknown_misprob}) by computing the element-wise maximum deviation $||\widehat{\bSigma}^{plug} - \bSigma||_{max}$. We fix $p=100$ and vary the sample size in $20 \le n \le 10000$. 
We repeat each scenario $20$ times and plot the log-transformed empirical distance against $\log(n/p)$. Different plug-in estimators (``orc'', ``emp'') and precision matrices (chain, star, random) are under consideration.

Figure \ref{fig:asymp_rate} shows that each graph connecting the averaged distances nearly forms a straight line. The results in the column ``orc'' confirm the rate of convergence in Theorem \ref{thm:dev_IPW}, while those in the column ``emp'' confirms that in Theorem \ref{thm:dev_IPW_unknown_misprob}.
\begin{figure}[H]
	\centering
	\includegraphics[page=1,width=0.8\linewidth]{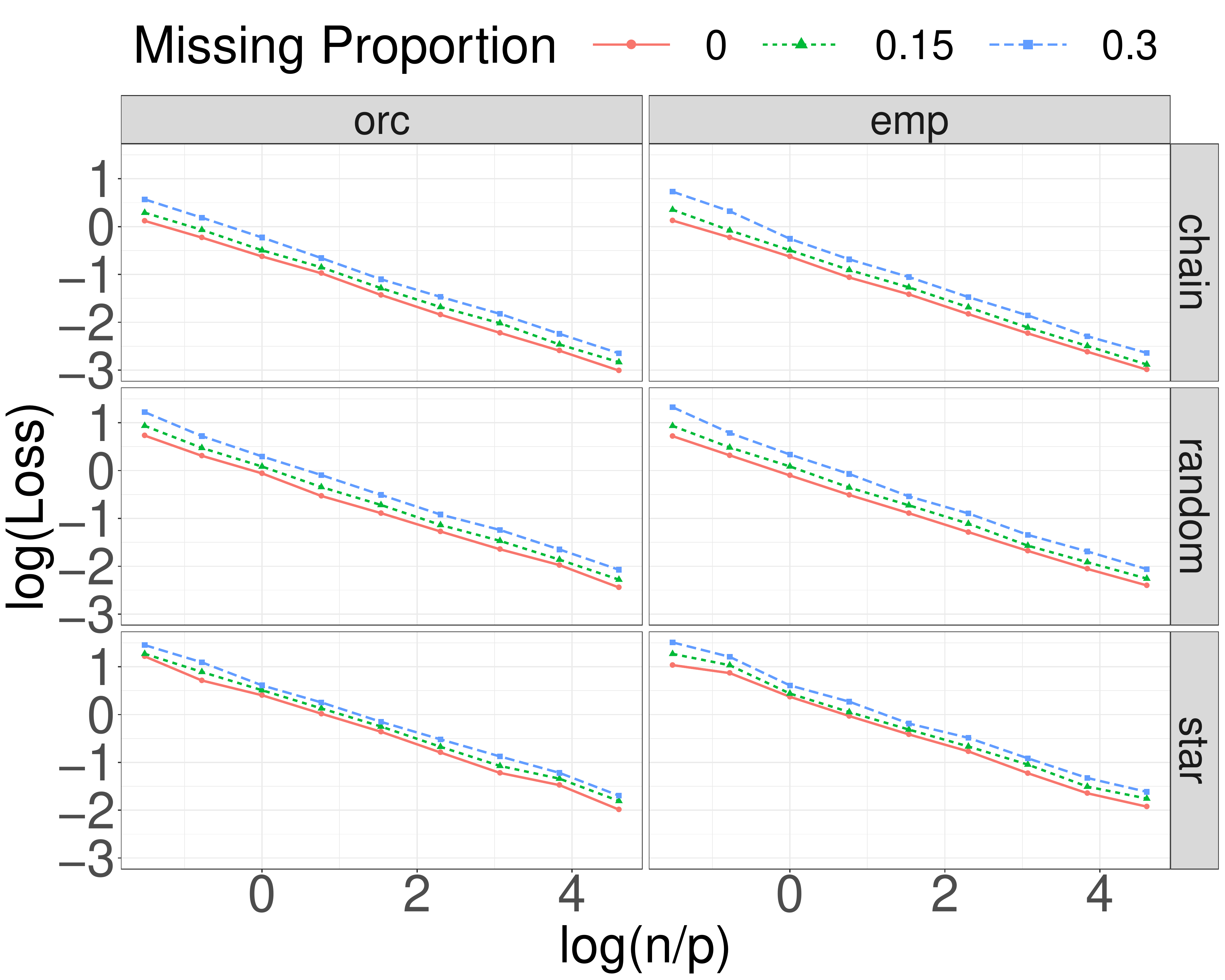}
	\caption{Convergence rate of the plug-in matrix (``orc''=$\widehat{\bSigma}^{IPW}$, ``emp''=$\widehat{\bSigma}^{emp}$) against $\log(n/p)$. Loss is computed by the element-wise maximum norm between the plug-in matrix and the true covariance matrix. The dependent missing structure and $p=100$ are assumed. Each dot (or mark) is an average loss from $20$ repetitions.}
	\label{fig:asymp_rate}
\end{figure}

\subsection{Comparison with imputation methods}\label{sec:simul_impute}
%\color{red}
In the missing data context, unobserved data is often substituted by some function of observed values. One very intuitive way to do it is the imputation method. Once the pseudo complete data is produced, we perform a usual statistical analysis. In this experiment, we compare different (single) imputation approaches with the IPW estimator for the precision matrix estimation.

Imputation methods we use are ``median'' (a median of available data for each variable), ``pmm'' (predictive mean matching from R package \pkg{Hmisc} (\cite{Hmisc:2019})), ``knn'' (an average of k-nearest neighbors from R package \pkg{impute} (\cite{impute:2018})), ``cart'', ``rf'', and ``norm'' (regression-based methods from R package \pkg{mice} (\cite{mice:2011})). We use the default parameter setting for each R function. More details of each method can be found in each reference.

By fixing $n=100$ and $r=1,2$, we randomly generate $100$ data sets based on different precision matrices. Missing observations are produced under the independent structure. 
Once missing observations are filled by a single imputation method, then we compute the sample covariance matrix with the imputed complete data and carry out the precision matrix estimation using the QUIC algorithm. We compare the competing methods based on support recovery of the estimated precision matrix. Figure \ref{fig:impute} shows the pAUC values, where the IPW method using the empirical estimator (``emp'') achieves the largest pAUC compared to the imputation approaches. This is more distinct when the dimension is larger than the sample size (i.e., $r=2$).
The results demonstrate that the IPW method is not only theoretically solid, but also practically useful. Admittedly, we have not thoroughly examined more diverse and complex imputation methods that may produce better performance, which calls for extensive numerical studies in the future.

\begin{figure}[H]	
	\centering
	\includegraphics[width=0.8\linewidth]{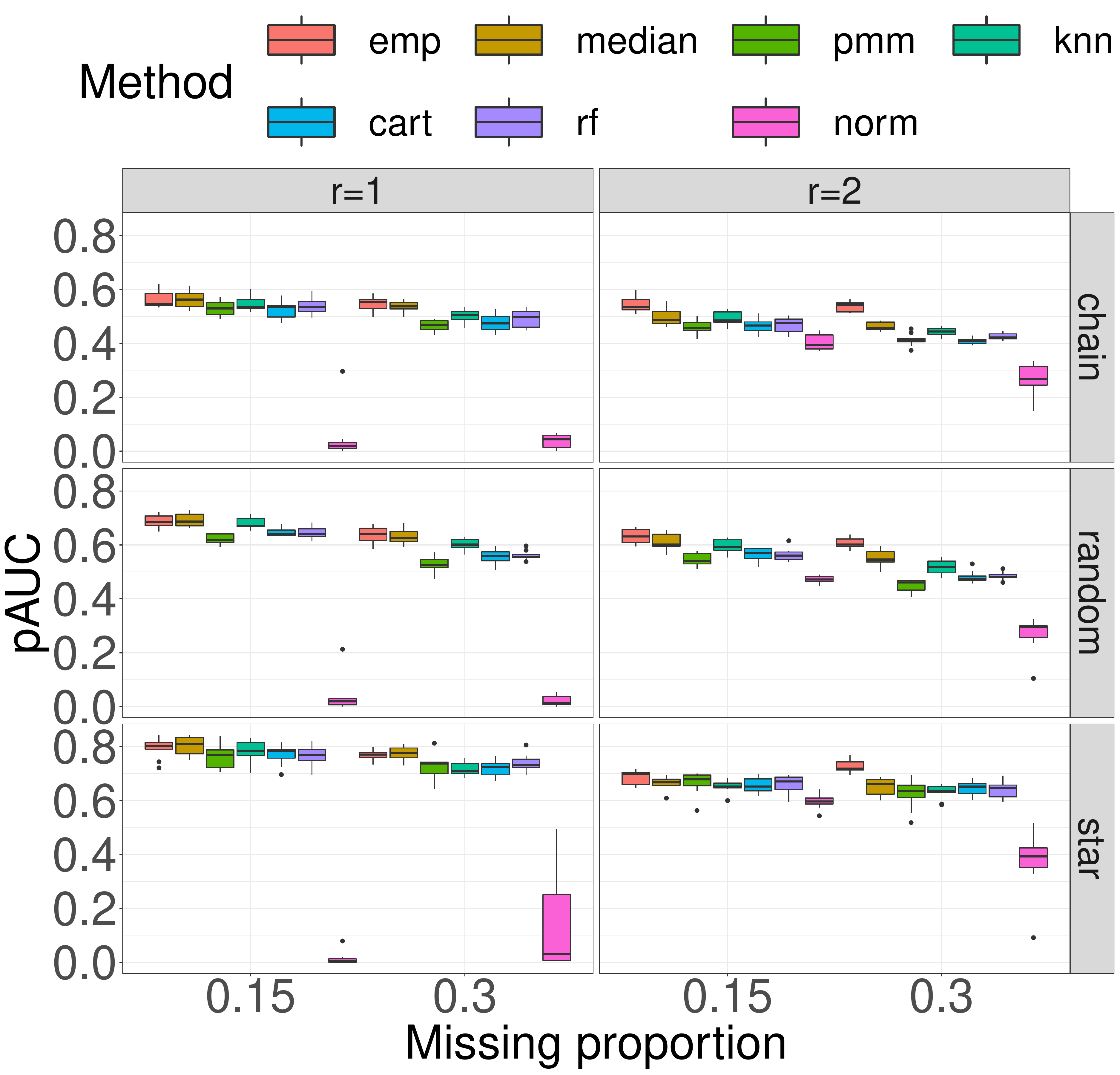}
	%		\end{minipage}
	\caption{Comparison of the pAUC values for different approaches to handle missingness in estimating a sparse precision matrix. Here, $r=1, 2$, $n=100$, and the independent missing structure are assumed. The empirical IPW estimator is plugged-in. We randomly generate $10$ data sets.}
	\label{fig:impute}
\end{figure}

\subsection{Unknown mean}\label{sec:simul_unknown_mean}
In this last experiment, a mean vector is no longer assumed to be zero and thus one needs to consider the corresponding estimator $\widehat{\bSigma}^{IPW\mu}$ given in (\ref{eq:IPWest_unknown_mean}). The goal of this simulation is to check how the estimation error of (\ref{eq:IPWest_unknown_mean}) changes according to the mean vector. 
We consider two types of structures in it: ``full'' ($\mu=k (1,\ldots, 1)^{\rm T}$) and ``sparse'' ($\mu=k(1,\ldots, 1, 0, \ldots, 0)^{\rm T}$). The sparse vector has zeros in the last half of its components. The magnitude $k>0$ is set by the size condition $||\mu||_2=1, 2, 4, 8$. In Figure \ref{fig:dist_unknown_mean}, the case when size of mean is 0 indicates when the mean vector is known, so it is included as a control group. The proportion of missing data varies over 15, 30\%.

\begin{figure}[H]	
	\centering
	%		\begin{minipage}{0.45\linewidth}
	\includegraphics[width=1\linewidth]{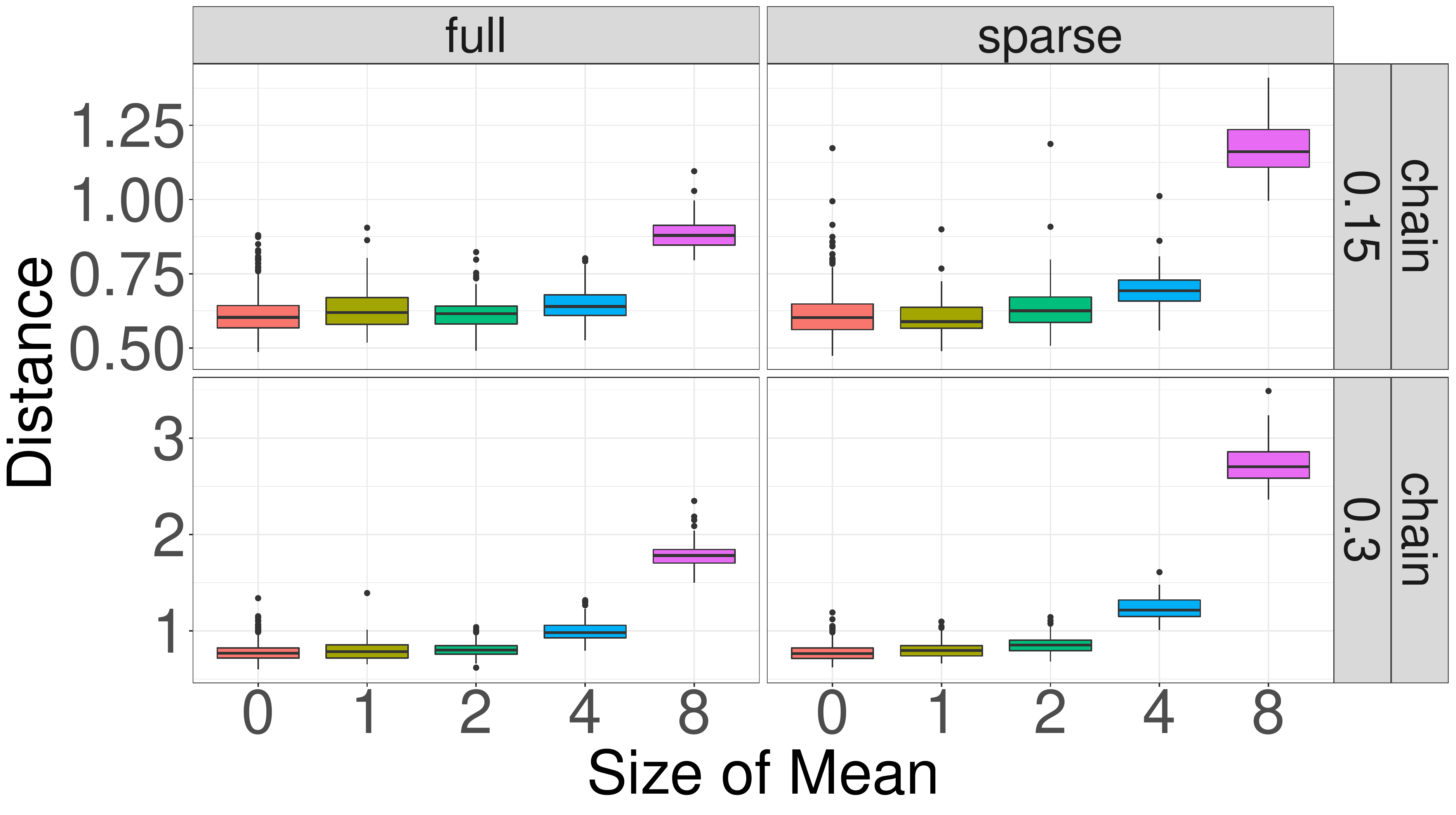}
	%		\end{minipage}
	\caption{Element-wise maximum norm $||\widehat{\bSigma}^{IPW\mu} - \bSigma||_{\max}$ according to a magnitude of (unknown) mean. Here, $n=100$, $r=1$, the chain structure in $\Omega$, and the dependent missing structure are assumed. We randomly generate $100$ data sets.}
	\label{fig:dist_unknown_mean}
\end{figure}

We measure the estimation error by the element-wise maximum norm $||\widehat{\bSigma}^{IPW\mu} - \bSigma||_{\max}$ to verify the bound in Theorem \ref{thm:dev_IPW_unknown_mean}.
The estimation error does not increase by replacing the true means with
their empirical estimates if the size of the true mean ($=\max \{ \mu_{max}, \mu_{max}^2\}$) is not large and less than some cut-off value ($=\sigma_{\max}$). However, if the size of the true mean is large, the estimation error starts to increase.
A larger estimation error in the sparse structure of mean compared to the full can be also explained. Given the same 2-norm $||\mu||_2$, the sparse vector has larger $\mu_{\max}$ than the full vector due to the sparsity, which results in more deviation.

The above observations are also valid in terms of the estimation of a precision matrix where $\widehat{\bSigma}^{IPW\mu}$ is plugged-in. The related results are given in Appendix \ref{app:simul_res_extra_unknown_mean}

\section{Application to real data}\label{sec:realdata}
%\color{red}
%이런 종류의 언급이 필요함
%Missing values in this data correspond to votes that
%are missed by senators and consist of roughly 2.6\% of
%total votes. Note that only 109 of the votes are fully
%observed, so some type of correction or imputation
%should be used instead of omitting rows.
%
%
%We can also access to the data used in \cite{Fan:2019} at https://github.com/unitedstates/congress, though there are some data pre-processing steps.
%\color{black}

We examine the estimation performance of the IPW estimator through a real data application. We use the riboflavin data available from the R package \pkg{hdi}, where 4088 gene expressions are observed across 71 samples. Since the ground truth precision matrix is not known, we construct it by solving the graphical lasso (\ref{eq:glasso}) at a fixed $\lambda$ with a complete data. We impose missing values in a similar manner described in Section \ref{sec:simulation}. Throughout this analysis, it is confirmed again that having more missing values yields worse estimation. Also, it is possible to see that the denser model that has a precision matrix with more non-zero elements is more difficult to achieve satisfactory accuracy in estimation and graph recovery. More details and results can be found in Section \ref{app:realdata} of Appendix.

\color{black}
\section{Discussion}\label{sec:discussion}
%\color{blue}
% % Summary of our work
This paper considers a theoretical establishment of the IPW estimator with missing observations. Contrary to the previous literature, this is achieved under dependency among missingness, meaning that missing indicators are not necessarily independent across variables. The rate of convergence of the IPW estimator is derived based on the element-wise maximum norm, which is (asymptotically) in the same order of the rate claimed in the past works. Our analysis can be applied to an estimation of a sparse precision matrix. Due to the meta-theorem, the favorable properties (consistency, support recovery) of the final estimator are preserved in the missing data context. 
%We also tackle feasibility of the proposed estimator and propose legitimate remedies for it.

% % plug-in procedures
The plug-in estimators (e.g. the sample covariance matrix and the IPW estimator) and their concentration are often not of primary interest, but the ultimate goal lies in applying them to downstream procedures (e.g. Hotelling's T$^2$, a portfolio optimization, etc). In the portfolio optimization, \cite{Fan:2012} show that the risk inequality is bounded by the error of the plug-in estimator $\widehat{\bSigma}^{plug}$;
$$
|w^{\rm T} \widehat{\bSigma}^{plug} w - w^{\rm T} \bSigma w | \le || \widehat{\bSigma}^{plug} - \bSigma||_{max}.
$$
Here, $w$ and $\bSigma$ are true (or optimal) parameters.
However, it is still elusive how the rate $||\hat{w} - w||_V$ for the optimal solution $\hat{w}$ that minimizes the risk $t \mapsto t^{\rm T} \widehat{\bSigma}^{plug} t$ is linked to the rate $|| \widehat{\bSigma}^{plug} - \bSigma||_{M}$ of the plug-in estimator. $||\cdot||_V$ and $||\cdot||_M$ are some norms of a vector and a matrix, respectively.
This line of research could be interesting for future work and in urgent need, not to mention its extension to the missing data context.

% % Missing mechanism
%\subsection{The case beyond missing completely at random}\label{sec:MAR}
The underlying assumptions on the missing mechanism (i.e., MCAR) and the missing structure (i.e., identical dependency across samples) are essentially not verifiable, but it is natural to think of extending our results to the cases beyond such patterns. 
For example, missing at random (MAR) mechanism assumes that missingness occurs independently of unobserved random variables given observed variables. Then, it is easy to show the corresponding IPW estimator is still unbiased. However, it is not straightforward to follow the analyses given in this paper to this case. It would be interesting to identify suitable assumptions that are less stronger than MCAR, but still guarantee the missing probability to be free from $X_{i,obs}$.

\newpage
\bibliographystyle{apalike}
\bibliography{References}

\begin{thebibliography}{}

\bibitem[Allen and Tibshirani, 2010]{Allen:2010}
Allen, G.~I. and Tibshirani, R. (2010).
\newblock Transposable regularized covariance models with an application to
  missing data imputation.
\newblock {\em The Annals of Applied Statistics}, 4(2):764--790.

\bibitem[Banerjee et~al., 2008]{Banerjee:2008}
Banerjee, O., El~Ghaoui, L., and d'Aspremont, A. (2008).
\newblock Model selection through sparse maximum likelihood estimation for
  multivariate gaussian or binary data.
\newblock {\em Journal of Machine Learning Research}, 9:485--516.

\bibitem[Bickel and Levina, 2008a]{Bickel:2008a}
Bickel, P.~J. and Levina, E. (2008a).
\newblock Covariance regularization by thresholding.
\newblock {\em The Annals of Statistics}, 36(6):2577--2604.

\bibitem[Bickel and Levina, 2008b]{Bickel:2008b}
Bickel, P.~J. and Levina, E. (2008b).
\newblock Regularized estimation of large covariance matrices.
\newblock {\em The Annals of Statistics}, 36(1):199--227.

\bibitem[Boucheron et~al., 2016]{Boucheron:2016}
Boucheron, S., Lugosi, G., and Pascal, M. (2016).
\newblock {\em Concentration inequalities: a nonasymptotic theory of
  independence}.
\newblock Oxford University Press.

\bibitem[Cai et~al., 2011]{Cai:2011:CLIME}
Cai, T., Liu, W., and Luo, X. (2011).
\newblock A constrained $\ell_1$ minimization approach to sparse precision
  matrix estimation.
\newblock {\em Journal of the American Statistical Association},
  106(494):594--607.

\bibitem[Cai et~al., 2016]{Cai:2016:EJS}
Cai, T.~T., Ren, Z., and Zhou, H.~H. (2016).
\newblock Estimating structured high-dimensional covariance and precision
  matrices: Optimal rates and adaptive estimation.
\newblock {\em Electronic Journal of Statistics}, 10(1):1--59.

\bibitem[Cai and Zhang, 2016]{Cai:2016:JMA}
Cai, T.~T. and Zhang, A. (2016).
\newblock Minimax rate-optimal estimation of high-dimensional covariance
  matrices with incomplete data.
\newblock {\em Journal of Multivariate Analysis}, 150:55--74.

\bibitem[Cui et~al., 2017]{Cui:2017}
Cui, R., Groot, P., and Heskes, T. (2017).
\newblock Robust estimation of gaussian copula causal structure from mixed data
  with missing values.
\newblock In {\em 2017 IEEE International Conference on Data Mining (ICDM)},
  pages 835--840.

\bibitem[Dai et~al., 2013]{Dai:2013}
Dai, B., Ding, S., and Wahba, G. (2013).
\newblock Multivariate bernoulli distribution.
\newblock {\em Bernoulli}, 19(4):1465--1483.

\bibitem[Datta and Zou, 2017]{Datta:2017}
Datta, A. and Zou, H. (2017).
\newblock Cocolasso for high-dimensional error-in-variables regression.
\newblock {\em Ann. Statist.}, 45(6):2400--2426.

\bibitem[Fan et~al., 2016]{Fan:2016}
Fan, J., Liao, Y., and Liu, H. (2016).
\newblock An overview of the estimation of large covariance and precision
  matrices.
\newblock {\em The Econometrics Journal}, 19(1):C1--C32.

\bibitem[Fan et~al., 2012]{Fan:2012}
Fan, J., Zhang, J., and Yu, K. (2012).
\newblock Vast portfolio selection with gross-exposure constraints.
\newblock {\em Journal of the American Statistical Association},
  107(498):592--606.

\bibitem[Friedman et~al., 2008]{Friedman:2008}
Friedman, J., Hastie, T., and Tibshirani, R. (2008).
\newblock Sparse inverse covariance estimation with the graphical lasso.
\newblock {\em Biostatistics}, 9(3):432--441.

\bibitem[Glanz and Carvalho, 2018]{Glanz:2018}
Glanz, H. and Carvalho, L. (2018).
\newblock An expectation–maximization algorithm for the matrix normal
  distribution with an application in remote sensing.
\newblock {\em Journal of Multivariate Analysis}, 167:31--48.

\bibitem[Greenewald et~al., 2017]{Greenewald:2017}
Greenewald, K., Park, S., Zhou, S., and Giessing, A. (2017).
\newblock Time-dependent spatially varying graphical models, with application
  to brain fmri data analysis.
\newblock In Guyon, I., Luxburg, U.~V., Bengio, S., Wallach, H., Fergus, R.,
  Vishwanathan, S., and Garnett, R., editors, {\em Advances in Neural
  Information Processing Systems 30}, pages 5832--5840. Curran Associates, Inc.

\bibitem[Han et~al., 2014]{Han:2014}
Han, F., Lu, J., and Liu, H. (2014).
\newblock Robust scatter matrix estimation for high dimensional distributions
  with heavy tails.
\newblock {\em Technical report, Princeton University}.

\bibitem[{Harrell Jr} et~al., 2019]{Hmisc:2019}
{Harrell Jr}, F.~E., with contributions~from Charles~Dupont, and many others.
  (2019).
\newblock {\em Hmisc: Harrell Miscellaneous}.
\newblock R package version 4.2-0.

\bibitem[Hastie et~al., 2018]{impute:2018}
Hastie, T., Tibshirani, R., Narasimhan, B., and Chu, G. (2018).
\newblock {\em impute: impute: Imputation for microarray data}.
\newblock R package version 1.56.0.

\bibitem[Hsieh et~al., 2014]{Hsieh:2014}
Hsieh, C.-J., Sustik, M.~A., Dhillon, I.~S., and Ravikumar, P. (2014).
\newblock Quic: Quadratic approximation for sparse inverse covariance
  estimation.
\newblock {\em Journal of Machine Learning Research}, 15:2911--2947.

\bibitem[Huang et~al., 2007]{Huang:2007}
Huang, J.~Z., Liu, L., and Liu, N. (2007).
\newblock Estimation of large covariance matrices of longitudinal data with
  basis function approximations.
\newblock {\em Journal of Computational and Graphical Statistics},
  16(1):189--209.

\bibitem[Katayama et~al., 2018]{Katayama:2018}
Katayama, S., Fujisawa, H., and Drton, M. (2018).
\newblock Robust and sparse gaussian graphical modelling under cell-wise
  contamination.
\newblock {\em Stat}, 7(1):e181.

\bibitem[Kim and Shao, 2013]{Kim:2013}
Kim, J.~K. and Shao, J. (2013).
\newblock {\em Statistical Methods for Handling Incomplete Data}, page~10.
\newblock Chapman and Hall, 1 edition.

\bibitem[Kolar and Xing, 2012]{Kolar:2012}
Kolar, M. and Xing, E.~P. (2012).
\newblock Estimating sparse precision matrices from data with missing values.
\newblock In {\em Proceedings of the 29th International Coference on
  International Conference on Machine Learning}, ICML'12, pages 635--642, USA.
  Omnipress.

\bibitem[Liang et~al., 2018]{Liang:2018}
Liang, F., Jia, B., Xue, J., Li, Q., and Luo, Y. (2018).
\newblock An imputation–regularized optimization algorithm for high
  dimensional missing data problems and beyond.
\newblock {\em Journal of the Royal Statistical Society: Series B (Statistical
  Methodology)}, 80(5):899--926.

\bibitem[Little and Rubin, 1986]{Little:1986}
Little, R. J.~A. and Rubin, D.~B. (1986).
\newblock {\em Statistical Analysis with Missing Data}.
\newblock John Wiley \& Sons, Inc., New York, NY, USA.

\bibitem[Liu et~al., 2012]{Liu:2012}
Liu, H., Han, F., Yuan, M., Lafferty, J., and Wasserman, L. (2012).
\newblock High-dimensional semiparametric gaussian copula graphical models.
\newblock {\em The Annals of Statistics}, 40(4):2293--2326.

\bibitem[Loh and Tan, 2018]{Loh:2018}
Loh, P.-L. and Tan, X.~L. (2018).
\newblock High-dimensional robust precision matrix estimation: Cellwise
  corruption under $\epsilon $-contamination.
\newblock {\em Electronic Journal of Statistics}, 12(1):1429--1467.

\bibitem[Loh and Wainwright, 2012]{Loh:2012}
Loh, P.-L. and Wainwright, M.~J. (2012).
\newblock High-dimensional regression with noisy and missing data: Provable
  guarantees with nonconvexity.
\newblock {\em The Annals of Statistics}, 40(3):1637--1664.

\bibitem[Lounici, 2014]{Lounici:2014}
Lounici, K. (2014).
\newblock High-dimensional covariance matrix estimation with missing
  observations.
\newblock {\em Bernoulli}, 20(3):1029--1058.

\bibitem[Pang et~al., 2014]{Pang:2014}
Pang, H., Liu, H., and Vanderbei, R. (2014).
\newblock The fastclime package for linear programming and large-scale
  precision matrix estimation in r.
\newblock {\em Journal of Machine Learning Research}, 15(1):489--493.

\bibitem[Park and Lim, 2019]{Park:2019}
Park, S. and Lim, J. (2019).
\newblock Non-asymptotic rate for high-dimensional covariance estimation with
  non-independent missing observations.
\newblock {\em Statistics \& Probability Letters}, 153:113--123.

\bibitem[{Pavez} and {Ortega}, 2018]{Pavez:2018}
{Pavez}, E. and {Ortega}, A. (2018).
\newblock Active covariance estimation by random sub-sampling of variables.
\newblock In {\em 2018 IEEE International Conference on Acoustics, Speech and
  Signal Processing (ICASSP)}, pages 4034--4038.

\bibitem[Pavez and Ortega, 2019]{Pavez:2019}
Pavez, E. and Ortega, A. (2019).
\newblock Covariance matrix estimation with non uniform and data dependent
  missing observations.

\bibitem[Rao et~al., 2017]{Rao:2017}
Rao, M., Javidi, T., Eldar, Y.~C., and Goldsmith, A. (2017).
\newblock Estimation in autoregressive processes with partial observations.
\newblock In {\em 2017 IEEE International Conference on Acoustics, Speech and
  Signal Processing (ICASSP)}, pages 4212--4216.

\bibitem[Ravikumar et~al., 2011]{Ravikumar:2011}
Ravikumar, P., Wainwright, M.~J., Raskutti, G., and Yu, B. (2011).
\newblock High-dimensional covariance estimation by minimizing
  $\ell_1$-penalized log-determinant divergence.
\newblock {\em Electronic Journal of Statistics}, 5:935--980.

\bibitem[Rohde and Tsybakov, 2011]{Rohde:2011}
Rohde, A. and Tsybakov, A.~B. (2011).
\newblock Estimation of high-dimensional low-rank matrices.
\newblock {\em The Annals of Statistics}, 39(2):887--930.

\bibitem[Rothman et~al., 2009]{Rothman:2009}
Rothman, A.~J., Levina, E., and Zhu, J. (2009).
\newblock Generalized thresholding of large covariance matrices.
\newblock {\em Journal of the American Statistical Association},
  104(485):177--186.

\bibitem[Rudelson and Vershynin, 2013]{Rudelson:2013}
Rudelson, M. and Vershynin, R. (2013).
\newblock Hanson-wright inequality and sub-gaussian concentration.
\newblock {\em Electronic Communications in Probability}, 18:9 pp.

\bibitem[Saulis and Statulevi\v{c}ius, 1991]{Saulis:1991}
Saulis, L. and Statulevi\v{c}ius, V. (1991).
\newblock {\em Limit theorems for large deviations}.
\newblock Springer Science Business Media.

\bibitem[Schneider, 2001]{Schneider:2001}
Schneider, T. (2001).
\newblock Analysis of incomplete climate data: Estimation of mean values and
  covariance matrices and imputation of missing values.
\newblock {\em Journal of Climate}, 14(5):853--871.

\bibitem[St{\"a}dler and B{\"u}hlmann, 2012]{Stadler:2012}
St{\"a}dler, N. and B{\"u}hlmann, P. (2012).
\newblock Missing values: sparse inverse covariance estimation
  and an extension to sparse regression.
\newblock {\em Statistics and Computing}, 22(1):219--235.

\bibitem[Thai et~al., 2014]{Thai:2014}
Thai, J., Hunter, T., Akametalu, A.~K., Tomlin, C.~J., and Bayen, A.~M. (2014).
\newblock Inverse covariance estimation from data with missing values using the
  concave-convex procedure.
\newblock In {\em 53rd IEEE Conference on Decision and Control}, pages
  5736--5742.

\bibitem[{van Buuren} and Groothuis-Oudshoorn, 2011]{mice:2011}
{van Buuren}, S. and Groothuis-Oudshoorn, K. (2011).
\newblock {mice}: Multivariate imputation by chained equations in r.
\newblock {\em Journal of Statistical Software}, 45(3):1--67.

\bibitem[Vershynin, 2018]{Vershynin:2018}
Vershynin, R. (2018).
\newblock {\em High-Dimensional Probability: An Introduction with Applications
  in Data Science}, page 11–37.
\newblock Cambridge Series in Statistical and Probabilistic Mathematics.
  Cambridge University Press.

\bibitem[Walter, 2005]{Walter:2005}
Walter, S.~D. (2005).
\newblock The partial area under the summary roc curve.
\newblock {\em Statistics in Medicine}, 24(13):2025--2040.

\bibitem[Wang et~al., 2014]{Wang:2014}
Wang, H., Fazayeli, F., Chatterjee, S., and Banerjee, A. (2014).
\newblock {Gaussian Copula Precision Estimation with Missing Values}.
\newblock In Kaski, S. and Corander, J., editors, {\em Proceedings of the
  Seventeenth International Conference on Artificial Intelligence and
  Statistics}, volume~33 of {\em Proceedings of Machine Learning Research},
  pages 978--986, Reykjavik, Iceland. PMLR.

\bibitem[Xu and Shao, 2012]{Xu:2012}
Xu, M.~H. and Shao, H. (2012).
\newblock Solving the matrix nearness problem in the maximum norm by applying a
  projection and contraction method.
\newblock {\em Advances in Operations Research}, 2012:1--15.

\bibitem[Yuan, 2010]{Yuan:2010}
Yuan, M. (2010).
\newblock High dimensional inverse covariance matrix estimation via linear
  programming.
\newblock {\em Journal of Machine Learning Research}, 11:2261--2286.

\bibitem[Zhao et~al., 2012]{Zhao:2012}
Zhao, T., Liu, H., Roeder, K., Lafferty, J., and Wasserman, L. (2012).
\newblock The huge package for high-dimensional undirected graph estimation in
  r.
\newblock {\em Journal of Machine Learning Research}, 13:1059--1062.

\bibitem[Zhou, 2019]{Zhou:2019}
Zhou, S. (2019).
\newblock Sparse hanson–wright inequalities for subgaussian quadratic forms.
\newblock {\em Bernoulli}, 25(3):1603--1639.

\end{thebibliography}

\newpage
\section*{\Large
	Supplementary Material for ``Estimating High-dimensional Covariance and Precision Matrices under General Missing Dependence''}
\vspace{1cm}
\appendix

\section{Auxiliary lemmas}
The first supporting lemma tells a tail bound of a variable with a cumulant generating function dominated by a quadratic function.
\begin{lemma}[Theorem 3.2 and Lemma 2.4 in \cite{Saulis:1991}]\label{lem:asym_key}
	Let a random variable $\xi_j$ with $\mathbb{E}\xi_j = 0, {\rm Var}(\xi_j) = \sigma_j^2$ satisfy the following; there exist positive constants $A, C, c_1, c_2, \ldots$, such that
	\begin{equation}\label{eq:cond_cumulant}
	\Big|\log \mathbb{E} \exp\{\lambda \xi_j \} \Big| \le c_j^2 \lambda^2, \quad |\lambda| < A, \quad \forall j,
	\end{equation}
	and
	$$
	\varlimsup_{n\to \infty} \sum\limits_{j=1}^n c_j^2 \Big/\sum\limits_{j=1}^n \sigma_j^2  \le C.
	$$
	Then, we have for $\xi=\sum\limits_{j=1}^n \xi_j \Big/\sqrt{\sum\limits_{j=1}^n \sigma_j^2}$, 
	$$
	{\rm P}\big[\pm \xi \ge x\big] \le 
	\exp\big(
	-{x^2}/{8C}
	\big), \quad 0 \le x \le 2 A C \sqrt{\sum\limits_{j=1}^n \sigma_j^2}.
	$$
	Furthermore, if $\xi_i$'s are identically distributed and satisfying the conditions above, then the variance term $\sigma_j^2$ does not appear in the concentration inequality:
	$$
	{\rm P}\bigg[\pm \sum\limits_{j=1}^n \xi_j \ge x \bigg] \le 
	\exp\bigg\{
	-\dfrac{x^2}{8nc_1^2}
	\bigg\}, \quad 0 \le x \le 2 A n c_1^2.
	$$
\end{lemma}

\noindent
The following auxiliary results for a sub-Gaussian variable $X$ facilitate one to check the condition (\ref{eq:cond_cumulant}) in Lemma \ref{lem:subG_con_ineq}.
%The following auxiliary results for a sub-Gaussian variable implies that the squared sub-Gaussian variable belongs to a sub-exponential distribution.
\begin{lemma}\label{lem:subG_mgf}
	Assume that $X$ is a random variables satisfying Assumption \ref{assum:subG_moment} for some $K>0$. 
	Then, it holds 
	\begin{enumerate}
		\item[(a)] for $|t| \le (2eK)^{-1}$,
		$$
		\mathbb{E} \exp (t X) \le \exp \big\{(1/2 + K^2 e^2)t^2 \big\},
		$$
		
		\item[(b)] and for $|t|<1/(2\kappa)$,
		\begin{equation}\nonumber %\label{eq:subG_square_upper}
		\mathbb{E} \Big[\exp\big\{t(X^2-1)\big\} \Big] \le \exp(c_0 t^2),
		\end{equation}
		where $\kappa = 4eK^2$ and $c_0=2\kappa^2 \{\exp(1/\kappa) - 1/2 - 1/\kappa\}$.
	\end{enumerate}
\end{lemma}
\begin{proof}
	
	We first prove (a). 
	For $t\in \mathbb{R}$, observe that
	$$
	\begin{array}{rcl}
	\mathbb{E} \exp(tX)& = & 1+ \dfrac{t^2}{2} + \sum\limits_{r \ge 3} \dfrac{\mathbb{E} X^r t^r}{r!}\\
	& \le & 1+ \dfrac{t^2}{2} + \sum\limits_{r \ge 3} \dfrac{\mathbb{E}|X|^r t^r}{r!}\\
	& \le & 1+ \dfrac{t^2}{2} + \sum\limits_{r \ge 3} \dfrac{K^r r^r |t|^r}{r!}\\
	& \le & 1+ \dfrac{t^2}{2} + \sum\limits_{r \ge 3} K^r e^r |t|^r ~(\because (r/e)^r \le r!)\\
	& \le & 1+ \dfrac{t^2}{2} + \dfrac{(K e |t|)^3}{1 - K e |t|}, \quad \text{if } |t| \le (K e)^{-1}.
	\end{array}
	$$
	Then, it holds for any $0<t_0< (K  e)^{-1}$ that for all $|t|<t_0$,
	$$
	\mathbb{E} \exp(tX) \le 1 + |t|^2(1/2 + K^2 e^2) \le \exp \big\{|t|^2(1/2 + K^2 e^2) \big\},
	$$
	which concludes the proof of (a).

	Next, we prove (b).
	Using the Minkowski inequality, we have
	$$
	\big(\mathbb{E}|X^2-1|^r\big)^{1/r} \le \big(\mathbb{E}|X|^{2r}\big)^{1/r} + 1 \le 2r K^2 + 1,
	$$
	which thus gives the upper bound of moments of $X^2-1$,
	$$
	\mathbb{E}|X^2-1|^r \le (2r K^2 + 1)^r \le 2^{r-1}(2^rr^r K^{2r} + 1).
	$$
	Therefore, 
	$$
	\begin{array}{rcl}
	\mathbb{E} \Big[\exp\big\{t(X^2-1)\big\} \Big] &= &
	1 + \sum\limits_{r\ge 2} \dfrac{t^r \mathbb{E}(X^2-1)^r}{r!}\\
	& \le & 1 + \sum\limits_{r\ge 2} \dfrac{|t|^r 2^{r-1}(2^r r^r K^{2r} + 1)}{r!}\\
	& \le & 1 + \dfrac{1}{2} \sum\limits_{r\ge 2} \Big\{\dfrac{(4|t|rK^2)^r}{r!} + \dfrac{(2|t|)^{r}}{r!}\Big\}\\
	& \le & 1 + \dfrac{1}{2} \sum\limits_{r\ge 2} \Big\{(4|t|eK^2)^r + \dfrac{(2|t|)^{r}}{r!}\Big\}\\
	& = & 1 + \dfrac{|t|^2}{2} \sum\limits_{r\ge 2} \Big\{(4eK^2)^2(4|t|eK^2)^{r-2} + \dfrac{4(2|t|)^{r-2}}{r!}\Big\}\\
	\end{array}
	$$
	where the last inequality is derived from $(n/e)^n \le n!$ for $n\ge 1$. Then, it holds for any $0<t_0< 1/(4eK^2)$ that for all $|t|<t_0$,
	$$
	\mathbb{E} \Big[\exp\big\{t(X^2-1)\big\} \Big] \le 1 + c t^2 \le 
	\exp(c t^2)
	$$
	where $c$ is a function of $t_0$ defined by
	$$
	c = c(t_0) = \dfrac{1}{2} \sum\limits_{r\ge 0} \Big\{(4eK^2)^2(4t_0eK^2)^r + \dfrac{4(2t_0)^{r}}{(r+2)!}\Big\}.
	$$
	Calculus of infinite series at the choice of $t_0= 1/(8eK^2)$ gives 
	$$
	c(t_0) = \dfrac{\exp(2t_0) - 1/2 - 2t_0}{2t_0^2},
	$$
	which concludes the proof of (b).
\end{proof}

\color{black}
\begin{lemma}\label{lem:subG_con_ineq}
	Assume that $X$ is a random variables satisfying Assumption \ref{assum:subG_moment} for some $K>0$.  Then, the i.i.d copies $X_1, \ldots, X_n$ of $X$ satisfy, 
	\begin{enumerate}
		\item[(a)] for $0 \le x \le eK + (2eK)^{-1}$,
		$$
		{\rm P}\bigg[ \Big|\sum\limits_{j=1}^n X_j \Big| \ge nx \bigg] \le 
		2\exp\bigg\{
		-\dfrac{nx^2}{8(1/2 + K^2 e^2)}
		\bigg\},
		$$
		\item[(b)] and for $0 \le x \le  4eK^2 R(K)$,
		$$
		{\rm P}\bigg[ \Big|\sum\limits_{j=1}^n (X_j^2 - 1) \Big| \ge nx \bigg] \le 
		2\exp\bigg\{
		-\dfrac{nx^2}{16(4eK^2)^2 R(K)}
		\bigg\},
		$$
		where $R(t) = \exp\{1/(4et^2)\} - 1/2 - 1/(4et^2)$, $t > 0$.
	\end{enumerate}
\end{lemma}
\begin{proof}
	The proofs of (a) and (b) directly come from applications of Lemma \ref{lem:asym_key} and \ref{lem:subG_mgf}.
\end{proof}

\section{Proofs}
\subsection{Proof of Lemma \ref{lem:element-wise_ineq}}\label{app:pf_lem}
\begin{proof}[Proof of Lemma \ref{lem:element-wise_ineq}]
	Assume $k$ and $\ell$ are distinct.	
	We start by decoupling the product of two sub-Gaussian variables $Y_{ik}Y_{i\ell}/\pi_{k \ell}$ using an identity $xy = \{(x+y)^2 - (x-y)^2\}/4$ so that we have for $t \ge 0$,
	\begin{equation}\label{eq:decoupling}
	\begin{array}{rl}
	\left\{\Big|
	\sum\limits_{i=1}^n \Big( \dfrac{Y_{ik} Y_{i\ell}}{\pi_{k\ell}} - \sigma_{k\ell} \Big)
	\Big|
	\ge n t
	\right\} &\subset
	\left\{
	\Big|
	\sum\limits_{i=1}^n \Big\{(Y_{ik}^* + Y_{i\ell}^*)^2 - \mathbb{E}(Y_{ik}^* + Y_{i\ell}^*)^2 \Big\}
	\Big|
	\ge \dfrac{2n\pi_{k \ell}t}{\sqrt{\sigma_{kk}\sigma_{\ell\ell}}}
	\right\}\\ 
	&
	\qquad
	\cup
	\left\{
	\Big|
	\sum\limits_{i=1}^n \Big\{(Y_{ik}^* - Y_{i\ell}^*)^2 - \mathbb{E}(Y_{ik}^* - Y_{i\ell}^*)^2 \Big\}
	\Big|
	\ge \dfrac{2n\pi_{k \ell}t}{\sqrt{\sigma_{kk}\sigma_{\ell\ell}}}
	\right\} 
	\end{array}
	\end{equation}
	where $Y_{ik}^* = Y_{ik} / \sqrt{\sigma_{kk}}$.
	Let $v_{k\ell}= \mathbb{E}|Y_{ik}^* + Y_{i\ell}^*|^2 = \pi_k + \pi_\ell + 2 \pi_{k \ell} \rho_{k\ell}$.
	To apply Lemma \ref{lem:subG_con_ineq} in Supplementary Material , we first show $Y_{ik}^* + Y_{i\ell}^*$ is a sub-Gaussian variable satisfying the conditions of the lemma.
	\begin{fact*} %\label{lem:check_cond}
		For $i=1,\ldots, n$ and $1\le k \neq \ell \le p$, we have
		$$
		\begin{array}{l}
		\sup\limits_{r\ge 1} \dfrac{\big\{\mathbb{E}|Y_{ik} + Y_{i\ell}|^r \big\}^{1/r}}{\sqrt{r v_{k\ell}}} \le 2 K / \sqrt{v_{k\ell}}.
		\end{array}
		$$
	\end{fact*}
	\begin{proof}
		To obtain an uniform bound on higher moments, we observe that
		$$
		\begin{array}{rcl}
		\dfrac{\big\{\mathbb{E}|Y_{ik}^* + Y_{i\ell}^*|^r \big\}^{1/r}}{\sqrt{r}} & \le & \dfrac{2^{1-1/r} \big\{\mathbb{E}|Y_{ik}^*|^r + \mathbb{E}|Y_{i\ell}^*|^r\big\}^{1/r}}{\sqrt{r}}\\
		&= & \dfrac{2^{1-1/r} \Big\{\pi_k \mathbb{E}\big|X_{ik}/\sqrt{\sigma_{kk}}\big|^r + \pi_\ell  \mathbb{E}\big|X_{i\ell}/\sqrt{\sigma_{\ell\ell}}\big|^r \Big\}^{1/r}}{\sqrt{r}}\\
		& \le & 
		\dfrac{2^{1-1/r} \Big\{\pi_k (\sqrt{r} K)^r + \pi_\ell ( \sqrt{r} K)^r \Big\}^{1/r}}{\sqrt{r}}\\
		& \le & 
		2 K \Big(\dfrac{\pi_k + \pi_{\ell}}{2}\Big)^{1/r}  \\
		\end{array}
		$$
		where the first inequality holds due to convexity of $x\mapsto |x|^r (r\ge 1)$ and the third inequality uses the moment condition of the sub-Gaussian variable $X_{ik}/\sqrt{\sigma_{kk}}$. 
		We note that $\Big(\dfrac{\pi_k + \pi_{\ell}}{2}\Big)^{1/r} \le 1$ for all $r \ge 1$ since $0 \le (\pi_k + \pi_{\ell})/2 \le 1$.
		which concludes the proof.
	\end{proof}

	By applying Lemma \ref{lem:subG_con_ineq} (b), we have for some numerical constants $c, C>0$,
	$$
	{\rm P}
	\bigg[
	\Big|
	\sum\limits_{i=1}^n \Big\{(Y_{ik}^* + Y_{i\ell}^*)^2 - v_{k\ell} \Big\}
	\Big|
	\ge 
	\dfrac{2n\pi_{k \ell}t}{\sqrt{\sigma_{kk}\sigma_{\ell\ell}}}
	\bigg]
	\le
	2\exp \bigg\{ 
	-\dfrac{C n\pi_{k \ell}^2 t^2}{ K^4 \sigma_{kk}\sigma_{\ell\ell} R(2K/\sqrt{v_{k\ell}})}
	\bigg\},
	$$
	for $0 \le t \le \dfrac{c (\sigma_{kk}\sigma_{\ell\ell})^{1/2}K^2 R(2K/ \sqrt{v_{k\ell}})}{\pi_{k \ell}}$.
	Hence, replacing $t$ by  
	$$
	\tilde{t} \equiv \dfrac{(\sigma_{kk} \sigma_{\ell\ell})^{1/2} K^2 R(2K/ \sqrt{v_{k\ell}})^{1/2}}{C^{1/2} \pi_{k \ell}} t, \quad t>0,
	$$ 
	in the above inequality, we get 
	$$
	{\rm P}
	\bigg[
	\Big|
	\sum\limits_{i=1}^n \Big\{(Y_{ik} + Y_{i\ell})^2 - \mathbb{E}(Y_{ik} + Y_{i\ell})^2 \Big\}
	\Big|
	\ge 2n\pi_{k \ell} \tilde{t}
	\bigg]
	\le
	2\exp \{ 
	- nt^2\}, \quad 0 \le t \le \tilde{c} \sqrt{R(2K/ \sqrt{v_{k\ell}})},
	$$
	for some numerical constant $\tilde{c}>0$.
	Note that 
	$$
	R\bigg(\dfrac{2K}{ \sqrt{\pi_k + \pi_\ell -2\pi_{k \ell} |\rho_{k\ell}|}}\bigg)  \le R\bigg(\dfrac{2K}{ \sqrt{v_{k\ell}}}\bigg)
	\le R(K),
	$$
	and using this bounds, we now have
	$$
	{\rm P}
	\bigg[
	\Big|
	\sum\limits_{i=1}^n \Big\{(Y_{ik} + Y_{i\ell})^2 - \mathbb{E}(Y_{ik} + Y_{i\ell})^2 \Big\}
	\Big|
	\ge 2n\pi_{k \ell} \dfrac{(\sigma_{kk} \sigma_{\ell\ell})^{1/2} K^2 R(K)^{1/2}}{C^{1/2} \pi_{k \ell}} t
	\bigg]
	\le
	2\exp \{ 
	- nt^2\},
	$$
	for $0 \le t \le \tilde{c} \sqrt{R\bigg(\dfrac{2K}{ \sqrt{\pi_k + \pi_\ell -2\pi_{k \ell} |\rho_{k\ell}|}}\bigg)}$. The similar statement holds with $Y_{ik} - Y_{i\ell}^*$. Therefore, combining these results with (\ref{eq:decoupling}) yield 
	$$
	{\rm P}
	\bigg[
	n^{-1} \Big|
	\sum\limits_{i=1}^n \Big( \dfrac{Y_{ik} Y_{i\ell}}{\pi_{k\ell}} - \sigma_{k\ell} \Big)
	\Big|
	\ge \dfrac{(\sigma_{kk} \sigma_{\ell\ell})^{1/2} K^2 R(K)^{1/2}}{C^{1/2} \pi_{k \ell}} t
	\bigg] \le 4\exp \{ 
	- nt^2\},
	$$
	for $0 \le t \le \tilde{c} \sqrt{R\bigg(\dfrac{2K}{ \sqrt{\pi_k + \pi_\ell -2\pi_{k \ell} |\rho_{k\ell}|}}\bigg)}$,
	which completes the proof for the case of $k\neq \ell$.
	
	The concentration inequality for diagonal entries (i.e., $k = \ell$) of the IPW estimate is similarly derived. One can easily check 
	$$
	\sup\limits_{r\ge 1} \dfrac{\big\{\mathbb{E}|Y_{ik}|^r \big\}^{1/r}}{\sqrt{r \pi_k \sigma_{kk}}} \le K / \sqrt{\pi_k}.
	$$
	Then, due to Lemma \ref{lem:subG_con_ineq} (b), we get
	$$
	{\rm P}
	\bigg[
	n^{-1} \Big|
	\sum\limits_{i=1}^n \Big( \dfrac{Y_{ik}^2 }{\pi_k} - \sigma_{kk} \Big)
	\Big|
	\ge \dfrac{\tilde{C} \sigma_{kk} K^2 R(K)^{1/2}}{\pi_k} t
	\bigg] \le 2\exp \{ 
	- nt^2\},
	$$
	for $0 \le t \le \sqrt{R(K/\sqrt{\pi_k})}$. This concludes the whole proof.
\end{proof}

\subsection{Proof of Theorem \ref{thm:dev_IPW}}\label{app:pf_dev_IPW}
\begin{proof}
	From Lemma \ref{lem:element-wise_ineq}, it holds that for $1 \le k, \ell \le p$,
	$$
	{\rm P}
	\bigg[
	n^{-1} \Big|
	\sum\limits_{i=1}^n \Big( \dfrac{Y_{ik} Y_{i\ell}}{\pi_{k\ell}} - \sigma_{k\ell} \Big)
	\Big|
	\ge \dfrac{C\sigma_{max} K^2 R(K)^{1/2}}{\pi_{min}} t
	\bigg] \le 4\exp (
	- nt^2),
	$$
	if $t\ge 0$, since $R$ is monotonically decreasing,
	$$
	\begin{cases}
	t^2 \le c  R\big(2K/ \sqrt{v_{min} }\big), & \text{if } k \neq \ell,\\
	t^2 \le cR\big(K / \sqrt{\pi_{min,d}}\big)& \text{if } k = \ell,
	\end{cases}
	$$
	where $\pi_{min, d} = \min\limits_k \pi_k$. 
	Then, by using an union bound argument, we get
	$$
	\begin{array}{rcl}
	{\rm P}
	\bigg[
	\max\limits_{k,\ell} \Big|\dfrac{1}{n}
	\sum\limits_{i=1}^n \Big( \dfrac{Y_{ik} Y_{i\ell}}{\pi_{k\ell}} - \sigma_{k\ell} \Big)
	\Big|
	\ge 
	\dfrac{C \sigma_{max}K^2 \sqrt{R(K)} ~ t }{\pi_{min}}
	\bigg] 
	&\le &
	4p^2\exp(-nt^2).
	\end{array}
	$$
	for $t^2/c \le R\big(K/ \sqrt{(v_{min}/4) \wedge \pi_{min,d}}\big) = R\big(2K/ \sqrt{v_{min}} \big)$. Note that $v_{min}/4 \le \pi_{min,d}$.

	Then, by plugging-in $t \leftarrow \alpha \sqrt{\log p / n}$ ($\alpha>0$), we get the convergence rate of the maximum norm of the IPW estimate,
	$$
	\begin{array}{rcl}
	{\rm P}
	\bigg[
	\max\limits_{k,\ell} \big|(\widehat{\bSigma}^{IPW})_{k\ell} - \sigma_{k\ell}\big|
	\ge  
	\dfrac{C \sigma_{max}K^2 \alpha }{\pi_{min}} \sqrt{\dfrac{R(K)\log p}{n}}
	\bigg]
	&\le& 4p^{2-\alpha^2},
	\end{array}
	$$
	if $0 \le \alpha^2 \le cR\big(2K/ \sqrt{v_{min} }\big) n/\log p$. Suppose $n,p$ satisfy
	$$
	n / \log p 
	> 
	\dfrac{9}{c^2 R\big(2K/ \sqrt{v_{min}}\big)}
	$$
	so that we can choose $\alpha^2 = 3$. This concludes the proof.
\end{proof}

\subsection{Proof of Corollary \ref{cor:meta}}\label{app:pf_cor}
\begin{proof} %[proof of Corollary \ref{cor:meta}]
	We summarize theorems/lemmas from the original works that bridge the rate of the plug-in estimator with those of the final precision matrix. If $\delta=\sqrt{\log p / n}$ in each theorem, then the rates of the precision matrix are optimal and guarantee both estimation consistency in different norms and support recovery ($\because || \cdot||_{max}$). As usual, $\widehat{\bSigma}^{plug}$ denotes the plug-in estimator.
	
	\subsubsection*{Graphical lasso}
	Suppose $S \subset [p]\times [p]$ is an union of a true edge set and diagonal elements. Define $\bGamma = \bOmega^{-1} \otimes \bOmega^{-1}$, 
	$$
	\bGamma_{SS} = \big(\bOmega^{-1} \otimes \bOmega^{-1}\big)_{SS} = \bOmega_{S}^{-1} \otimes \bOmega_{S}^{-1},
	$$
	and similarly $\bGamma_{eS} = \big(\bOmega^{-1} \otimes \bOmega^{-1}\big)_{eS}$, $e \in S^c$. Also, denote $\kappa_{\bSigma} = ||\bSigma||_\infty$ and $\kappa_{\bGamma}=||(\bGamma_{SS})^{-1}||_\infty$. $d$ is the maximum degree of the graph defined by $d = \max_i \sum_{j} \text{I}(|\omega_{ij}| \neq 0)$ and $s$ is the number of true edges.

	\begin{theorem*}[Lemmas 4, 5, 6, \cite{Ravikumar:2011}]
		Assume the irrepresentability condition holds with degree of $\alpha \in (0,1]$
		$$
		\max\limits_{e \in S^c} || \bGamma_{e S} \bGamma_{SS}^{-1} ||_1 \le 1 - \alpha.
		$$
		If $||\widehat{\bSigma}^{plug} - \bSigma||_{max} \le \delta = \delta_{n, p}$ and $n$ satisfies 
		$$
		\delta_{n,p} \le \Big[6d(1 + 8 \alpha^{-1})\max\{
		\kappa_{\bGamma^*}\kappa_{\bSigma^*}, \kappa_{\bGamma^*}^2\kappa_{\bSigma^*}^3\} \Big]^{-1},
		$$
		then we have
		\begin{enumerate}
			\item $||\widehat{\bOmega} - \bOmega||_{max} \le 2 \kappa_{\bGamma^*} \big(||\widehat{\bSigma}^{plug} - \bSigma||_{max} + 8\alpha^{-1}\delta \big)  \le2 \kappa_{\bGamma^*} (1 + 8\alpha^{-1}) \delta$,			
			
			\item $||\widehat{\bOmega} - \bOmega||_{2} \le 2 \kappa_{\bGamma^*}(1 + 8\alpha^{-1}) \min\{\sqrt{s+p}, d\} \delta$,
			
			\item $||\widehat{\bOmega} - \bOmega||_{F} \le 2 \kappa_{\bGamma^*}(1 + 8\alpha^{-1}) \sqrt{s+p} \,\delta$,
			
		\end{enumerate}		\color{black}
		where $\widehat{\bOmega}$ is the graphical lasso estimator that solves (\ref{eq:glasso}).
	\end{theorem*}
	\noindent
	We note that $\delta_{n,p}$ corresponds to $\bar{\delta}_f(n, p^\tau)$ in the original reference.

	\subsubsection*{CLIME}
	Let us introduce the class of a precision matrix used in \cite{Cai:2011:CLIME}. For $0 \le q  < 1$,
	$$
	\mathcal{U}(q, c_0(p)) = \bigg\{
	\bOmega \succ 0 : ||\bOmega||_{1} \le M, \, \max_{1\le i \le p} \sum\limits_{j=1}^p |\omega_{ij}|^q \le s_0(p)
	\bigg\}.
	$$
	
	\begin{theorem*}[Theorem 6, \cite{Cai:2011:CLIME}]
		If $||\bOmega||_{1} ||\widehat{\bSigma}^{plug} - \bSigma||_{max} \le \delta$, then we have
		\begin{enumerate}
			\item $||\widehat{\bOmega} - \bOmega||_{max} \le 4||\bOmega||_{1}\delta$,
			
			\item $||\widehat{\bOmega} - \bOmega||_{2} \le C s_0(p) (4 || \bOmega||_{1}\delta)^{1-q}$, if $\bOmega \in \mathcal{U}(q, c_0(p))$,
			
			\item $||\widehat{\bOmega} - \bOmega||_{F}^2 / p \le C s_0(p) (4 || \bOmega||_{1}\delta)^{2-q}$, if $\bOmega \in \mathcal{U}(q, c_0(p))$,
			
		\end{enumerate}
		\color{black}
		where $\widehat{\bOmega}$ is the CLIME estimator that solves (\ref{eq:clime}) and $C>0$ is a numerical constant.
	\end{theorem*}
	
	\subsubsection*{Graphical Dantzig selector}
	The graphical Dantzig selector aims to solve $p$ optimization problems below (\cite{Yuan:2010})
	\begin{equation}\label{eq:graphical_dantzig}
	\min_{\beta_j \in \mathbb{R}^{p-1}} || \beta_j||_1, \quad \text{subject to } || \widehat{\bSigma}^{plug}_{-j,j} - \widehat{\bSigma}^{plug}_{-j,-j} \beta_j ||_{\infty} \le \lambda,
	\end{equation}
	for $j=1,\ldots, p$. Let $d$ be the maximum degree of the graph, or equivalently $d = \max_i \sum_{j} \text{I}(|\omega_{ij}| \neq 0)$.
	\begin{theorem*}[A consequence of Lemma 11, \cite{Yuan:2010}]
		Assume $\bOmega \in O(v, \eta, \tau)$ defined by
		$$
		O(v, \eta, \tau) = \Big\{ 
		\bOmega \succ 0 : v^{-1} \le \lambda_{min}(\bOmega) \le \lambda_{max}(\bOmega) \le v, ||\bSigma \bOmega - \bI||_{max} \le \eta, ||\bOmega||_1 \le \tau
		\Big\}.
		$$
		If $\tau v ||\widehat{\bSigma}^{plug} - \bSigma||_{max} + \eta v \le \delta$, 
		then we have
		$$
		||\widehat{\bOmega} - \bOmega||_1 \le C d \delta,
		$$
		where $\widehat{\bOmega}$ is the graphical Dantzig estimator that solves (\ref{eq:graphical_dantzig}) and $C$ depends only on $v, \tau, \lambda_{min}(\bOmega), \lambda_{max}(\bOmega)$.
	\end{theorem*}
	\noindent
	Note that the $\ell_1$-norm of a matrix bounds the spectral norm, so we also have
	$$
	||\widehat{\bOmega} - \bOmega||_2 \le C d \delta.
	$$
\end{proof}

\subsection{Proof of Theorem \ref{thm:dev_IPW_unknown_mean}}\label{app:pf_dev_IPW_unknown_mean}
\begin{proof}
	Recall that the the proposed estimator when mean is not known has its form as follows: $\widehat{\bSigma}^{IPW\mu} = 
	\Big((\widehat{\bSigma}^{IPW\mu})_{k\ell}, 1\le k, \ell \le p\Big)$ with 
	\begin{equation}%\label{eq:IPWest_unknown_mean}
	(\widehat{\bSigma}^{IPW\mu})_{k\ell} = \dfrac{\sum_{i=1}^n \tilde{Y}_{ik} \tilde{Y}_{i\ell}}{n \pi_{k \ell}} - \dfrac{\sum_{i\neq j}^n \tilde{Y}_{ik} \tilde{Y}_{j\ell}}{n(n-1) \pi_k \pi_{\ell}}.
	\end{equation}
	Let $(k,\ell)$ be a dual in $\{1,\ldots, p\}^2$. Using $\tilde{Y}_{ik} = \delta_{ik} X_{ik} + \delta_{ik} \mu_k = Y_{ik} + \delta_{ik} \mu_k$, we can decompose the first term in (\ref{eq:IPWest_unknown_mean}) as follows.
	$$
	\begin{array}{rl}
	&\dfrac{\sum_{i=1}^n \tilde{Y}_{ik} \tilde{Y}_{i\ell}}{n \pi_{k \ell}} - (\sigma_{k\ell} + \mu_k \mu_\ell)
	\\[1em] 
	=&\bigg\{
	\dfrac{\sum_{i=1}^n Y_{ik} Y_{i\ell}}{n \pi_{k \ell}} - \sigma_{k\ell}\bigg\} + \bigg\{ \dfrac{\sum_{i=1}^n \delta_{ik}\delta_{i\ell} \mu_k X_{i\ell}}{n \pi_{k \ell}}
	\bigg\} 
	+ 
	\bigg\{\dfrac{\sum_{i=1}^n \delta_{ik}\delta_{i\ell}X_{ik}\mu_\ell }{n \pi_{k \ell}}\bigg\}
	+ 
	\bigg\{\dfrac{\sum_{i=1}^n \delta_{ik}\delta_{i\ell} \mu_k \mu_\ell}{n \pi_{k \ell}} - \mu_k \mu_\ell \bigg\}\\
	=& A_1 + A_2 + A_3 + A_4.
	\end{array}
	$$
	A deviation inequality for $A_1$ comes from Lemma \ref{lem:element-wise_ineq}. On the other hands, since $A_2$, $A_3$, and $A_4$ are independent sum of sub-Gaussian variables, the related concentration inequalities can be found in Lemma \ref{lem:subG_con_ineq} (a) and \ref{lem:dev_binom}.
	%$$
	%\begin{array}{rl}
	%&
	%\dfrac{\sum_{i\neq j} \tilde{Y}_{ik} \tilde{Y}_{j\ell}}{n(n-1) \pi_k \pi_{\ell}} - \mu_k \mu_\ell
	%\\[1em]
	%=& 
	%\bigg\{\dfrac{\sum_{i\neq j} \delta_{ik}\delta_{j\ell} X_{ik}X_{j\ell}}{n(n-1) \pi_k \pi_{\ell}} 
	%\bigg\}+
	%\bigg\{\dfrac{\sum_{i\neq j} \delta_{ik}\delta_{j\ell}  X_{ik}\mu_\ell}{n(n-1) \pi_k \pi_{\ell}} 
	%\bigg\}+
	%\bigg\{\dfrac{\sum_{i\neq j} \delta_{ik}\delta_{j\ell} \mu_k X_{j\ell}}{n(n-1) \pi_k \pi_{\ell}}
	%\bigg\}+
	%\bigg\{\dfrac{\sum_{i\neq j} \delta_{ik}\delta_{j\ell} \mu_k \mu_\ell}{n(n-1) \pi_k \pi_{\ell}}
	%- \mu_k \mu_\ell
	%\bigg\} \\
	%=& B_1 + B_2 + B_3 + B_4.
	%\end{array}
	%$$
	The second term in (\ref{eq:IPWest_unknown_mean}) can be decomposed by
	$$
	\begin{array}{rl}
	&
	\dfrac{\sum_{i\neq j} \tilde{Y}_{ik} \tilde{Y}_{j\ell}}{n(n-1) \pi_k \pi_{\ell}} - \mu_k \mu_\ell
	\\[1em]
	=& 
	\dfrac{\sum_{i\neq j} (\tilde{Y}_{ik} - \mathbb{E}\tilde{Y}_{ik}) (\tilde{Y}_{j\ell} - \mathbb{E} \tilde{Y}_{j\ell})}{n(n-1) \pi_k \pi_{\ell}} + 
	\dfrac{\sum_{i\neq j}(\tilde{Y}_{ik} - \mathbb{E}\tilde{Y}_{ik})\mathbb{E} \tilde{Y}_{j\ell} }{n(n-1) \pi_k \pi_{\ell}} + 
	\dfrac{\sum_{i\neq j} (\tilde{Y}_{i\ell} - \mathbb{E}\tilde{Y}_{i\ell})\mathbb{E} \tilde{Y}_{ik}}{n(n-1) \pi_k \pi_{\ell}}\\[1em]
	=& 
	\dfrac{\sum_{i\neq j} (\tilde{Y}_{ik} - \mathbb{E}\tilde{Y}_{ik}) (\tilde{Y}_{j\ell} - \mathbb{E} \tilde{Y}_{j\ell})}{n(n-1) \pi_k \pi_{\ell}} + 
	\dfrac{\mu_\ell \sum_{i=1}^n (\tilde{Y}_{ik} - \mathbb{E}\tilde{Y}_{ik})}{n \pi_k} + 
	\dfrac{ \mu_k \sum_{i=1}^n (\tilde{Y}_{i\ell} - \mathbb{E}\tilde{Y}_{i\ell})}{n \pi_{\ell}}\\[1em]
	=&
	\dfrac{\sum_{i\neq j} (\tilde{Y}_{ik} - \mathbb{E}\tilde{Y}_{ik}) (\tilde{Y}_{j\ell} - \mathbb{E} \tilde{Y}_{j\ell})}{n(n-1) \pi_k \pi_{\ell}} 
	+ 
	\dfrac{\mu_\ell \sum_{i=1}^n \delta_{ik} X_{ik}}{n \pi_k} +
	\dfrac{\mu_\ell  \sum_{i=1}^n (\delta_{ik} - \pi_k) }{n \pi_k} 
	\\[1em]
	& 
	\qquad + \dfrac{ \mu_k \sum_{i=1}^n \delta_{i\ell} X_{i\ell}}{n \pi_{\ell}}
	+ \dfrac{ \mu_k  \sum_{i=1}^n (\delta_{i\ell} - \pi_\ell) }{n \pi_{\ell}}\\
	=&
	B_1 + B_2 + B_3 + B_4 + B_5.
	\end{array}
	$$
	The concentration of each term except $B_1$ is easily derived using Lemma \ref{lem:subG_con_ineq} (a) and \ref{lem:dev_binom}. 
	To analyze the concentration of $B_1$ which is a dependent sum of cross-product of sub-Gaussian variables, we need a new version of Hanson-Wright inequality. Lemma \ref{lem:hanson_mod} is more general than that given in \cite{Rudelson:2013} in the sense that two random variables $X_i, Y_i$ are not necessarily equal.
	The generalization is possible because of the decoupling technique from which we can separately handle $\{X_i:i \in \Lambda\}$ and $\{Y_i:i \notin \Lambda\}$ for some $\Lambda \subset \{1,\ldots, n\}$. Details of the proof of Lemma \ref{lem:hanson_mod} can be found in Section \ref{app:pf_hanson_mod}.
	
	\begin{lemma}\label{lem:hanson_mod}
		Let $(X, Y)$ be a pair of (possibly correlated) random variables satisfying $\mathbb{E} X =\mathbb{E} Y=0$, and
		$$
		\sup_{r\ge 1} \dfrac{\big\{ \mathbb{E}|X|^r \big\}^{1/r}}{\sqrt{r}} \le K_X, \quad  \sup_{r\ge 1} \dfrac{\big\{ \mathbb{E}|Y|^r \big\}^{1/r}}{\sqrt{r}} \le K_Y.
		$$
		Assume $n$ copies $\{(X_i, Y_i)\}_{i=1}^n$ of $(X,Y)$ are independently observed. For a matrix $\bA=(a_{ij}, 1\le i,j \le n)$ with zero diagonals $a_{ii}=0$, we have that
		$$
		{\rm P} \bigg[
		\big|\sum_{i \neq j} a_{ij} X_i Y_j \big| > t 
		\bigg] \le 2\exp \Big\{
		- c \min \Big(\dfrac{t^2}{K_X^2 K_Y^2 ||\bA||_F^2}, \dfrac{t}{K_X K_Y ||\bA||_2}\Big)
		\Big\}, \quad t\ge 0.
		$$
		for some numerical constant $c>0$.
	\end{lemma}
	\noindent
	Now, we get the concentration bound for $B_1$ using the lemma above;
	$$
	{\rm P} \bigg[
	\Big|
	\dfrac{\sum_{i\neq j} (\tilde{Y}_{ik} - \mathbb{E}\tilde{Y}_{ik}) (\tilde{Y}_{j\ell} - \mathbb{E} \tilde{Y}_{j\ell})}{n(n-1) \pi_k \pi_{\ell}}
	\Big| > t 
	\bigg] \le 2\exp 
	\bigg\{
	-\dfrac{c \pi_k\pi_\ell n t}{\sigma_{kk}^{1/2} \sigma_{\ell\ell}^{1/2}K^2}
	\bigg\},
	$$
	for $t \ge \dfrac{\sigma_{kk}^{1/2} \sigma_{\ell\ell}^{1/2} K^2}{\pi_k \pi_\ell n}$,
	since the matrix in $\mathbb{R}^{n\times n}$ with off-diagonals $1$ and diagonals $0$ has both Frobenius and spectral norms bounded above by $n$.
	\color{black}
	By replacing $t$ with $\dfrac{t\sigma_{kk}^{1/2} \sigma_{\ell\ell}^{1/2} K^2}{c\pi_k \pi_\ell}\sqrt{\dfrac{\log p}{n}}$, we have
	$$
	\begin{array}{rcl}
	{\rm P} \bigg[
	\Big|
	\dfrac{\sum_{i\neq j} (\tilde{Y}_{ik} - \mathbb{E}\tilde{Y}_{ik}) (\tilde{Y}_{j\ell} - \mathbb{E} \tilde{Y}_{j\ell})}{n(n-1) \pi_k \pi_{\ell}}
	\Big| > \dfrac{t \sigma_{kk}^{1/2} \sigma_{\ell\ell}^{1/2} K^2}{\pi_k \pi_\ell}\sqrt{\dfrac{\log p}{n}}
	\bigg] 
	&\le & 2\exp \big\{
	-t \sqrt{n \log p }
	\big\}
	\end{array}
	$$
	for $t \sqrt{n \log p} \ge c$.
	Then, if we assume $n > \log p$, the probability above is bounded by $2 p^{-t}$. 
	
	Combining all results for $A_1, \ldots, A_4$, $B_1, \ldots, B_5$, we can derive the concentration inequality for each component of $\widehat{\bSigma}^{IPW\mu}$, which completes the proof.
\end{proof}

\subsection{Proof of Lemma \ref{lem:hanson_mod}}\label{app:pf_hanson_mod}
\begin{proof}
	Without loss of generality, we assume $K_X=K_Y=1$.
	Let $\{\eta_i\}_{i=1}^n$ be independent Bernoulli variables with success probability $1/2$. Then, by observing $\mathbb{E} \eta_i (1-\eta_j)= \text{I}(i\neq j) / 4$, it can be seen that $S \equiv \sum_{i\neq j} a_{ij}X_i Y_j = 4 \mathbb{E}_{\{\eta_i \}}S_\eta$ where $S_\eta=\sum_{i,j} \eta_i (1-\eta_j) a_{ij}X_i Y_j$ and $\mathbb{E}_{\{\eta_i \}}$ is an expectation taken over $\{\eta_i \}$. Let $\Lambda_\eta=\{i : \eta_i=1 \}$ be the index set of successes. Since $S_\eta = \sum_{i\in \Lambda_\eta,j \in \Lambda_\eta^c} a_{ij}X_i Y_j$ is a function of $\{Y_j:j \in \Lambda_\eta^c\}$ given $\{\eta_i\}$ and $\{X_i:i\in \Lambda_\eta\}$,
	$S_\eta$ conditionally follows is a sub-Gaussian distribution. 
	
	We assume $\{\eta_i\}$ is conditioned on all the following statements unless specified otherwise. Then, the previous results yield
	$$
	\begin{array}{rcl}
	\mathbb{E}_{\{(X_j,Y_j):j \in \Lambda_\eta^c \}} \Big[ \exp(4\lambda S_\eta) \Big| \{X_i:i\in \Lambda_\eta\} \Big]&=&
	\mathbb{E}_{\{Y_j:j \in \Lambda_\eta^c \}} \Big[ \exp(4\lambda S_\eta)\Big| \{X_i:i\in \Lambda_\eta\} \Big]\\
	& \le &
	\exp\Big\{
	c \lambda^2 \sum_{j \in \Lambda_\eta^c} (\sum_{i \in \Lambda_\eta} a_{ij} X_i)^2
	\Big\},
	\end{array}
	$$
	where the equality holds since $\exp(4\lambda S_\eta)$ does not depend on $\{X_j\}_{j\in \Lambda_\eta^c}$ and 
	the inequality is from sub-Gaussianity of $S_\eta$. 
	Taking expectation with respect to $\{X_i:i\in \Lambda_\eta\}$ on both sides, we get the following result;
	$$
	\begin{array}{rcl}
	\mathbb{E}_{\{X_i:i\in \Lambda_\eta\}, \{(X_j,Y_j):j \in \Lambda_\eta^c \}} \big[\exp(4\lambda S_\eta)\big] &\le& 
	\mathbb{E}_{\{X_i:i\in \Lambda_\eta\}}\bigg[ \exp\Big\{
	c \lambda^2 \sum_{j \in \Lambda_\eta^c} (\sum_{i \in \Lambda_\eta} a_{ij} X_i)^2
	\Big\}\bigg] \\
	&=& 
	\mathbb{E}_{\{X_i\}} \bigg[\exp\Big\{
	c \lambda^2 \sum_{j \in \Lambda_\eta^c} (\sum_{i \in \Lambda_\eta} a_{ij} X_i)^2
	\Big\}\bigg],
	\end{array}
	$$
	where the equality holds from independence among $n$ samples. Also, since the left-hand side does not depend on $\{Y_i:i\in \Lambda_\eta\}$, we get
	$$
	\mathbb{E}_{\{(X_i,Y_i)\}} \big[\exp(4\lambda S_\eta) \big| \{\eta_i\} \big] \le 
	\mathbb{E}_{\{X_i\}} \bigg[ \exp\Big\{
	c \lambda^2 \sum_{j \in \Lambda_\eta^c} \big(\sum_{i \in \Lambda_\eta} a_{ij} X_i\big)^2
	\Big\}
	\Big|
	\{\eta_i\} 
	\bigg] (\equiv T_\eta),
	$$
	where we begin to display the conditional dependency on $\{\eta_i\}$. Following the step 3 and 4 in \cite{Rudelson:2013}, we can achieve an uniform bound of $T_\eta$ independent of $\{\eta_i\}$ and thus get
	$$
	T_\eta \le \exp\{ C \lambda^2 ||A||_F^2 \} \quad \text{for } \lambda \le c / ||A||_2,
	$$
	for some positive constants $c$ and $C$.
	Then, we have 
	$$
	\begin{array}{rcl}
	\mathbb{E} \big[\exp(\lambda S)\big]  & = & \mathbb{E} \big[\exp(\mathbb{E}_{\{\eta_i \}} 4\lambda S_\eta) \big] \\
	& \le &  \mathbb{E}_{\{(X_i,Y_i)\}_i, \{\eta_i \}} \big[\exp(4\lambda S_\eta)\big] (\because \text{Jensen's inequality}) \\
	&=&
	\mathbb{E} \Big[
	\mathbb{E}\big[\exp(4\lambda S_\eta) \big| \{\eta_i\} \big] 
	\Big] ~ = ~ \mathbb{E} [T_\eta] ~ \le ~ \exp\{ C \lambda^2 ||A||_F^2 \} \quad \text{for } \lambda \le c / ||A||_2,
	\end{array}
	$$
	Following the step 5 in \cite{Rudelson:2013}, we can get the concentration of $S$ given in the lemma below. Let $||X||_{\psi_2}$ be a $\psi_2$-norm of $X$ defined by
	$$
	||X||_{\psi_2} = \inf \Big\{R>0: \mathbb{E} e^{\frac{|X|^2}{R^2}} \le 2\Big\}.
	$$
	\begin{lemma}\label{lem:new_HW_ineq}
		Let $(X, Y)$ be a pair of (possibly correlated) random variables satisfying $\mathbb{E} X =\mathbb{E} Y=0$, and
		\begin{equation}\label{eq:cond_psi2}
		||X||_{\psi_2} \le K_X, ||Y||_{\psi_2} \le K_Y.
		\end{equation}
		Assume $n$ samples $\{(X_i, Y_i)\}_{i=1}^n$ are identically and independently observed. For a matrix $A=(a_{ij}, 1\le i,j \le n)$ with zero diagonals, we have that
		$$
		{\rm P} \Big[
		\big|\sum_{i \neq j} a_{ij} X_i Y_j \big| > t 
		\Big] \le 2\exp \Big\{
		- c \min \Big(\dfrac{t^2}{K_X^2 K_Y^2 ||A||_F^2}, \dfrac{t}{K_X K_Y ||A||_2}\Big)
		\Big\}, \quad t\ge 0.
		$$
		for some numerical constant $c>0$.
	\end{lemma}
	\noindent
	Note that the finite $\psi_2$-norm in (\ref{eq:cond_psi2}) characterizes a sub-Gaussian random variable and can be replaced by the uniformly bounded moments in (\ref{eq:cond_cumulant}), since $\sup_{r\ge 1} \big\{ \mathbb{E}|X|^r \big\}^{1/r}/\sqrt{r} \le K$ implies $||X||_{\psi_2} \le 2eK$. In other words, provided $X_i$ and $Y_j$ satisfy the moment condition with constants $K_X$ and $K_Y$, respectively, the conclusion of the lemma above still holds (with different $c$). This completes the proof.
\end{proof}

\subsection{Proof of Theorem \ref{thm:dev_IPW_unknown_misprob}}\label{app:pf_dev_IPW_unknown_misprob}
Theorem \ref{thm:dev_IPW_unknown_misprob} is not difficult to show if Lemmas \ref{lem:element-wise_ineq}, \ref{lem:dev_emp_decomp}, and \ref{lem:dev_invprop} are used together. Let us show and prove the two additional lemmas.
First, the following lemma shows how the concentration of (\ref{eq:IPWest_unknown_misprob}) is related to that of $\hat{\pi}_{jk}$.
\begin{lemma}\label{lem:dev_emp_decomp}
	Assume
	%	\begin{equation}
	$$
	\begin{array}{l}
	\max_{k,\ell}|1/\pi_{k\ell} - 1 / \hat{\pi}_{k\ell}| < B_1, \quad 	\hat{\pi}_{k\ell} >0, \forall k,\ell,\\
	\big|\big|\bS_Y - \bSigma^\pi\big|\big|_{max} < B_2, \quad 
	\big|\big|\widehat{\bSigma}^{IPW} - \bSigma\big|\big|_{max} < B_3\\
	\end{array}
	$$
	%	\end{equation}
	where $B_1, B_2$, and $B_3$ are positive constants. Then, we have
	$$
	\big|\big|\widehat{\bSigma}^{IPW\pi} - \bSigma\big|\big|_{max}  
	\le 
	B_1B_2 + B_1 \sigma_{max} + B_3.
	$$
\end{lemma}

\begin{proof}
	By the triangular inequality, we observe
	%	$$
	%	\big|\big|\widehat{\bSigma}^{IPW\pi} - \bSigma\big|\big|_{max} \le \big|\big|\widehat{\bSigma}^{IPW\pi} - \widehat{\bSigma}^{IPW}\big|\big|_{max} + \big|\big|\widehat{\bSigma}^{IPW} - \bSigma\big|\big|_{max}.
	%	$$
	%	It is sufficient to show the first random variable at the right-hand side is of the same order as that of the second that is already known in Theorem \ref{thm:dev_IPW}.
	%	We can decompose the difference of two estimators by
	$$
	\begin{array}{rcl}
	\big|\big|\widehat{\bSigma}^{IPW\pi} - \bSigma\big|\big|_{max} 
	&\le& \big|\big|\widehat{\bSigma}^{IPW\pi} - \widehat{\bSigma}^{IPW}\big|\big|_{max} + \big|\big|\widehat{\bSigma}^{IPW} - \bSigma\big|\big|_{max}\\
	%\big|\big|\widehat{\bSigma}^{IPW\pi} - \widehat{\bSigma}^{IPW}\big|\big|_{max} 
	&\le&
	\max_{k,\ell}|1/\pi_{k\ell} - 1 / \hat{\pi}_{k\ell} | \cdot \big|\big| \bS_Y\big|\big|_{max}  + \big|\big|\widehat{\bSigma}^{IPW} - \bSigma\big|\big|_{max}\\
	&\le & \max_{k,\ell}|1/\pi_{k\ell} - 1 / \hat{\pi}_{k\ell} | \cdot 
	\big|\big| \bS_Y - \boldsymbol{\bSigma}^\pi\big|\big|_{max} + 
	\max_{k,\ell}|1/\pi_{k\ell} - 1 / \hat{\pi}_{k\ell} | \cdot 
	\big|\big| \boldsymbol{\bSigma}^\pi\big|\big|_{max}\\
	& & + \big|\big|\widehat{\bSigma}^{IPW} - \bSigma\big|\big|_{max}
	\end{array}
	$$
	where $\bS_Y = n^{-1} \sum_{i=1}^n Y_i Y_i^{\rm T}$ and $\bSigma^\pi = \Big(\pi_{jk}\sigma_{jk}, 1\le j,k \le p\Big)$.
	Thus, we get
	$$
	\begin{array}{rcl}
	\big|\big|\widehat{\bSigma}^{IPW\pi} - \bSigma\big|\big|_{max}  
	&\le &
	B_1B_2 + B_1 \big|\big|\boldsymbol{\bSigma}^\pi\big|\big|_{max} + B_3.
	\end{array}
	$$
	Finally, we note that
	$$
	\big|\big|\bSigma^\pi \big|\big|_{max} \le \big|\big|\bSigma\big|\big|_{max} = \sigma_{max}
	$$
	where the last equality holds for a symmetric positive definite matrix.
\end{proof}

\begin{lemma}\label{lem:dev_invprop}
	Assume the sample size and dimension satisfy $n/\log p > C/ \pi_{min}$ for some numerical constant $C>0$. Then, it holds that with probability at most $2/p$
	\begin{equation}\label{eq:invprop_rate}
	\max_{k,\ell}|1/\pi_{k\ell} - 1 / \hat{\pi}_{k\ell}^{emp}| \ge \sqrt{\dfrac{C \log p}{\pi_{min}^2 n}},  \text{ and }\, \hat{\pi}_{k\ell}^{emp} >0, \forall k,\ell.
	\end{equation}
\end{lemma}
\begin{proof}
	
	First, we observe that on the event $G = G_{n,p}= \{\hat{\pi}^{emp}_{k \ell}>0, \forall k, \ell\}$, we have for $t>0$
	$$
	%\begin{array}{rcl}
	|1/\pi_{k\ell} - 1 / \hat{\pi}^{emp}_{k\ell}| \ge t ~ \Leftrightarrow ~ (1-t\pi_{k \ell}) \hat{\pi}^{emp}_{k\ell} \ge \pi_{k\ell} \text{ or } 
	(1+t\pi_{k \ell}) \hat{\pi}^{emp}_{k\ell} \le \pi_{k\ell}.
	%\end{array}
	$$
	Let $A_{k\ell}=\{ (1-t\pi_{k \ell}) \hat{\pi}^{emp}_{k\ell} \ge \pi_{k\ell} \}$ and $B_{k\ell} = \{(1+t\pi_{k \ell}) \hat{\pi}^{emp}_{k\ell} \le \pi_{k\ell}\}$. Using these notations, we get
	$$
	\begin{array}{rcl}
	{\rm P}\bigg[ \Big\{\max_{k,\ell}|1/\pi_{k\ell} - 1 / \hat{\pi}^{emp}_{k\ell}| \ge t \Big\} \cap G \bigg]
	&=& 
	{\rm P}\bigg[ G \cap \big\{\cup_{k,\ell} (A_{k\ell} \cup B_{k\ell})\big\} \bigg] \\
	&\le& {\rm P}\big[ \cup_{k,\ell} (A_{k\ell} \cup B_{k\ell}) \big] ~\le~ \sum\limits_{k,\ell} {\rm P} (A_{k\ell} \cup B_{k\ell}).
	\end{array}
	$$
	We introduce the deviation inequality for a sum of Bernoulli variables.
	\begin{lemma}[\cite{Boucheron:2016}, p 48]\label{lem:dev_binom}
		Let $\{\delta_i\}_{i=1}^n$ be independent Bernoulli variables with probability $\pi$ of being $1$. Then, there exists a numerical constant $C>0$ such that for $t>0$,
		$$
		{\rm P}\bigg[\pm\sum_{i=1}^n (\delta_i - \pi) \ge nt\bigg] \le \exp(-C n \pi t^2).
		$$
	\end{lemma}
	\noindent
	If $t< \pi_{k \ell}^{-1}$, by using Lemma \ref{lem:dev_binom}, it holds
	$$
	{\rm P}(A_{k\ell})  = 
	{\rm P}\bigg[\hat{\pi}^{emp}_{k\ell} - \pi_{k\ell} \ge  \dfrac{t \pi_{k\ell}}{1 - t \pi_{k\ell}} \bigg] 
	\le \exp\Big\{ 
	- \dfrac{Cnt^2 \pi_{k\ell}^3}{(1 - t \pi_{k\ell})^2}
	\Big\}.
	$$
	Similarly, we have
	$$
	{\rm P}(B_{k\ell}) \le \exp\Big\{ 
	- \dfrac{Cnt^2 \pi_{k\ell}^3}{(1 + t \pi_{k\ell})^2}
	\Big\}.
	$$
	If we define $\pi_{min}=\min_{k, \ell} \pi_{k \ell}$, we get by the union argument
	$$
	{\rm P}(A_{k\ell} \cup B_{k\ell}) \le 
	2\exp\Big\{ 
	- \dfrac{Cnt^2 \pi_{k\ell}^3}{(1 + t \pi_{k\ell})^2} 
	\Big\}
	\le 2\exp\Big\{ 
	- \dfrac{Cnt^2 \pi_{min}^3}{(1 + t \pi_{min})^2}
	\Big\}
	$$
	where the last inequality depends on monotonicity of $x\in (0,1) \mapsto \dfrac{x^3}{(1 + tx)^2}$ for $t>0$.
	Combining these results, we can conclude
	$$
	{\rm P}\bigg[ \Big\{\max_{k,\ell}|1/\pi_{k\ell} - 1 / \hat{\pi}^{emp}_{k\ell}| \ge t \Big\} \cap G \bigg] \le 
	2p^2 \exp\Big\{ 
	- \dfrac{Cnt^2 \pi_{min}^3}{(1 + t \pi_{min})^2}
	\Big\}.
	%	\le 	2\exp\Big\{ - \dfrac{C\pi_{min}^3 \alpha^2 \log p }{(1 + \alpha\pi_{min} \sqrt{\log p / n})^2}
	%	\Big\},
	$$
	If $t \leftarrow \sqrt{4\log p / (C\pi_{min}^2 n)}$ and assume $n/\log p > 12 /(C \pi_{min})$, then we can derive 
	$$
	{\rm P}\bigg[ \Big\{\max_{k,\ell}|1/\pi_{k\ell} - 1 / \hat{\pi}^{emp}_{k\ell}| \ge 
	\dfrac{2}{\pi_{min}}\sqrt{\dfrac{\log p}{C n}}
	\Big\} \cap G \bigg]
	\le \dfrac{2}{p},
	$$
	which completes the proof.
\end{proof}

\section{Additional analyses and details in Section 4}
\subsection{Non-PSD input for CLIME}\label{app:nonPSD_clime}
In what follows, we distinguish between a plug-in matrix (estimator) $\widehat{\bSigma}^{plug}$ and an initial matrix (estimator) $\bSigma^{(0)}$ (or $\bOmega^{(0)}$) that is used to initialize iterative steps.

We analyze the CLIME method proposed by \cite{Cai:2011:CLIME}, which solves
\begin{equation}\label{eq:clime}
\min\limits_{\bOmega \in \mathbb{R}^{p\times p}} |\bOmega|_1 \quad \text{s.t.} \quad || \widehat{\bSigma}^{plug} \bOmega - \bI ||_{max} \le \lambda.
\end{equation}
%This optimization problem has an equivalent form
%\begin{equation}\label{eq:clime_decomp}
%\min\limits_{\beta_j \in \mathbb{R}^p} ||\beta_j||_1 \quad \text{s.t.} \quad || \widehat{\bSigma}^{plug} \beta_j - e_j ||_{max} \le \lambda, \quad 1 \le j \le p,
%\end{equation}
%where $e_j$ is a unit vector (in $\ell_2$-norm) with $1$ placed at the $j$-th entry. A matrix $[\hat{\beta}_1, \ldots, \hat{\beta}_p] \in \mathbb{R}^{p\times p}$ that binds solution vectors also solves (\ref{eq:clime}). 
\cite{Cai:2011:CLIME} divide (\ref{eq:clime}) into $p$ column-wise problems and relax each problem to be a linear programming, which leads to Algorithm \ref{alg:clime}.
\begin{algorithm}[H]
	\caption{The CLIME algorithm}
	\label{alg:clime}
	\begin{algorithmic}[1]
		\REQUIRE An initial matrix $\bOmega^{(0)}$ of $\bOmega$, the plug-in matrix $\widehat{\bSigma}^{plug}$.
		\FOR{$j = 1, \ldots, p,$}
		\STATE Solve the linear programming below. We use the $j$-th column of $\bOmega^{(0)}$ for initialization of $\beta_j$
		\begin{equation}\label{eq:clime_relax}
		(\hat{r}, \hat{\beta_j}) = \arg\min_{r, \beta_j \in \mathbb{R}^p} ||r||_1
		\quad \text{s.t.}\quad
		|\beta_j| \le r (\text{element-wise}), 
		|| \widehat{\bSigma}^{plug}\beta_j - e_j ||_{max} \le \lambda.
		\end{equation}
		\ENDFOR % end for j
		\ENSURE $\widehat{\bOmega} = [\hat{\beta}_1, \ldots, \hat{\beta}_p]$: the final estimate.
	\end{algorithmic}
\end{algorithm}
\noindent
It is easily seen that the optimization problem (\ref{eq:clime}) is convex regardless of the plug-in matrix. Moreover, Algorithm \ref{alg:clime} does not require any constraint in the two inputs for a well-defined solution, contrary to Algorithm \ref{alg:glasso}.
However, the current implementations (e.g. \pkg{clime} version 0.4.1 from \cite{Cai:2011:CLIME}, \pkg{fastclime} version 1.4.1 from \cite{Pang:2014})
set the initial by solving $\bOmega^{(0)}(\widehat{\bSigma}^{plug} + \lambda \bI) = \bI$, which is not applicable to our case since an initialization from $\bOmega^{(0)}(\widehat{\bSigma}^{IPW} + \lambda \bI) = \bI$ is not well-posed unless $\widehat{\bSigma}^{IPW} + \lambda \bI$ is positive definite.
\cite{Katayama:2018} also point out that the solution of (\ref{eq:clime}) may not exist, unless an input matrix $\widehat{\bSigma}^{plug}$ is guaranteed to be PSD. We conjecture this irregularity is due to the initialization.
Thus, our proposal for the inputs is 
$$
\widehat{\bSigma}^{plug} \leftarrow \widehat{\bSigma}^{IPW}, \quad \bOmega^{(0)} \leftarrow \text{diag}\big(\widehat{\bSigma}^{IPW})^{-1}.
$$
Similarly to the graphical lasso, one should modify the implemented R functions (e.g. \pkg{clime} in \pkg{clime} package) to separately handle two inputs, since it is not allowed for now to control two input matrices $\bOmega^{(0)}$ and $\widehat{\bSigma}^{plug}$ independently.

\subsection{Failure of Algorithm \ref{alg:glasso} under missing data}\label{sec:simul_fail}
%It is numerically verified that the coordinate descent algorithm fails when the IPW estimator is plugged-in as discussed in Section \ref{sec:PSDness}, but \cite{Hsieh:2014}'s algorithm does not. 
It is mentioned that the undesirable property, non-PSDness, of the IPW estimator may hamper downstream multivariate procedures. We give one of the examples where it causes a problem; the graphical lasso. Recall that the existing algorithms available in \pkg{glasso} and \pkg{huge} packages are not suitable especially with the tuning parameter fixed at small $\lambda$, since they use the non-PSD initial matrix $\bSigma^{(0)}=\widehat{\bSigma}^{IPW} + \lambda \bI$. As a consequence, in Figure \ref{fig:ROC_QUIC_huge} where data is similarly generated to the simulation study (see Section \ref{sec:simul_res}), the blue solid ROC curves end at FPR values far less than $1$ when the coordinate descent algorithm provided in \pkg{huge} is used. On the contrary, the QUIC algorithm (red dashed) returns a full length of ROC curves. It is noted that since the graphical lasso has a unique solution, two algorithms create the same path, as long as convergence is reached.
\begin{figure}[H]	
	\centering
	%		\begin{minipage}{0.45\linewidth}
	\includegraphics[width=1\linewidth]{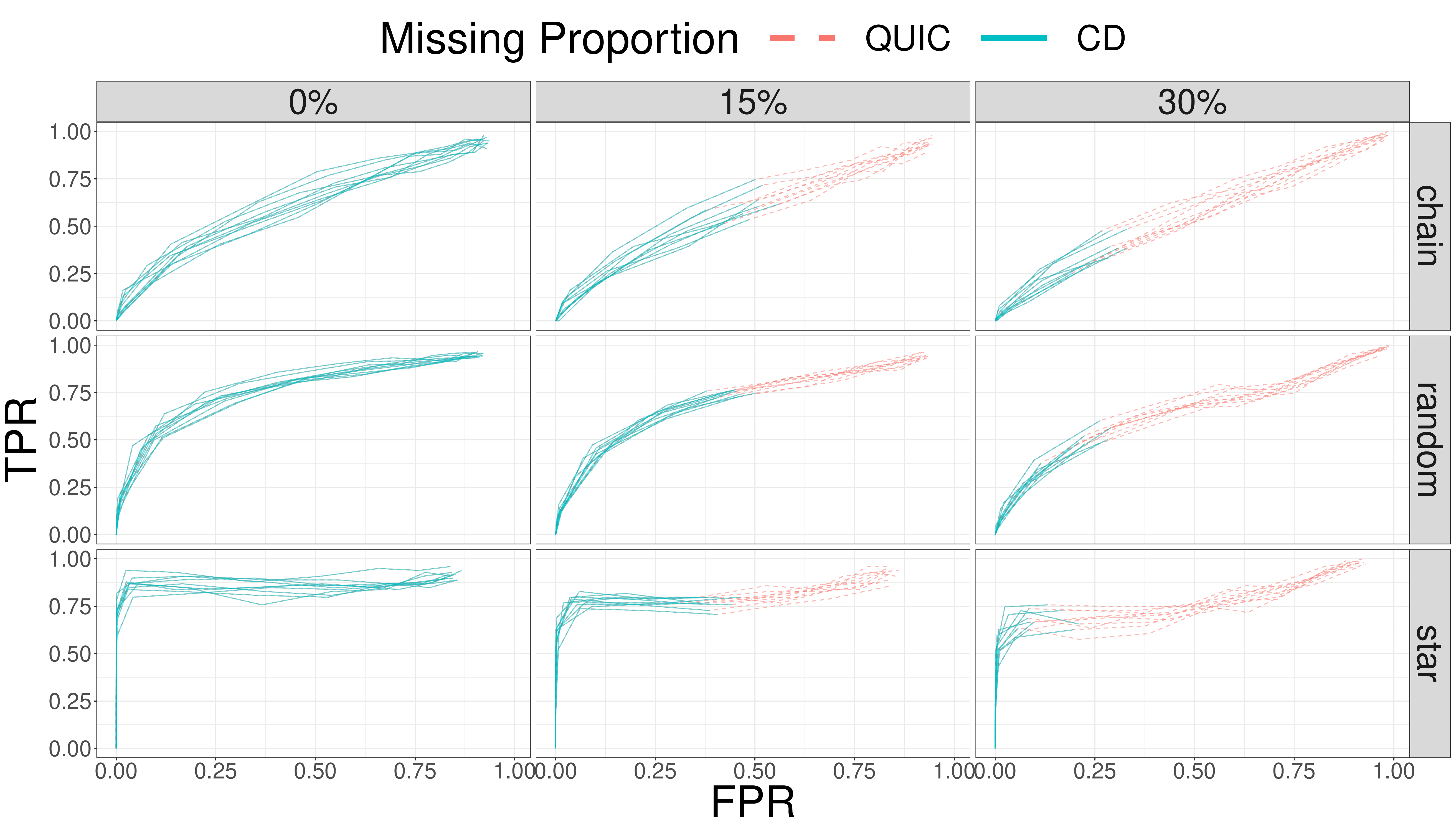}
	%		\end{minipage}
	\caption{Comparison of ROC curves between two different algorithms for solving the graphical lasso using incomplete data. Here, $n=100$, $r=1$, and the dependent missing structure are assumed. The oracle IPW estimator is plugged-in. We randomly generate $10$ data sets.}
	\label{fig:ROC_QUIC_huge}
\end{figure}

\section{Additional analyses and details in Section 5}\label{sec:simul_res}

\subsection{Details of the simulation setting}\label{app:simul_setting}
Recall that we generate Gaussian random vectors $X_i$, $i=1,\ldots, n$, in $\mathbb{R}^p$ with mean vector $0$ and precision matrix $\bOmega = (\omega_{ij}, 1 \le i, j \le p)$ under different pairs of $n=50, 100, 200$ and $p$ satisfying $r(=p/n)=0.2, 1, 2$. The graph structure induced by a precision matrix and the missing structure are described in details below.
\subsubsection*{Graph structure (precision matrix)}
\begin{enumerate}
	%	\item Independence : $\bOmega = \bI$ (\textcolor{red}{NOT INCLUDED IN RESULT YET})
	%	\item Heterogeneous AR(1) :
	%	$$
	%	\bOmega = \text{diag}(\omega_1, \ldots, \omega_p) (\rho^{|i-j|}) \text{diag}(\omega_1, \ldots, \omega_p)
	%	$$
	%	where $\rho=0.5$ and $\omega_k$ is generated from $\text{Unif}(0,5)$ randomly.
	
	\item Chain-structured graph : The edge set $E$ of a graph is defined by the structure of a chain graph. $\omega_{ij}= 0.1$, if $(i,j) \in E$, and $0$, otherwise; $\omega_{ii} = 1$.
	
	\item Star-structured graph :
	The edge set $E$ of a graph is defined by the structure of a star-shaped graph.
	$\omega_{ij}= 0.9 / \sqrt{p - 1}$\footnote{The off-diagonal element $\omega_{ij}$ should be less than $1/\sqrt{p-1}$ to satisfy $\bOmega \succ 0$.}, if $(i,j) \in E$, and $0$, otherwise; $\omega_{ii} = 1$.
	
	\item Erd\"{o}s-R\'{e}nyi random graph : Each off-diagonal component in the upper part of $\bB$ is independently generated, and equals to $0.5$ with probability $\log p / p$ and $0$ otherwise. Then, the lower part of $\bB$ is filled with the transposed upper part. Finally, some positive constant is added to the diagonals, i.e., $\bOmega = \bB + 1.5|\lambda_{min}| ~ \bI$, to satisfy PDness where $\lambda_{min}$ is the smallest eigenvalue of $\bB$.
\end{enumerate}
Every $\bOmega$ is rescaled so that the largest eigenvalue of $\bOmega$ is set as $1$.
%We add $\lambda$ to diagonal elements of $\bOmega$ where $\lambda>0$ is determined to make the condition number of $\bOmega$ be $p$.

\subsubsection*{Missing structure}
Two structures are under consideration to impose missing on data. The first structure is the independent structure where every component of $X_i$ is independently exposed to missing with equal probability;
\begin{equation}\label{eq:missing_ind}
\delta_{ik} \sim \text{Ber}(\pi^{(1)}), \quad k=1, \ldots, p, \text{ independently}
\end{equation}
where $0 < \pi^{(1)} <1$. Another structure is designed to model dependency within missing indicators. We assume missingness in the first half of $p$ components (assume even $p$ here) forces missing values in the other halves. First, we generate $p$ independent missing indicators as before
$$
\tilde{\delta}_{ik} \sim \text{Ber}(\pi^{(2)}), \quad k=1, \ldots, p, \text{ independently},
$$
for $0<\pi^{(2)}<1$. Then, dependent indicators are defined by
$$
\delta_{ik} = \tilde{\delta}_{i k}, \quad \delta_{i,k+p/2}=\min\{\tilde{\delta}_{i k}, \tilde{\delta}_{i, k+p/2}\}, \quad k=1, \ldots, p/2.
$$
Thus, the $(k+p/2)$-th component cannot be observed unless its pair is observed, or $\delta_{ik}=1$ ($k=1, \ldots, p/2$).
An average proportion of missing elements is $1-\pi^{(1)}$ for the independent case and $(1 - \pi^{(2)})(2+\pi^{(2)})/2$ for the dependent case.
Consequently, the proportion of missing denoted by $\alpha$ can be tuned by changing $\pi^{(1)}$ or $\pi^{(2)}$. For example, under the dependent missing structure, for $\alpha=0.3$,
$\pi^{(2)}$ is uniquely determined by solving the quadratic equation 
$$
(1 - \pi^{(2)})(2+\pi^{(2)})/2 = 0.3.
$$
We choose different values $\alpha = 0, 0.15, 0.3$. The case $\alpha=0$ where all samples are completely observed is included as a reference.

\subsubsection*{Estimators}
Based on our experience, the graphical lasso is preferred to the CLIME in estimation of sparse precision matrices since the implemented R packages are either too conservative to find true edges (R package \pkg{fastclime}) or too slow (R package \pkg{clime}). We exploit \pkg{QUIC} algorithm proposed by \cite{Hsieh:2014} to solve the graphical lasso (\ref{eq:glasso}). The grid of a tuning parameter $\lambda\in \Lambda$ is defined adaptively to the plug-in matrix $\widehat{\bSigma}^{plug}$
$$
\Lambda = \Big\{\exp\{\log(\kappa M) - d \log(\kappa) / (T-1)\}: d=0, \ldots, T-1 \Big\},  
$$
where $0<\kappa <1$ and $M = \big|\big|\widehat{\bSigma}^{plug}  - \diag(\widehat{\bSigma}^{plug}) \big|\big|_{max}$. Note that the points in $\Lambda$ are equally spaced in log-scale from $\log(\kappa M)$ to $\log M$ by length of $T$. $\kappa$ is set as $0.1$ and $T$ as $10$.

\subsection{Additional results of the simulation study}\label{app:simul_res_extra}
We investigate the finite sample performance by changing various parameters (e.g. $r=p/n$, missing proportion) in the simulation study. To evaluate estimation accuracy and support recovery of the Gaussian graphical model, different matrix norms and an area under the receiver operating characteristic (ROC) curve are used.

\subsubsection*{Estimation accuracy}
We numerically examine behaviors of the inverse covariance matrix estimated using the IPW estimator as simulation parameters vary. To this end, the Frobenius and spectral norms are used to measure the accuracy of an estimator. We fix the $\lfloor 0.7\:T \rfloor$-th tuning parameter in $\Lambda$ (in an increasing order) to get a single sparse precision matrix, because selection of the tuning parameter is not of our primary interest and our findings stated below do not change much according to the tuning parameter.
%the in $\Lambda$ makes a good compromise between goodness of fit and sparsity.

% % ratio
\begin{figure}[H]
	%	\begin{adjustwidth}{-2.4cm}{}
	\centering
	\includegraphics[page=1,width=0.49\linewidth]{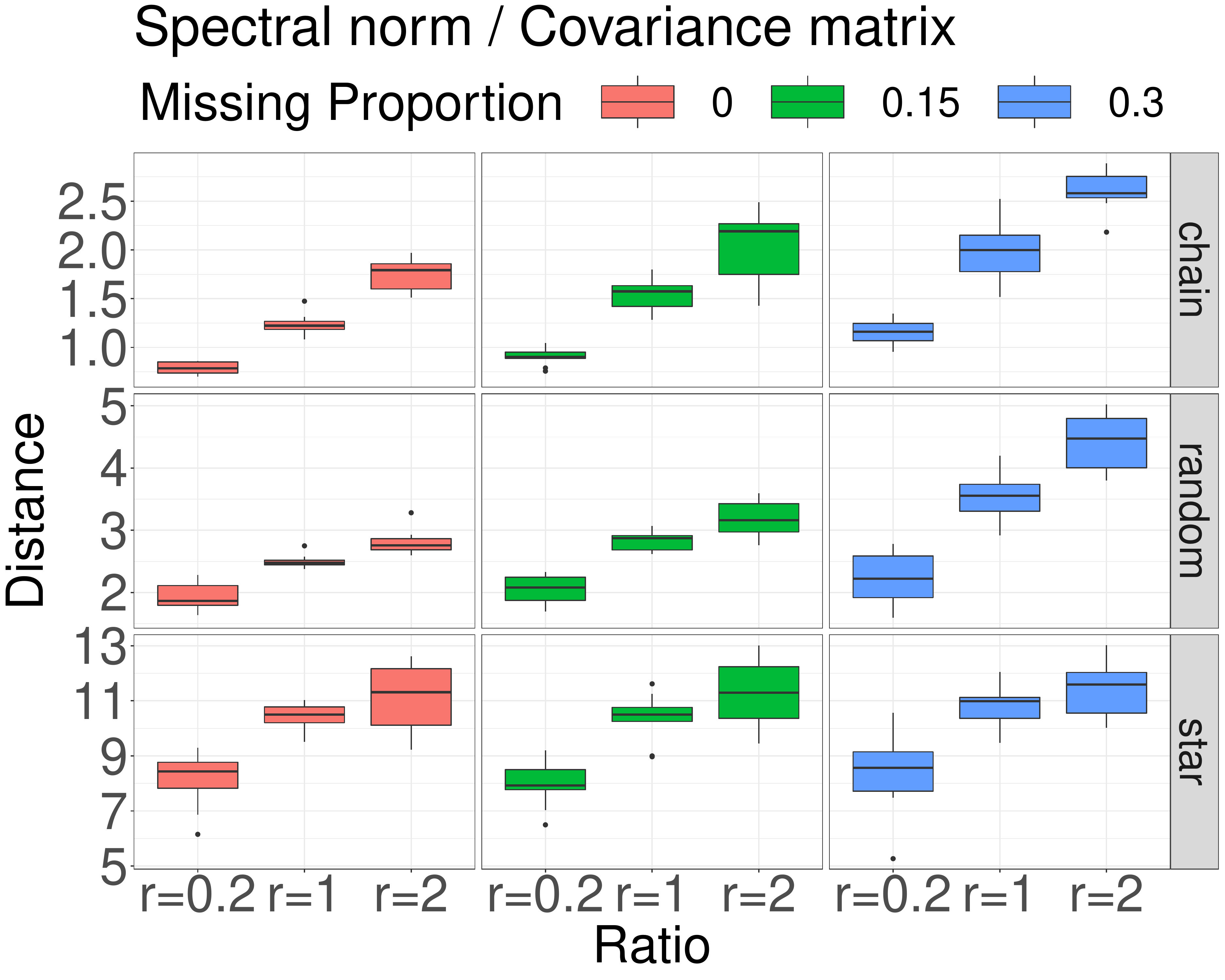}
	\includegraphics[page=1,width=0.49\linewidth]{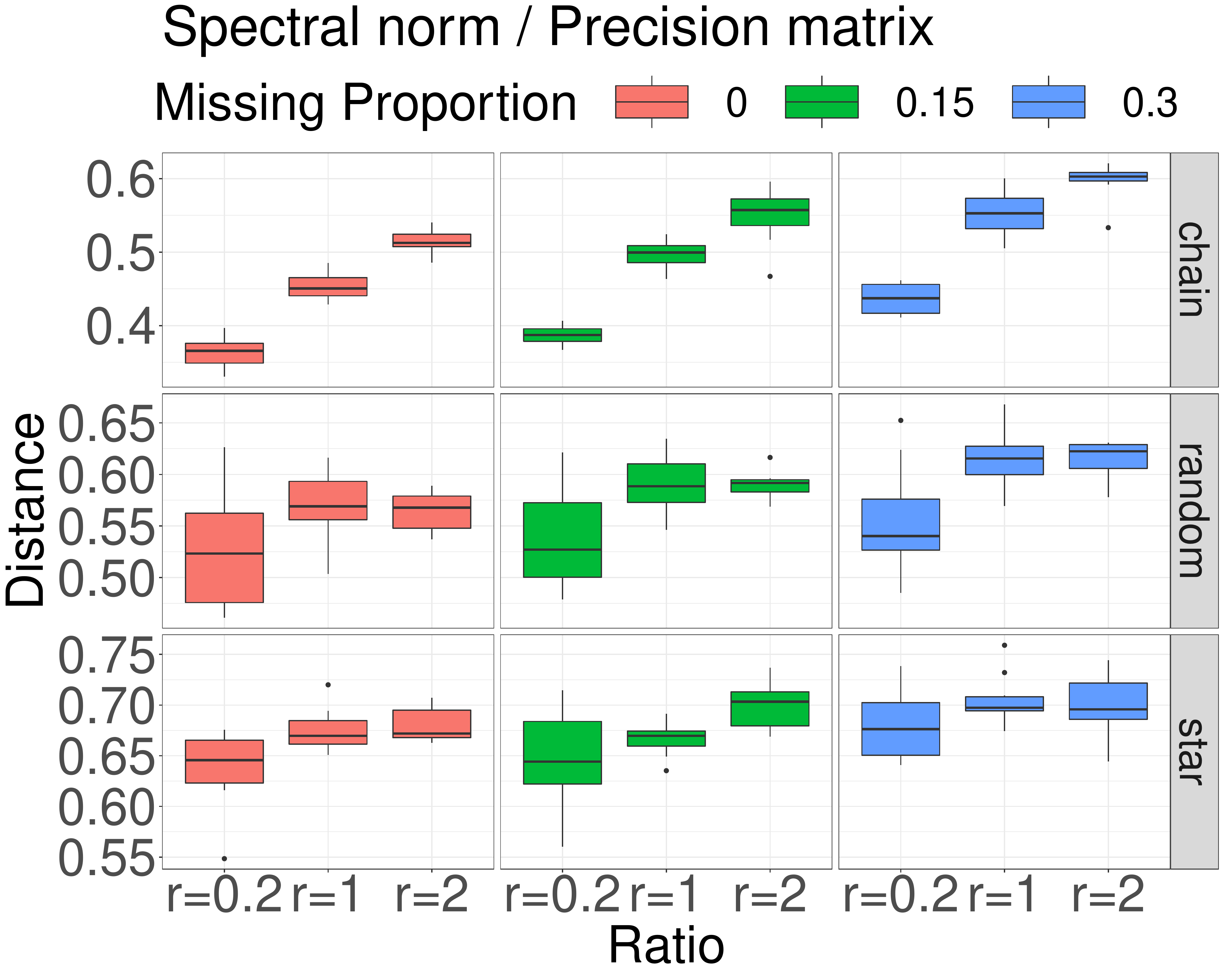}
	\caption{Boxplots of the spectral norm with different ratios $r(=p/n)=0.2, 1, 2$. Here, the dependent missing structure and $n=100$ are assumed. The oracle IPW estimator is plugged-in.
		$||\widehat{\bOmega}^{-1} - \bOmega^{-1}||$ (left) and 		 $||\widehat{\bOmega} - \bOmega||$ (right) are measured.}
	\label{fig:dist_ratio_sp}
	%	\end{adjustwidth}
\end{figure}

\begin{figure}[H]
	%	\begin{adjustwidth}{-2.4cm}{}
	\centering
	\includegraphics[page=1,width=0.49\linewidth]{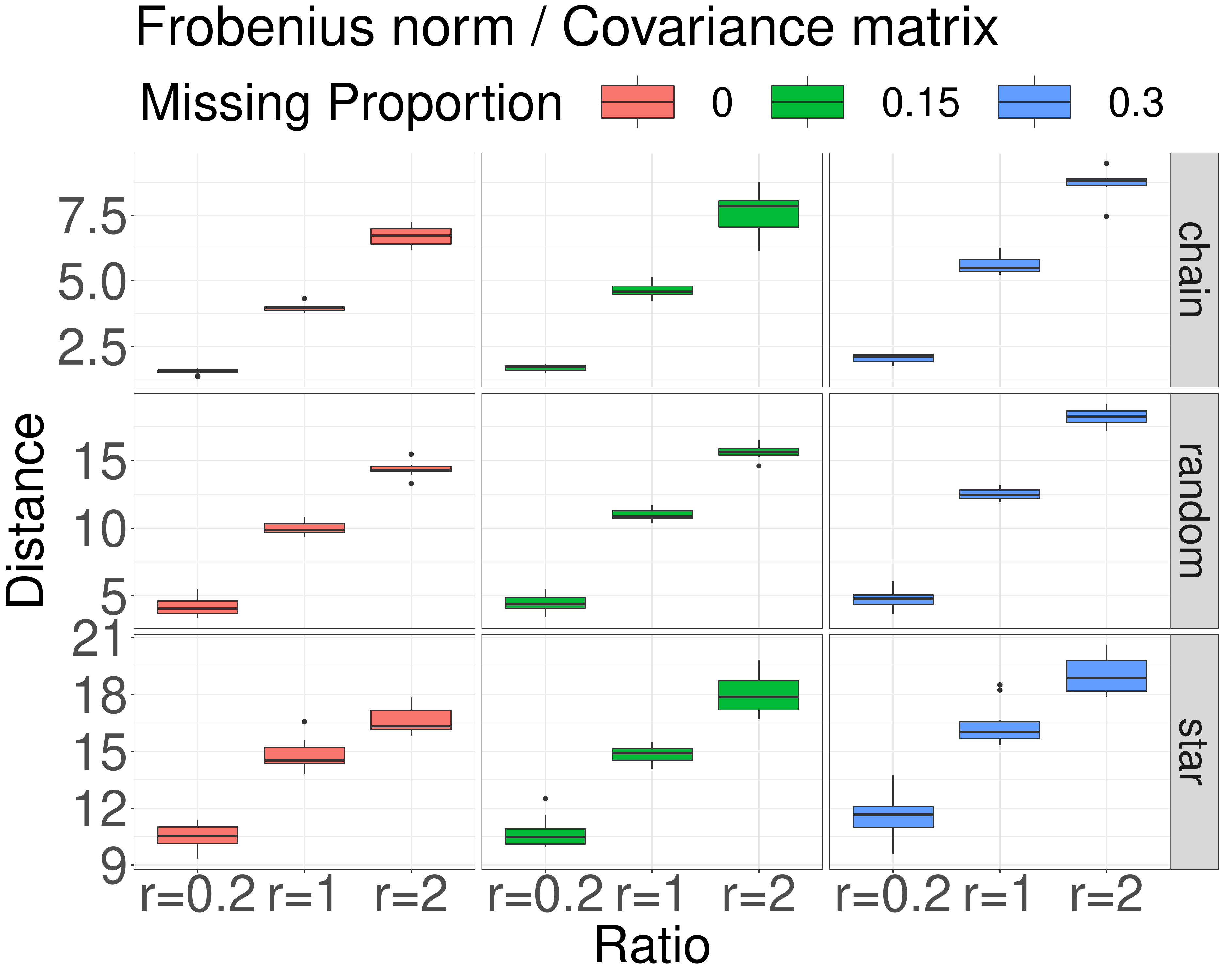}
	\includegraphics[page=1,width=0.49\linewidth]{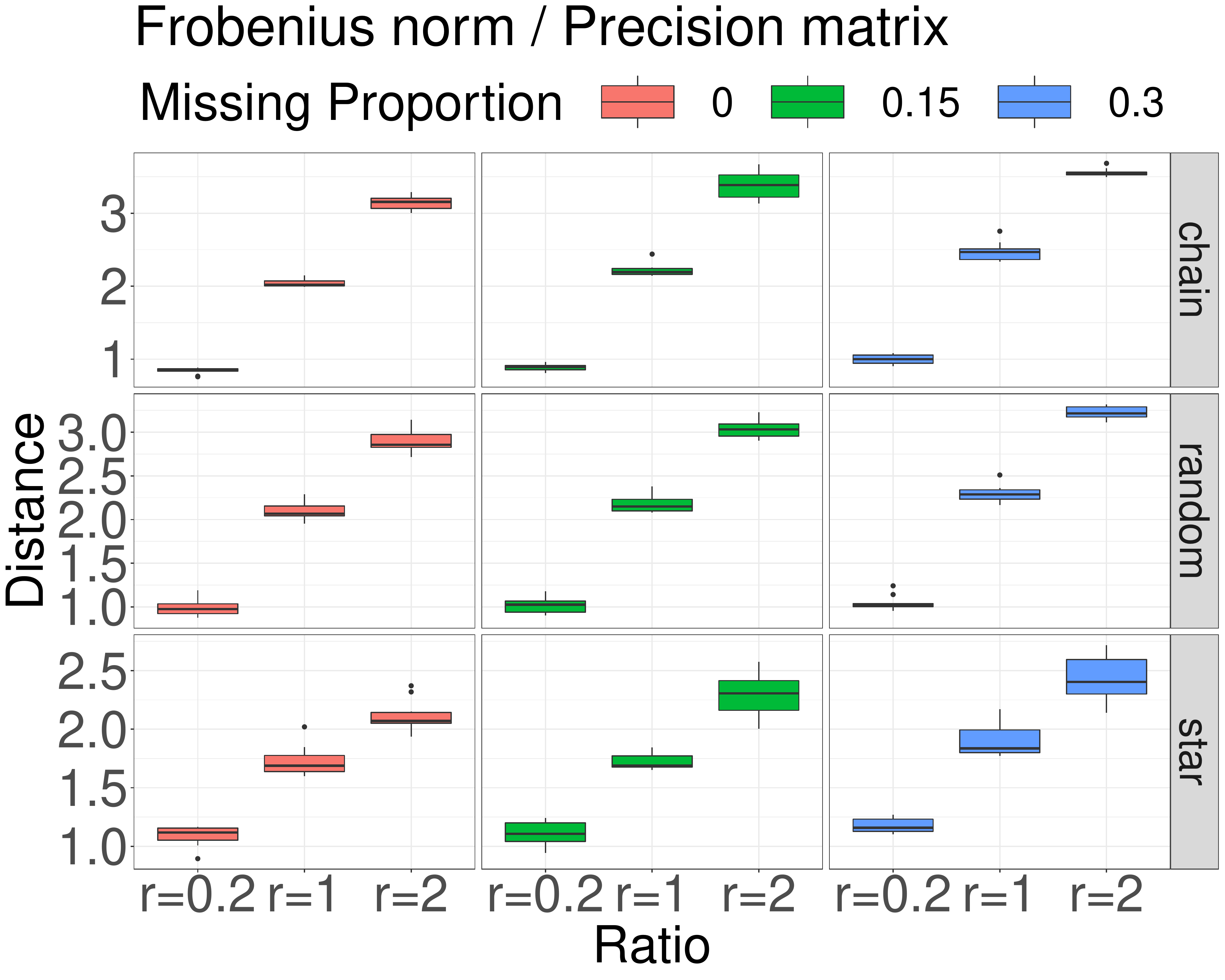}
	\caption{Boxplots of the Frobenius norm with different ratios $r(=p/n)=0.2, 1, 2$. Here, the dependent missing structure and $n=100$ are assumed. The oracle IPW estimator is plugged-in.
		$||\widehat{\bOmega}^{-1} - \bOmega^{-1}||$ (left) and 		 $||\widehat{\bOmega} - \bOmega||$ (right) are measured.}
	\label{fig:dist_ratio_fr}
	%	\end{adjustwidth}
\end{figure}
\noindent
% ratio
Figures \ref{fig:dist_ratio_sp} and \ref{fig:dist_ratio_fr} show
that the ratio of the sample size and dimension is one of the key factors that determines the magnitude of estimation error. It is uniformly observed that larger size of a precision matrix is more difficult to estimate, but the degree of difficulty depends on the shape of the true graphs (or precision matrix).

% type of plug-in matrix
\begin{figure}[H]
	%	\begin{adjustwidth}{-2.4cm}{}
	\centering
	\includegraphics[page=1,width=0.49\linewidth]{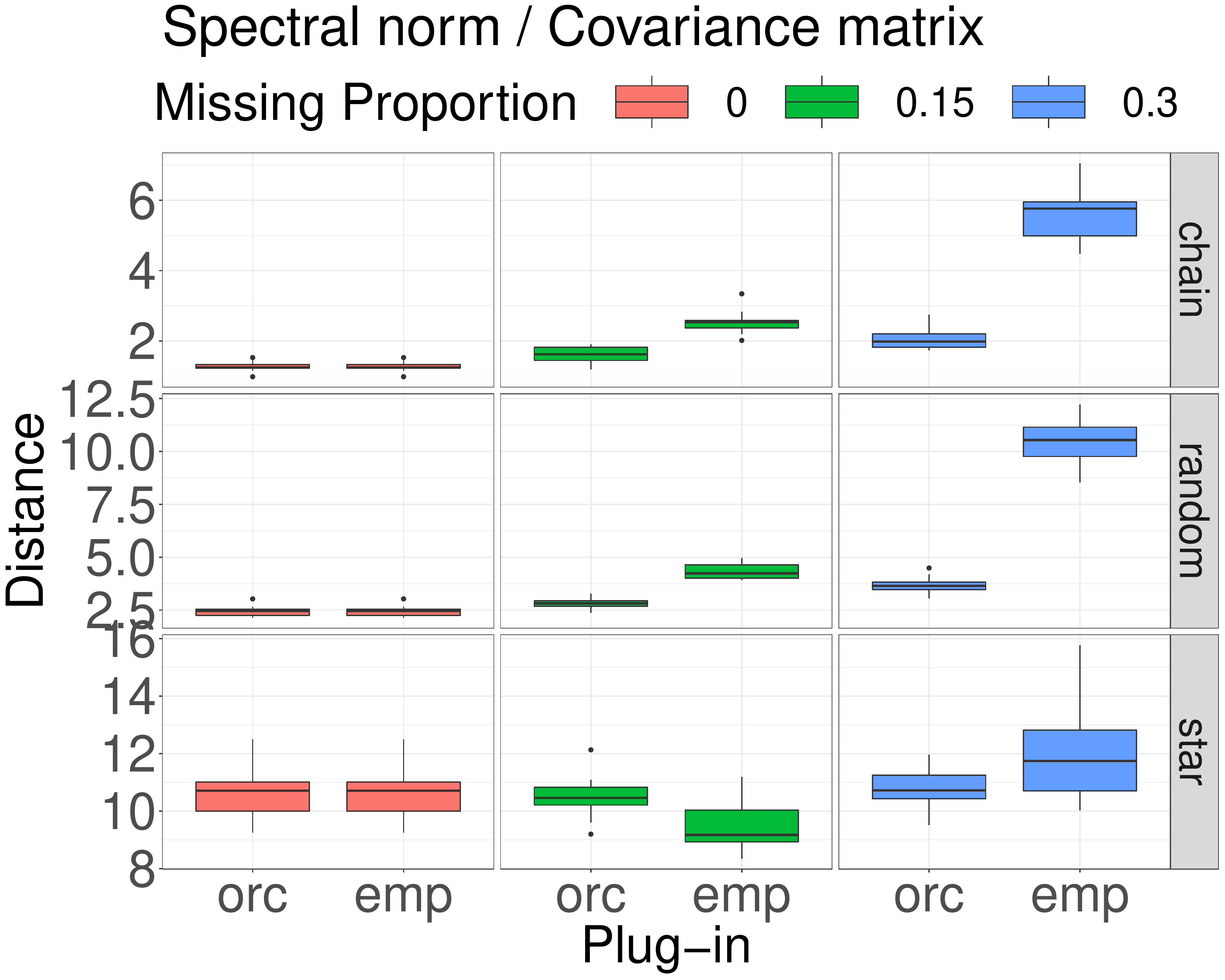}
	\includegraphics[page=1,width=0.49\linewidth]{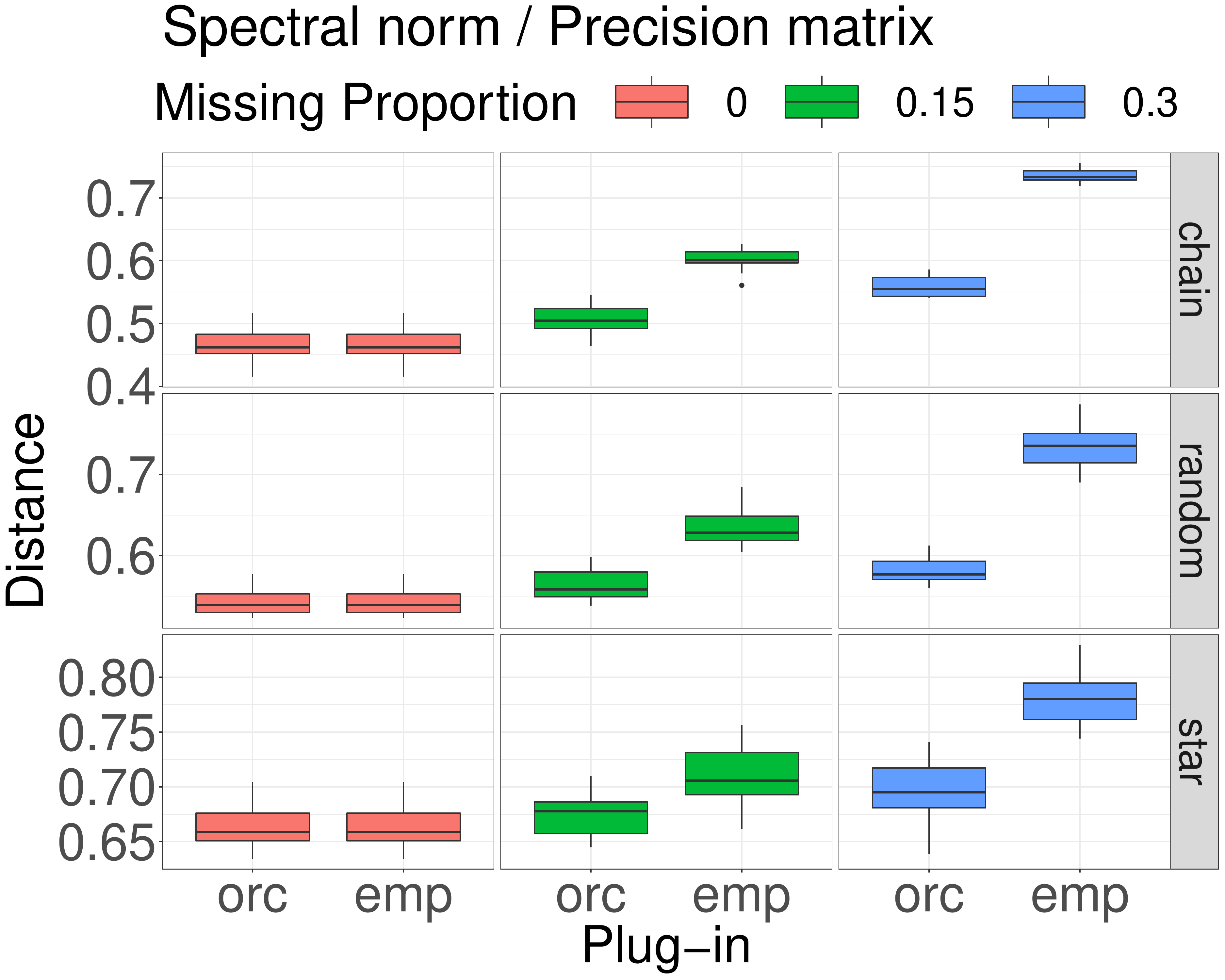}
	\caption{Boxplots of the spectral norm with different plug-in estimators (``emp'' and ``orc''). Here, the dependent missing structure, $n=100$ and $r=1$ are assumed.
		$||\widehat{\bOmega}^{-1} - \bOmega^{-1}||$ (left) and 		 $||\widehat{\bOmega} - \bOmega||$ (right) are measured.}
	\label{fig:dist_type_plugin_sp}
	%	\end{adjustwidth}
\end{figure}
\noindent
\begin{figure}[H]
	%	\begin{adjustwidth}{-2.4cm}{}
	\centering
	\includegraphics[page=1,width=0.49\linewidth]{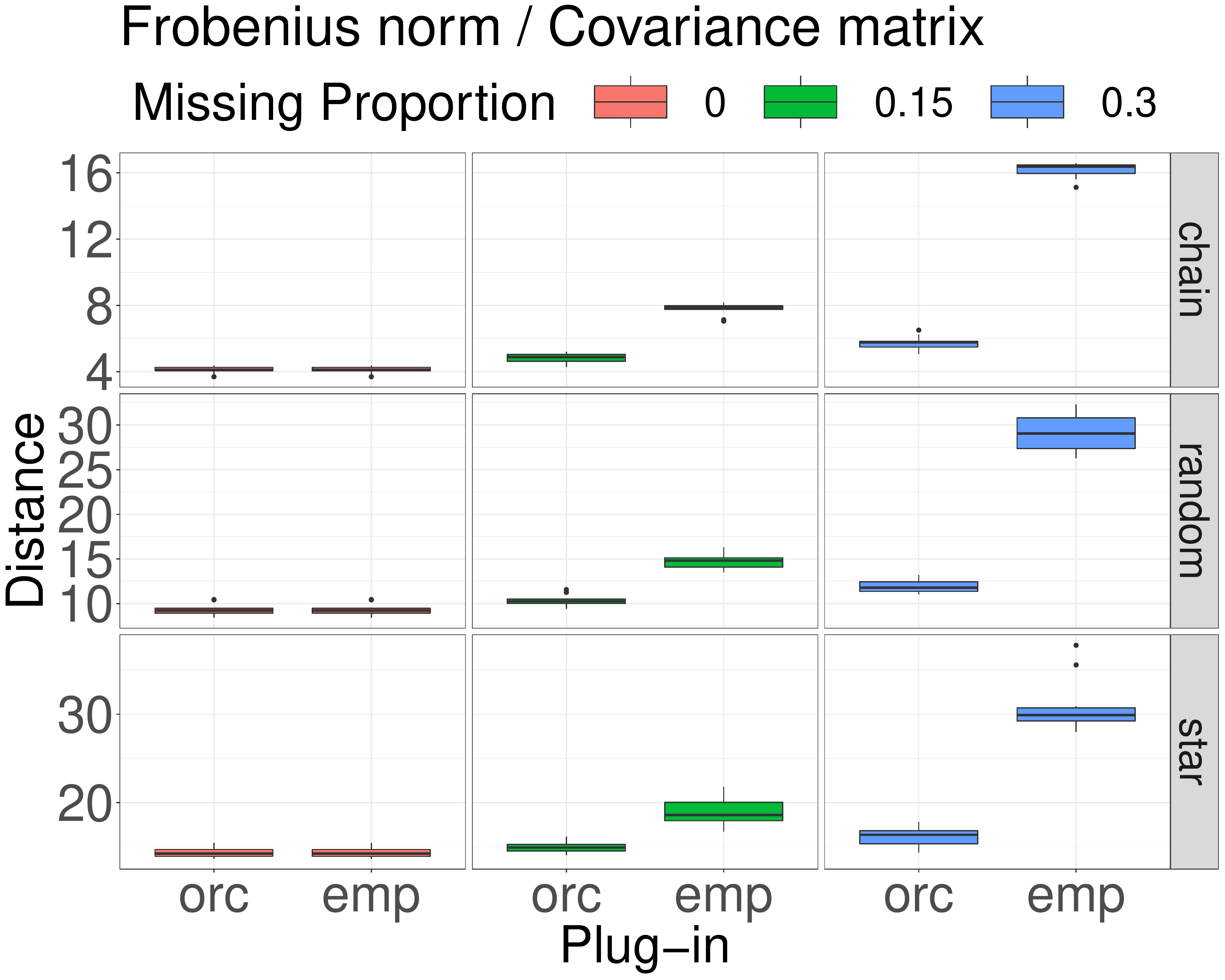}
	\includegraphics[page=1,width=0.49\linewidth]{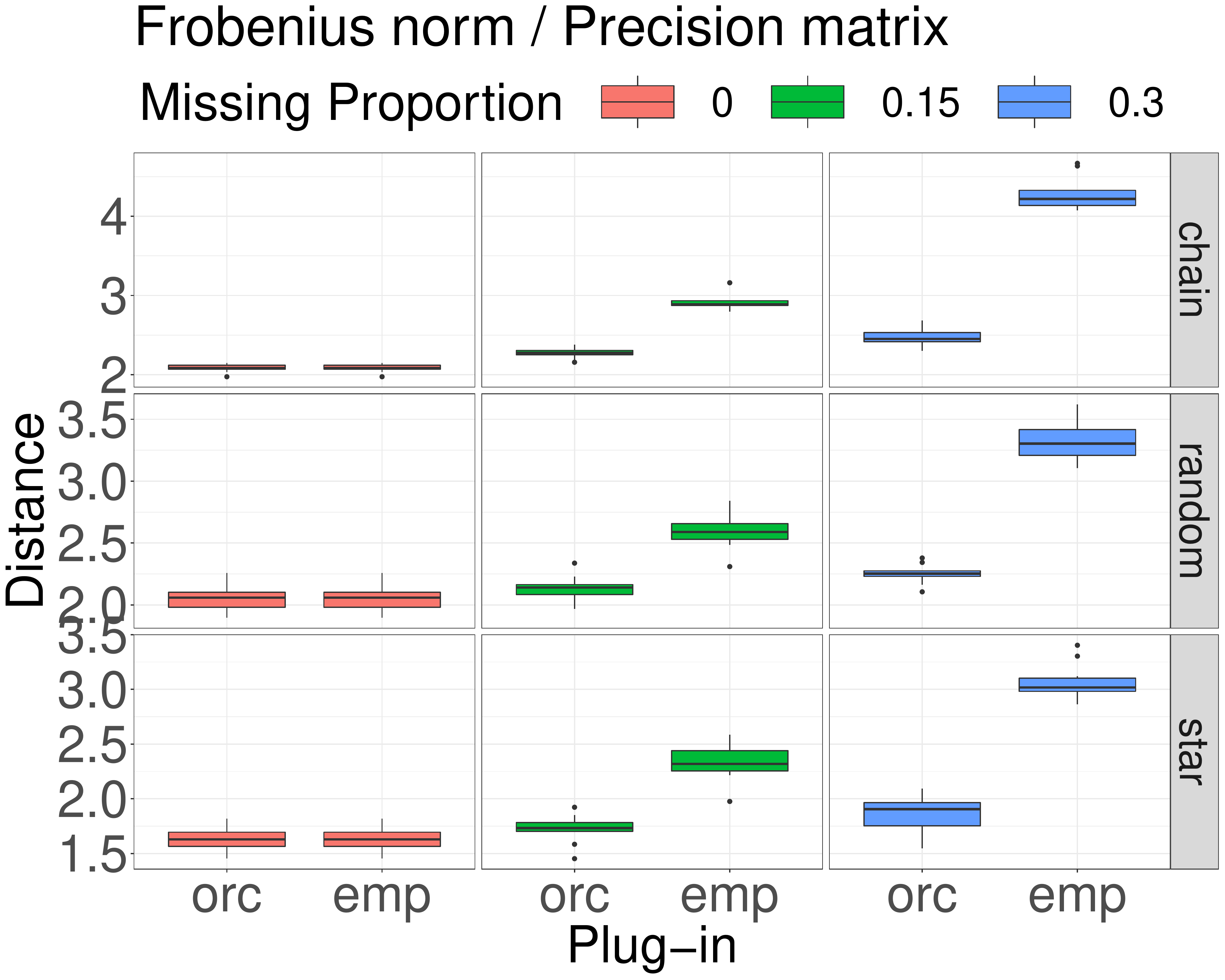}
	\caption{Boxplots of the Frobenius norm with different plug-in estimators (``emp'' and ``orc''). Here, the dependent missing structure, $n=100$ and $r=1$ are assumed.
		$||\widehat{\bOmega}^{-1} - \bOmega^{-1}||$ (left) and 		 $||\widehat{\bOmega} - \bOmega||$ (right) are measured.}
	\label{fig:dist_type_plugin_fr}
	%	\end{adjustwidth}
\end{figure}
\noindent
Figures \ref{fig:dist_type_plugin_sp} and \ref{fig:dist_type_plugin_fr} compare the performance of the two plug-in matrices.
When complete data is available, no adjustment for missing is needed so that there is no difference in errors (see the leftmost red boxplots in each sub-figure). If missing occurs in data, the precision matrix estimator based on the oracle IPW estimator is closer to the true matrix (either $\bSigma$ or $\bOmega$), and the extent is more evident as the missing proportion $\alpha$ increases.

% % missing structure
\begin{figure}[H]
	%	\begin{adjustwidth}{-2.4cm}{}
	\centering
	\includegraphics[page=1,width=0.49\linewidth]{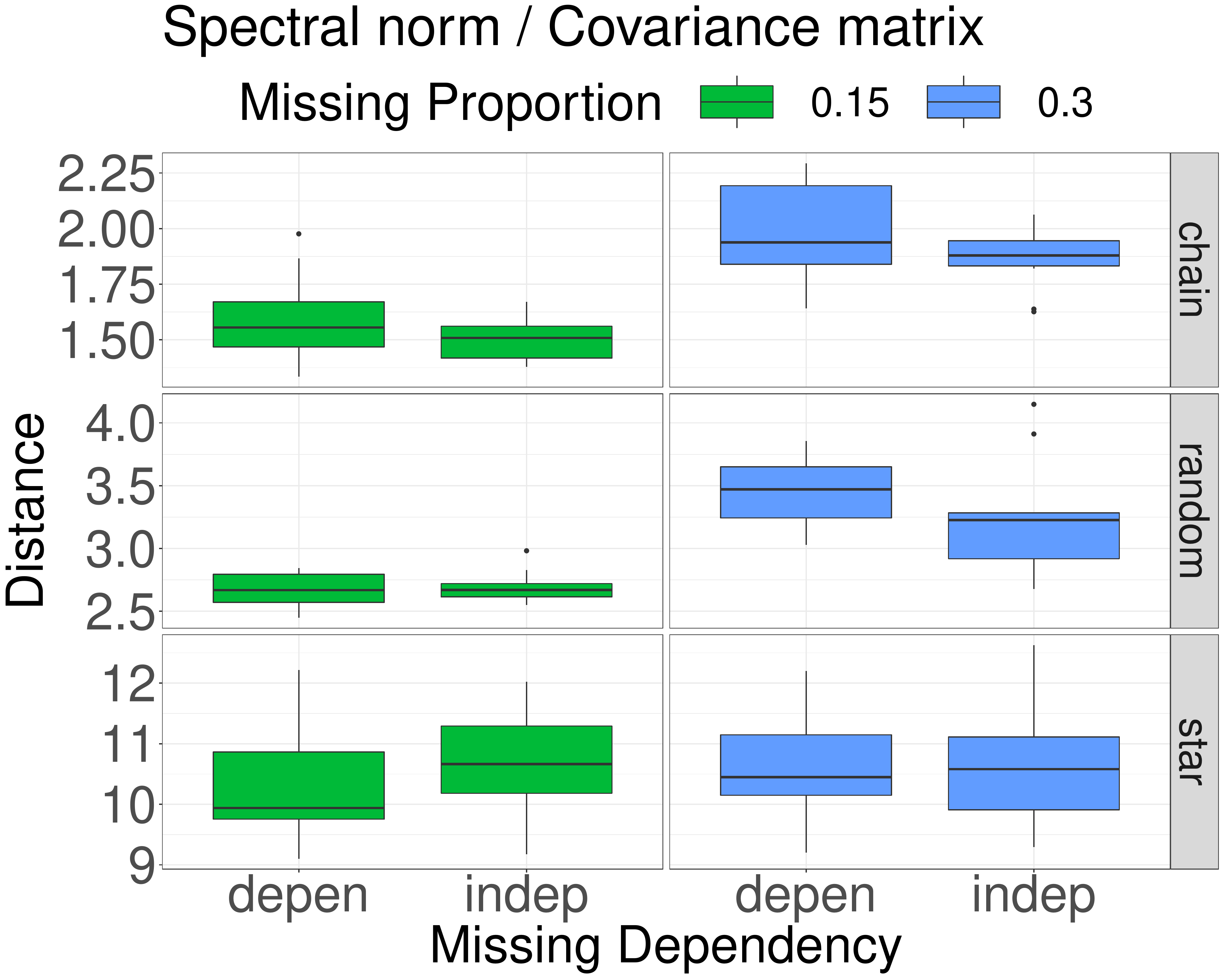}
	\includegraphics[page=1,width=0.49\linewidth]{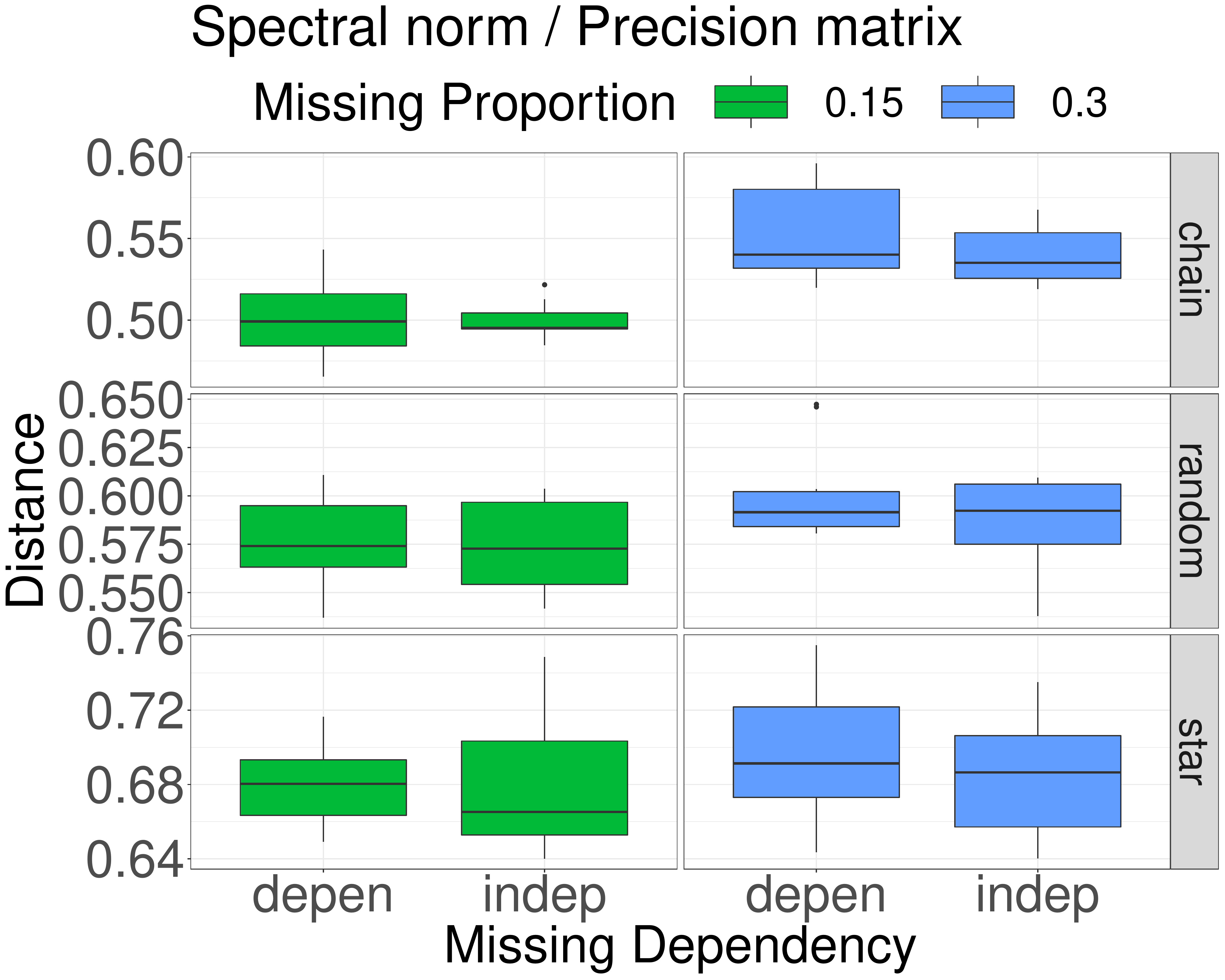}
	\caption{Boxplots of the spectral norm with different missing structures (``depen'' and ``indep''). Here, $n=100$ and $r=1$ are assumed. The oracle IPW estimator is plugged-in.
		$||\widehat{\bOmega}^{-1} - \bOmega^{-1}||$ (left) and 		 $||\widehat{\bOmega} - \bOmega||$ (right) are measured.}
	\label{fig:dist_type_miss_sp}
	%	\end{adjustwidth}
\end{figure}
\begin{figure}[H]
	%	\begin{adjustwidth}{-2.4cm}{}
	\centering
	\includegraphics[page=1,width=0.49\linewidth]{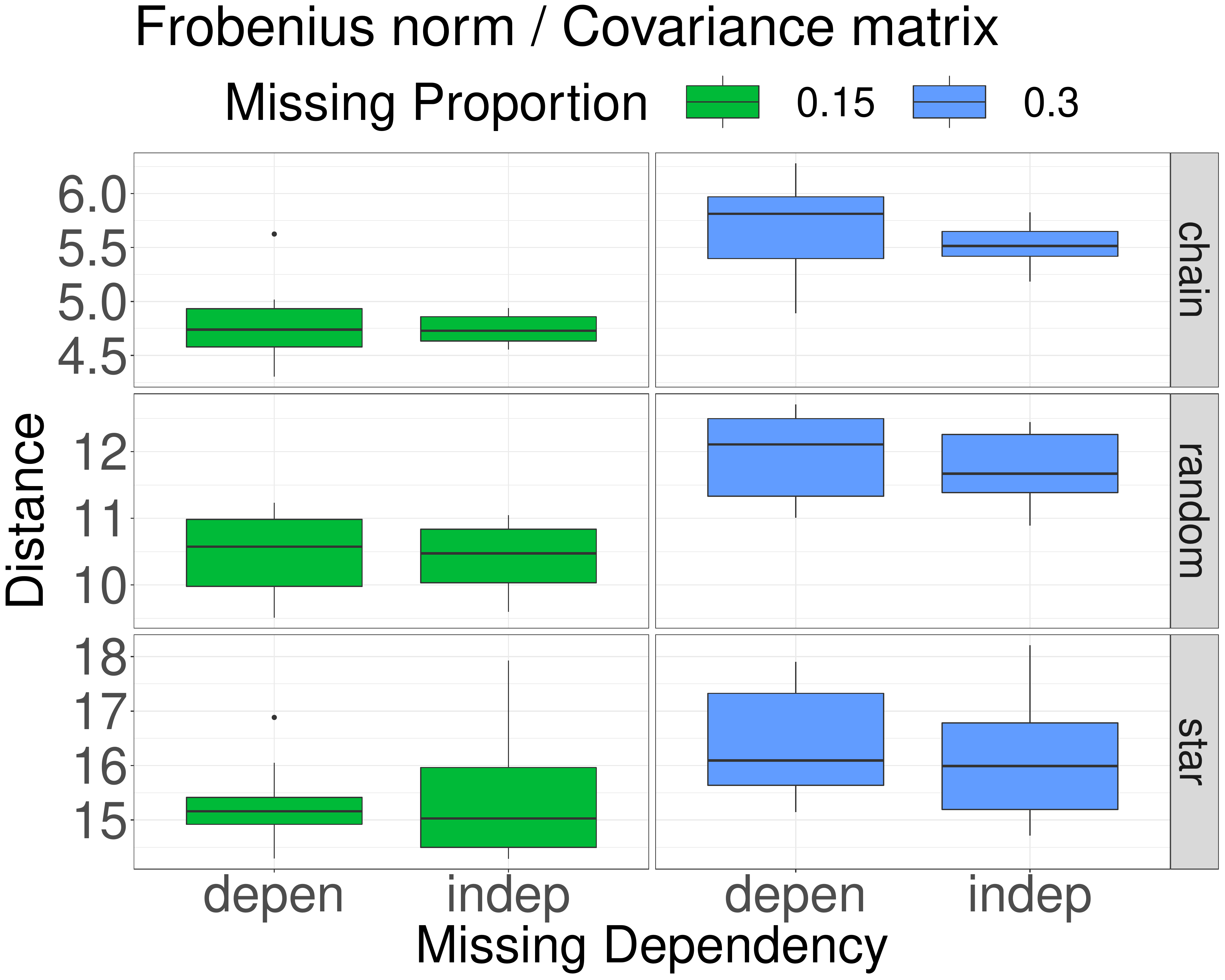}
	\includegraphics[page=1,width=0.49\linewidth]{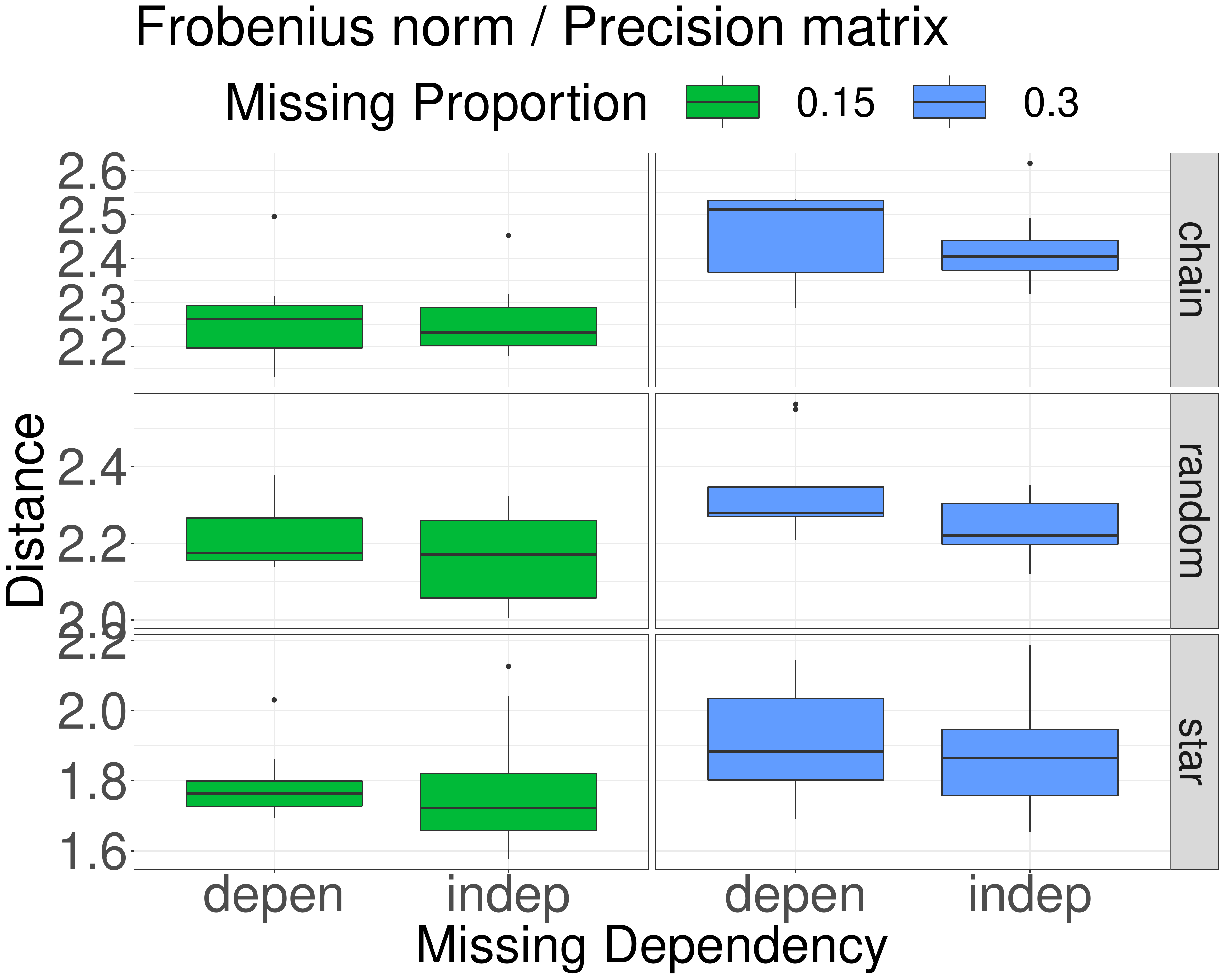}
	\caption{Boxplots of the Frobenius norm with different missing structures (``depen'' and ``indep''). Here, $n=100$ and $r=1$ are assumed. The oracle IPW estimator is plugged-in.
		$||\widehat{\bOmega}^{-1} - \bOmega^{-1}||$ (left) and 		 $||\widehat{\bOmega} - \bOmega||$ (right) are measured.}
	\label{fig:dist_type_miss_fr}
	%	\end{adjustwidth}
\end{figure}
Figures \ref{fig:dist_type_miss_sp} and \ref{fig:dist_type_miss_fr} imply that dependency in missing degrades estimation accuracy, as the missing proportion is set at the same level in both missing structures. We do not show the results when using complete data (i.e., $\alpha=0$) since the two missing structures are the same by definition.

%Finally, it is uniformly observed in all simulations that estimation is less accurate when there is more loss in observations.

\subsubsection*{Support recovery}
We investigate the support recovery of the Gaussian graphical model using the ROC curve. It is observed that the ROC curves end at different false positive rate (FPR) values, especially when different missing proportions are assumed (see Figure \ref{fig:ROC_example}).
\begin{figure}[H]	
	\centering
	%		\begin{minipage}{0.45\linewidth}
	\includegraphics[width=0.6\linewidth]{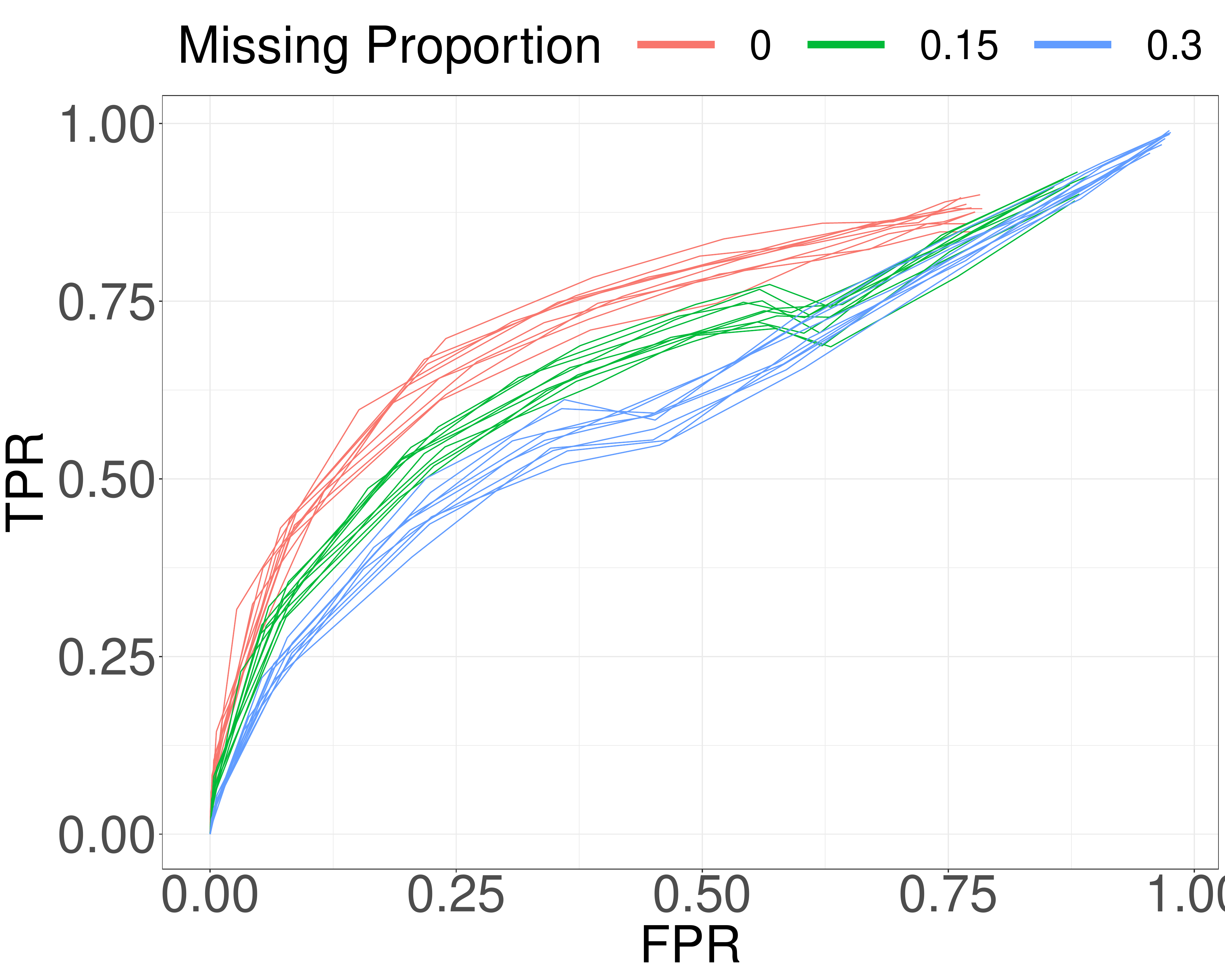}
	%		\end{minipage}
	\caption{The ROC curves according to different missing proportions with $10$ times of repetition. Here, $n=100$, $r=1$, a random graph structure, and the dependent missing structure are assumed. The oracle IPW estimator is plugged-in.}
	\label{fig:ROC_example}
\end{figure}
\noindent
Thus, it is not fair to directly compare an area under the curve (AUC) because the maximum value of AUC depends on the endpoint (largest value) of FPR and thus cannot reach $1$ if the endpoint is less than $1$. Instead, we use the rescaled partial AUC (pAUC) proposed by \cite{Walter:2005}. The pAUC rescales the AUC by the largest FPR in the ROC curve (see \cite{Walter:2005} for more details). Then, the rescaled AUCs from different curves that end at different FPR values have the same range $[0, 1]$.
\begin{figure}[H]
	%	\begin{adjustwidth}{-2.4cm}{}
	\centering
	\includegraphics[page=1,width=0.6\linewidth]{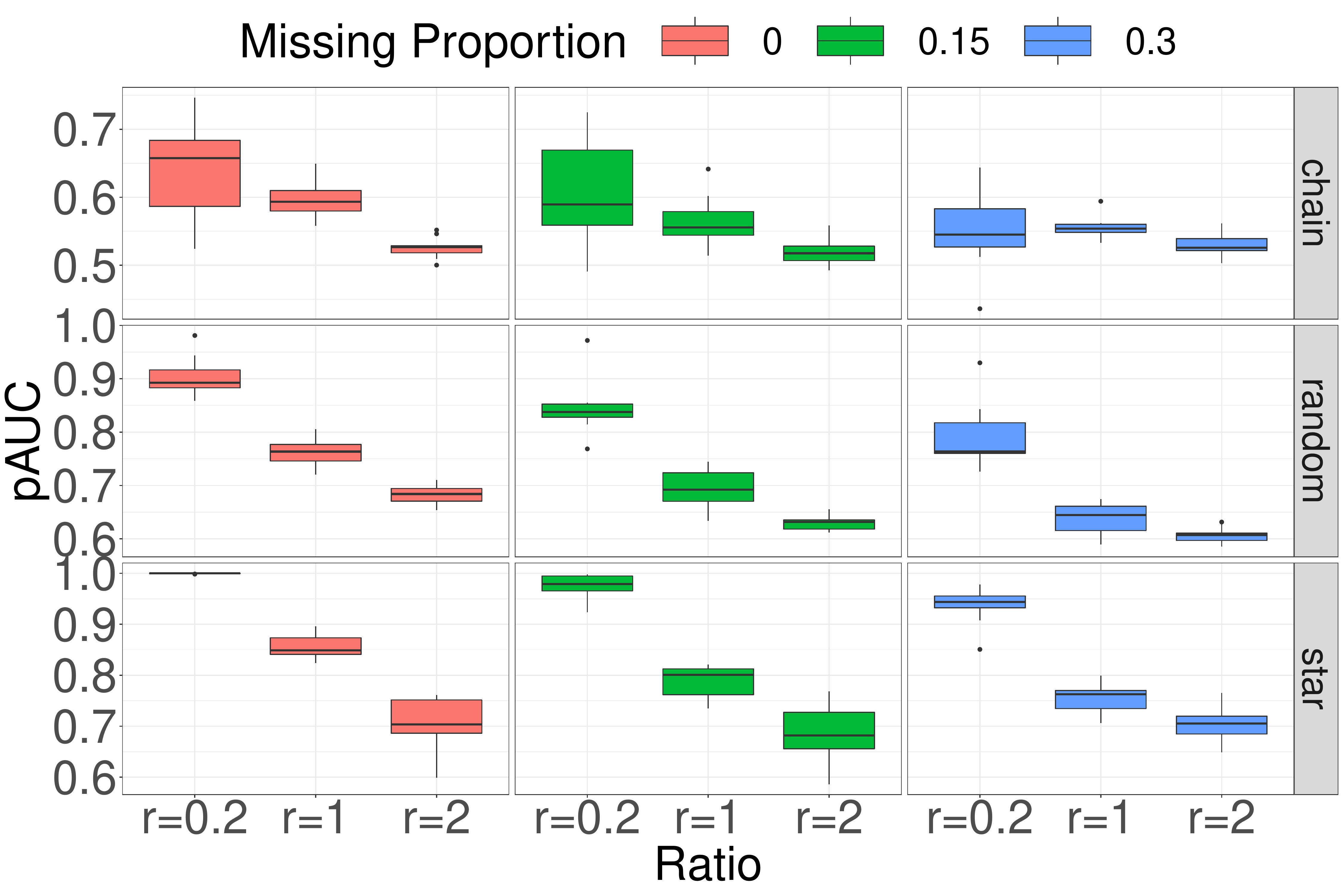}
	\includegraphics[page=1,width=0.5\linewidth]{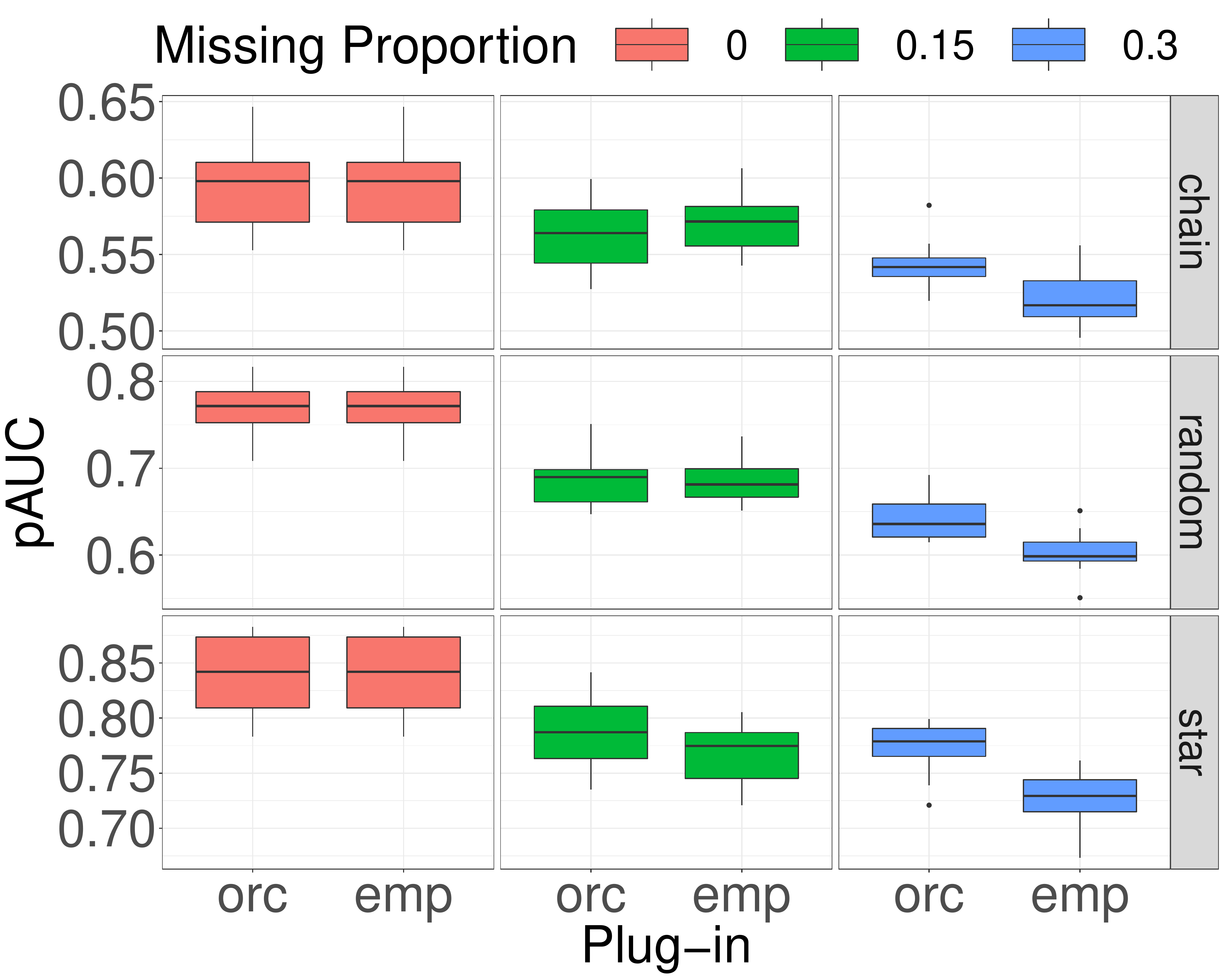}
	\includegraphics[page=1,width=0.49\linewidth]{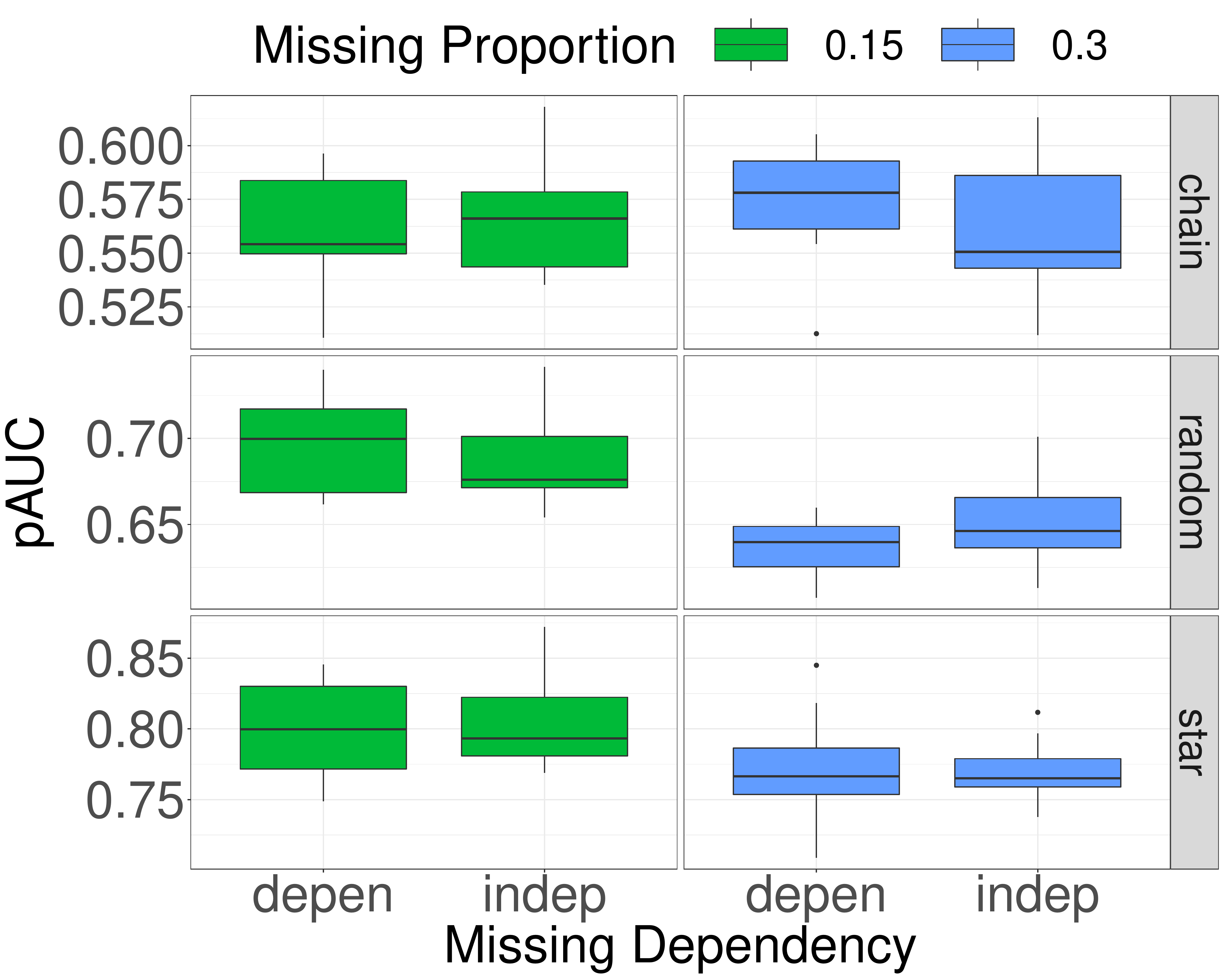}
	\caption{(Top) Boxplots of the pAUC with different ratios $r(=p/n)=0.2, 1, 2$. Here, $n=100$ and the dependent missing structure are assumed. The oracle IPW estimator is plugged-in.
		(Bottom left) Boxplots of the pAUC for support recovery with different plug-in estimators. Here, $n=100$, $r=2$, and the dependent missing structure are assumed.
		(Bottom right) Boxplots of the pAUC for support recovery with different missing structures (``depen'' and ``indep''). Here, $n=100$ and $r=1$ are assumed. The oracle IPW estimator is plugged-in.}
	\label{fig:AUC_ratio}
	%	\end{adjustwidth}
\end{figure}
\noindent
Figure \ref{fig:AUC_ratio} shows the results of the pAUC as the simulation parameters are varying. Considering a large value of the pAUC implies better performance in the support recovery, we have similar interpretations based on the given results as before.

\subsubsection*{Unknown mean}\label{app:simul_res_extra_unknown_mean}
Continuing from Section \ref{sec:simul_unknown_mean}, we measure the performance of the graphical lasso estimator obtained by plugging-in $\widehat{\bSigma}^{IPW\mu}$ in (\ref{eq:glasso}). $||\widehat{\bOmega} - \bOmega||$ is computed by the Frobenius and spectral norms and its boxplots are given in Figure \ref{fig:dist_prec_unknown_mean}.

\begin{figure}[H]
	%	\begin{adjustwidth}{-2.4cm}{}
	\centering
	\includegraphics[page=1,width=0.49\linewidth]{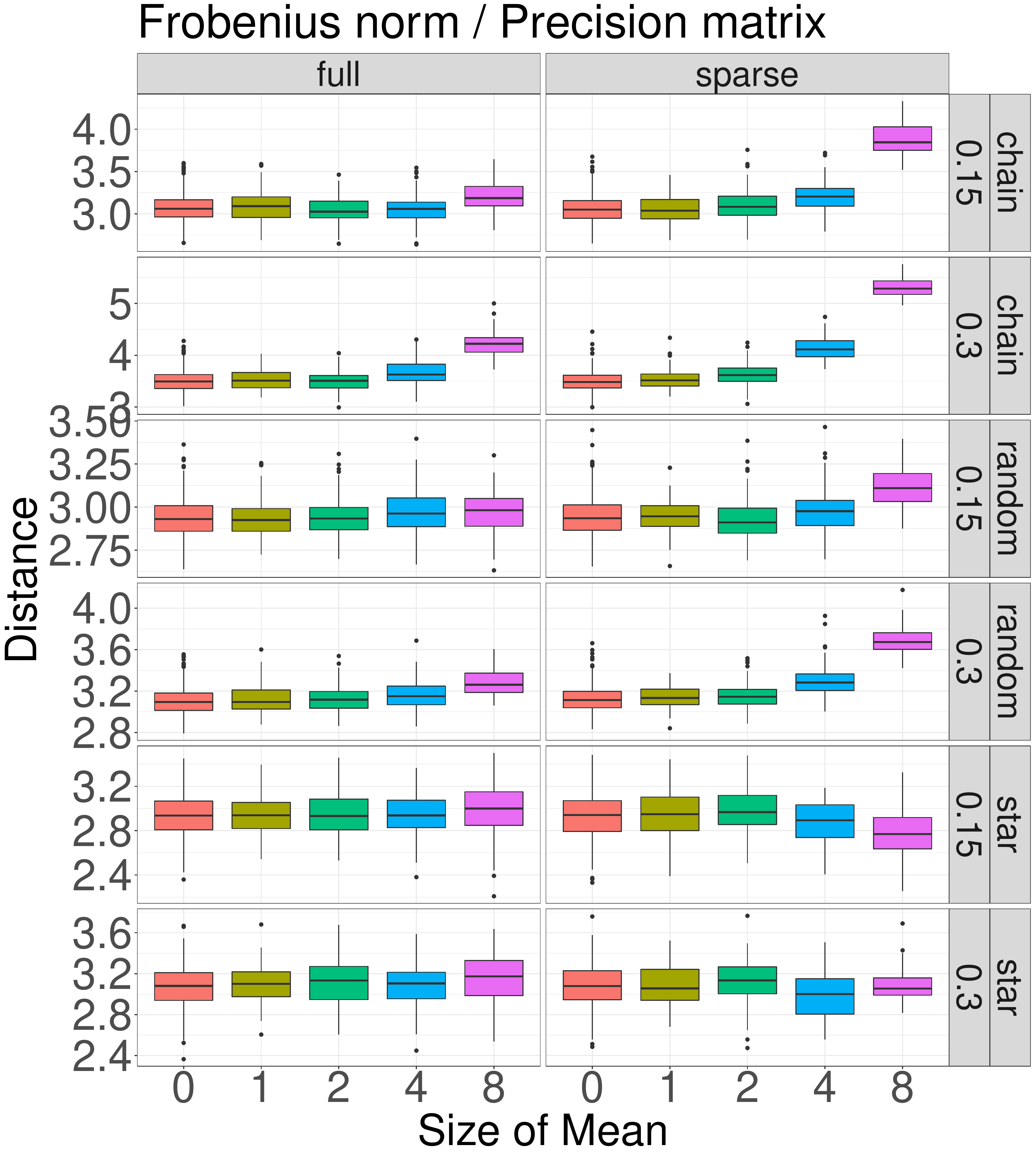}
	\includegraphics[page=1,width=0.49\linewidth]{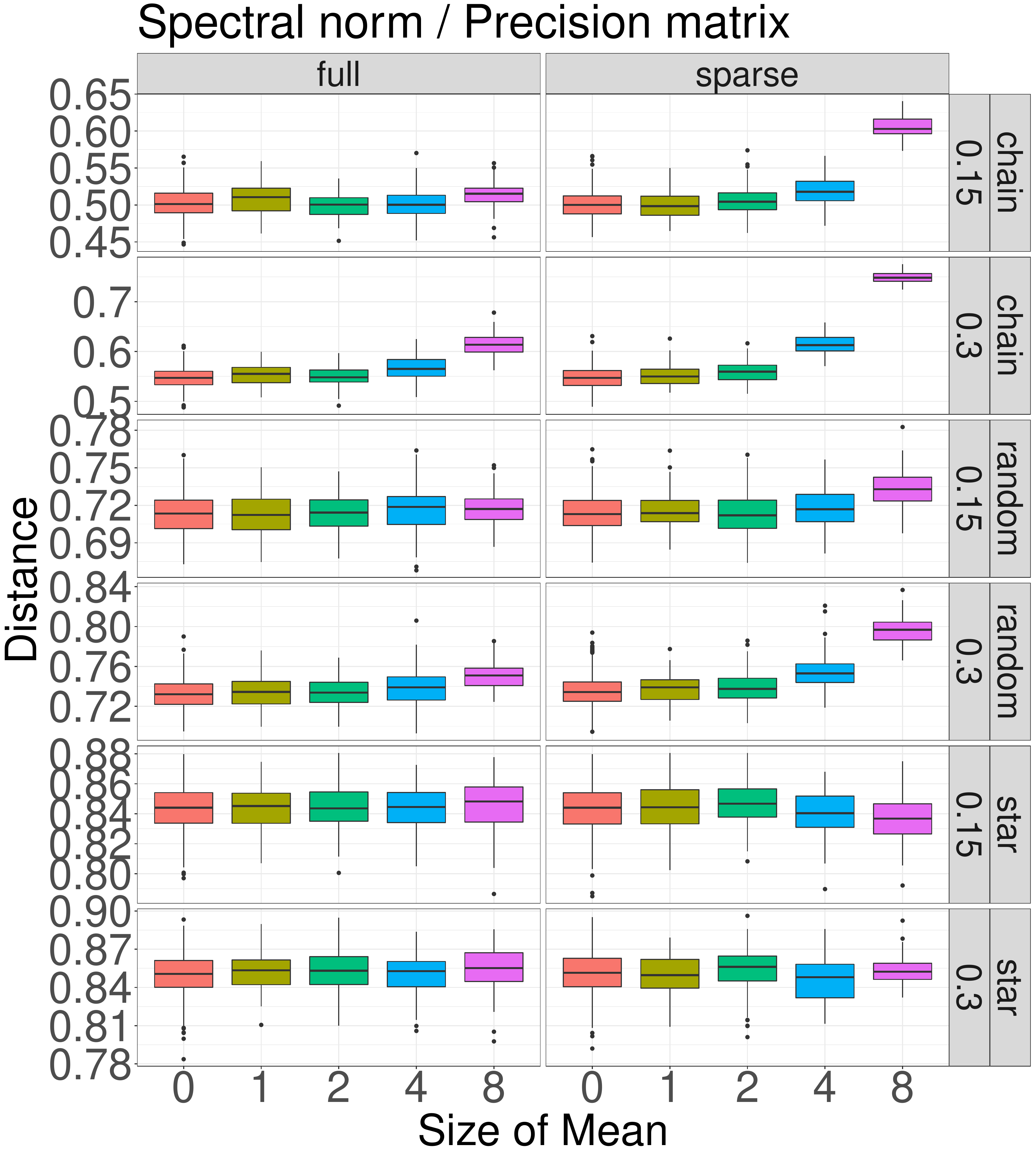}
	\caption{Boxplots of the Frobenius (left) spectral (right) norms $||\widehat{\bOmega} - \bOmega||$ with different magnitudes of an unknown mean vector. Here, $n=100$, $r=1$, and the dependent missing structure are assumed. We randomly generate $100$ data sets.}
	\label{fig:dist_prec_unknown_mean}
	%	\end{adjustwidth}
\end{figure}

\section{Details in Section 6}\label{app:realdata}
We use the riboflavin data available from the R package \pkg{hdi}, where 4088 gene expressions are observed across 71 samples. Each variable is log-transformed and then centered. We select $1000$ genes with the largest empirical variances for the sake of simplicity. As in the previous analyses, the QUIC algorithm is used to solve the graphical lasso.

With the complete data set, we solve the graphical lasso  (\ref{eq:glasso}) with a fixed $\lambda$ and set the obtained estimate $\bOmega_{\lambda}$ as the ground truth precision matrix. We generate three different models with $\lambda_1< \lambda_2 < \lambda_3$. Note that the estimated precision matrix with a smaller tuning parameter (e.g. $\lambda_1$) gives a denser true model that has a precision matrix with more non-zero elements. We also consider another ground-truth precision matrix with an optimal tuning parameter that is chosen by the cross-validation procedure, following 
\cite{Kolar:2012}. Let an index set of $n$ samples split into $K$ folds $\{G_k\}_{k=1}^K$ of equal size. Without samples in the $k$-th fold, we estimate the precision matrix at a fixed $\lambda$, denoted by $\bOmega_\lambda^{(k)}$. We finally choose $\lambda_{CV}$ among a grid of $\lambda$'s that minimizes the cross-validated (negative) log-likelihood function below;
%\begin{equation}\label{eq:cv_glasso}
$$
CV(\lambda) = \sum_{k=1}^K \sum_{i \in G_k} \Big\{ \log | \bOmega_\lambda^{(k)} | + X_{i}^{\rm T} \bOmega_\lambda^{(k)} X_{i} \Big\}.
$$
%\end{equation}
We let $\bOmega_{CV} = \bOmega_{\lambda_{CV}}$ the precision matrix at this level of the optimal sparsity $\lambda_{CV}$. It turns out $\lambda_{CV}$ is close to, but slightly smaller than the smallest tuning parameter $\lambda_1$. The four precision matrix models have $36,170$ ($\lambda_1$), $5,860$ ($\lambda_2$), $14$ ($\lambda_3$), $35,630$ ($\lambda_{CV}$) non-zero elements (except diagonals) in each.

We impose missing values on the complete data matrix in a similar manner described in Section \ref{sec:simul_res}. For this analysis, we assume the independent missing structure and note that results do not alter significantly using the dependent structure. To estimate $\bOmega_{\lambda_i}$, we solve the graphical lasso (\ref{eq:glasso}) using the incomplete data with the tuning parameter fixed at $\lambda_i$. Since the optimality of the tuning parameter can vary as different data is available due to missing, the cross-validation procedure is separately performed, instead of using the same $\lambda_{CV}$ to estimate $\bOmega_{CV}$. Let $\widehat{\bOmega}_\lambda^{(k)}$ be the solution with the tuning parameter $\lambda$ without the $k$-th fold of incomplete data. Then, the (cross-validated) log-likelihood is computed over observed data as follows;
$$
CV_{mis}(\lambda) = \sum_{k=1}^K \sum_{i \in G_k} \Big\{ \log | (\bQ_i^{(k)})^{-1} | + X_{i,obs}^{\rm T} (\bQ_i^{(k)})^{-1} X_{i,obs} \Big\}
$$
where $\bQ_i^{(k)} = ((\widehat{\bOmega}_\lambda^{(k)})^{-1})_{i,obs} = 
\Big(
((\widehat{\bOmega}_\lambda^{(k)})^{-1})_{k\ell}, k,\ell \in \{k:\delta_{ik}=1\}
\Big)$ and $X_{i,obs} = \big(X_{ik}, k \in \{k:\delta_{ik}=1\} \big)^{\rm T}$.
Let $\hat{\lambda}_{CV}$ the optimal parameter that minimizes $CV_{mis}$ and $\widehat{\bOmega}_{CV}$ the graphical lasso solution using all observed data at $\hat{\lambda}_{CV}$.

Figure \ref{fig:realdata} presents three different measures to assess precision matrix estimation. An error distance between the truth and an estimate is evaluated by the spectral norm.
Due to readability, the boxplots of the distance for dense models (``D'' and ``CV'') under the missing proportion $30\%$ are not shown, but their summary statistics are provided in Table \ref{tab:realdata}.
It is confirmed again that having more missing values yields worse estimation. Also, it is possible to see that the denser model is more difficult to achieve satisfactory accuracy in estimation and graph recovery.
\begin{figure}[H]	
	\centering
	%		\begin{minipage}{0.45\linewidth}
	\includegraphics[width=1\linewidth]{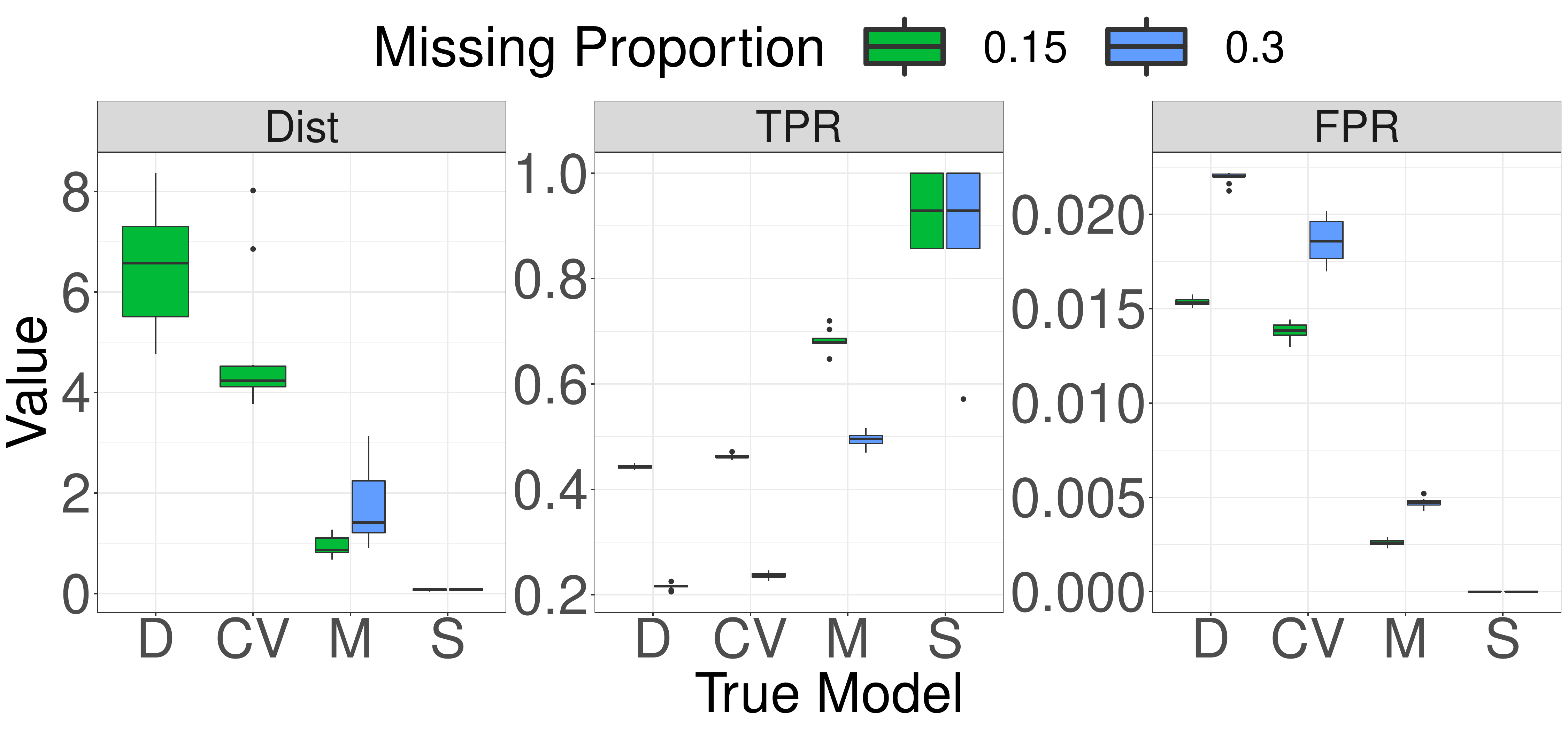}
	%		\end{minipage}
	\caption{Boxplot of performance measures (left: the error distance, middle: TPR, right: FPR) using the riboflavin data. ``D'', ``M'', ``S'', and ``CV'' on the x-axis stand for the dense ($\lambda_1$), moderate ($\lambda_2$), sparse ($\lambda_3$), and cross-validated ($\lambda_{CV}$) models, respectively. Due to readability, two boxplots for the distance from ``D'' and ``CV'' are not shown when the missing proportion is $30\%$.}
	\label{fig:realdata}
\end{figure}
\begin{table}[H]
	\centering
	\begin{tabular}{c|c|c|c|c|c}
		\hline
		&	min & Q1 & Q2 & Q3 & max\\
		\hline
		D & 62.135 & 771.178 & 4340.741 & 8749.103 & 16449.95\\
		\hline
		CV & 26.656 & 30.359 & 53.212 & 3939.772 & 34043.44\\
		\hline
	\end{tabular}
	\caption{Quantiles for the spectral norms of the dense (``D'') and cross-validated (``CV'') models with the missing proportion $30\%$.}
	\label{tab:realdata}
\end{table}

\end{document}